\documentclass{article}
\usepackage{authblk}
\usepackage{amsmath,amssymb,amsfonts}
\usepackage{graphicx, hyperref} 
\usepackage{amsthm}
\usepackage{mathrsfs}
\graphicspath{{./Figs/}}
\usepackage{xcolor}
\usepackage{textcomp}				
\usepackage{comment}
\usepackage[round]{natbib}
\usepackage{multirow}
\usepackage{booktabs}				 

\usepackage[normalem]{ulem}		

\usepackage[title]{appendix}

\newcommand{\mB}{\mathcal{B}}

\newcommand{\mD}{\mathcal{D}}
\newcommand{\mS}{\mathcal{S}}

\newtheorem{Proposition}{Proposition}[section]
\newtheorem{Lemma}{Lemma}[section]
\newtheorem{remark}{Remark} [section]
\usepackage[normalem]{ulem}			
\date{}
\title{Innovative Extensions to Option Pricing: Asymmetric Brownian Motion and Random Walk Approaches}
 \author[1 *]{Jagdish Gnawali}
   \author[2]{Abootaleb Shirvani}
   \author[1] {Dilmi C.W. Hettiachchi-Halpe-Kankanamalage}
  \author [1] {W. Brent Lindquist}
  \author[1]{Svetlozar T. Rachev}
  \author[3]{Frank J. Fabozzi}
\affil[1]{\small Department of Mathematics \& Statistics, Texas Tech University, Lubbock, TX 79409-1042, USA}
\affil[2]{\small Department of Mathematical Sciences, Kean University, Union, NJ 07083, USA}
\affil[3]{\small Carey Business School, Johns Hopkins University, Baltimore, MD 21202, USA}

\affil[*]{Corresponding author:  jgnawali@ttu.edu}

\begin{document}

\maketitle

{\bf Abstract}

Classical option pricing models, such as Bachelier and Black--Scholes--Merton, postulate symmetric Brownian diffusion, which limits their capacity to reflect empirical phenomena including return skewness, heavy tails, and volatility asymmetry. This paper develops an innovative extension: the Geometric Asymmetric Brownian Motion (GABM), unifying asymmetric Brownian motion and random walk methodologies within the Bachelier--Black--Scholes--Merton framework. The approach harnesses the Cherny--Shiryaev--Yor invariance principle (CSYIP) to define asymmetric random walk integrals, where local time at the origin generates skewness and state-dependent risk. Closed-form option pricing formulas are derived, and a discrete-time binomial tree algorithm is constructed and shown to converge rigorously to the GABM limit. By incorporating a smoothed functional form based on the normal inverse Gaussian distribution, the model allows for flexible, state-dependent volatility calibration. Numerical experiments demonstrate the resulting option price and implied volatility surfaces, highlighting the framework's enhanced ability to capture persistent market asymmetry and complex risk behaviors observed in empirical data.\\


\section{Introduction} \label{sec:intro}

Classical option-pricing rests on elegant yet symmetric diffusions--Bachelier’s arithmetic Brownian motion and Black--Scholes--Merton’s geometric Brownian motion \citep{Bachelier_1900, Black_1973, Merton_1973}. Their tractability and no-arbitrage discipline helped found modern derivatives theory, but symmetry in increments and Gaussian tails leaves little room for the robust empirical regularities of financial markets: skewed returns, heavy tails, leverage-type asymmetry, and pronounced volatility smiles and smirks \citep{Cont_2001, Jackwerth_2000}. Behavioral frictions, feedback trading, and information-diffusion delays further generate directional asymmetries that standard Brownian diffusions cannot express parsimoniously \citep{Black_1986, Barberis_1998, Shiller_1981}.

We propose a principled extension of the Bachelier--Black--Scholes--Merton (BBSM) paradigm that keeps the arbitrage-free core while embedding asymmetry directly into the driving process. Building on Donsker-type invariance and the Cherny--Shiryaev--Yor (CSY) invariance principle \cite{Cherny_2003}, we construct asymmetric random-walk integrals whose continuum limit yields a Geometric Asymmetric Brownian Motion (GABM). Local time at the origin induces state-dependent drift and variance components that tilt the distribution endogenously, unifying asymmetric random-walk and Brownian methodologies inside the BBSM framework. We derive closed-form prices for vanilla options under GABM and prove convergence of a binomial scheme to the continuous-time limit. To match option data flexibly, we smooth state-varying risk with a normal-inverse-Gaussian (NIG) kernel, obtaining tractable implied-volatility surfaces with calibrated skew and term-structure dynamics \citep{Eberlein_1995, Prause_1999}.

Our approach complements and differs from well-known alternatives. Stochastic-volatility and Lévy jump-diffusion models generate smiles via variance or jump risk \citep{Heston_1993, Carr_1999}, while temperate/hyperbolic Lévy families fit tails through non-Gaussian innovations \citep{Eberlein_1995, Rosinski_2007}. By contrast, we preserve a Brownian backbone and introduce asymmetry through pathwise features--random-walk bias and local time-yielding clear economic interpretation, diffusion-limit convergence, and simple hedging under risk-neutral valuation. Empirically, the resulting option prices and implied surfaces match persistent left-skew and state-dependent steepness observed in equity options without sacrificing analytical transparency \citep{Bakshi_2003, Christoffersen_2013}.

The remainder proceeds as follows. Section~\ref{sec:2} formalizes the CSY invariance and defines asymmetric random-walk integrals. Section~\ref{sec:3} introduces continuous-time GABM and closed-form prices. Sections~\ref{sec:4} \& \ref{sec:5} extend to correlated asymmetric motions and to settings without a money-market account. Section~\ref{sec:6} develops a provably convergent binomial algorithm and documents accuracy. Section~\ref{subsec:empirical} calibrates equity options and extracts implied-volatility surfaces and risk premia. The paper also implements the theoretical framework through numerical computation. In particular, a recombining binomial lattice consistent with the invariance limit is constructed, enabling stable pricing across large option panels. This algorithm provides fast and precise valuations while preserving convergence to the continuous-time GABM dynamics. Empirical calibration demonstrates that the approach not only reproduces implied-volatility skews and smiles but also achieves surface stability superior to symmetric benchmarks.
Section~\ref{sec:conclusion} concludes with implications for dynamic hedging and avenues for subordinated Lévy enrichments consistent with our invariance foundation \citep{Hu_2022,Hu_2024}.

\section{Asymmetric Random Walk and Asymmetric Brownian Motion} \label{sec:2}

In this section, we introduce (i) an \emph{asymmetric random walk (ARW)} driven by a piecewise-continuous transformation of the state and (ii) its continuous-time limit, an \emph{asymmetric Brownian motion (ABM)}. Our construction relies on the Cherny--Shiryaev--Yor invariance principle (CSYIP), which extends Donsker's theorem to stochastic integrals of random walks \footnote{See: \cite{Cherny_2003}.}.

\subsection{Discrete-time construction}
 Let $\{\xi_k\}_{k\geq 1}$ be independent and identically distributed (i.i.d.) "random signs" with $p(\xi_k = 1)= p(\xi_k=-1)= \frac{1}{2}$. 
For $n \in \mathbb N,$ define scaled increments $\xi_k^{(n)}:= \frac{1}{\sqrt{n}}\xi_k, \quad  k \in \mathbb N_0 = \mathbb N \cup \{0\}$ and fix a piecewise-continuous function (PCF)  $h: \mathbb R\rightarrow \mathbb R$. Set  
\begin{subequations} \label{eq:1ab}    
          \begin{align}
              X_{\frac{k}{n}}^{(n)}:& = \sum_{j=1}^k \xi_j^{(n)},\\
              Y_{\frac{k}{n}}^{(n)}:& = \sum_{j=1}^k h\left ( X_{\frac{j-1}{n}}^{(n)}\right) \xi_j^{(n)}, \quad X_0^{(n)}= Y_0^{(n)}=0
         \end{align}
          \end{subequations} 
 and extend both processes to continuous time by linear interpolation on each interval \( t \in \left[ \frac{k}{n}, \frac{k+1}{n} \right] \). Specifically, for \( t \in \left[ \frac{k}{n}, \frac{k+1}{n} \right] \) let
$\lambda: = nt -k \in [0, 1], \quad  X_t^{(n)}: = (1 - \lambda) X_{\frac{K}{n}}^{(n)} + \lambda X_{\frac{(k+1)}{n}}^{(n)},$ and analogously for \( Y_t^{(n)} \).

\begin{Proposition} \label{prop:2.1} (CSYIP for state-dependent integrands)

 If $h:\mathbb R \rightarrow \mathbb R$ is a PCF, then as $n \uparrow \infty$, the bi-variate process $\left ( X_t^{(n)}, Y_t^{(n)}\right), t \geq 0$ converges in law to $\left ( B_t, \int_0^t h(B_s) dB_s \right),$ where $B_t$ is a standard Brownian motion.
\end{Proposition}

\begin{remark}
    The PCF requirement ensures that the stochastic integral in the limit is well defined (See: Section 2 in \cite{Cherny_2003}).
\end{remark}

\subsection{Local time at the origin and Tanaka's formula}

Let $B=\{B_t\}_{t\geq 0}$ be a standard Brownian motion and let $L=\{L_t\}_{t\geq 0}$ its local time at zero. By definition, 

\begin{equation} \label{eq:2_localTime}
    L_t = \lim_{\epsilon \downarrow 0} \frac{1}{2\epsilon} \int_0^t \mathbf{1}_{\{|B_s|\leq \epsilon\}}\,ds
 \quad t\geq 0
\end{equation}
and Tanaka's formula yields

\begin{equation}  \label{eq:3}
    |B_t| = \int_0^t \text{sgn}(B_s) \, dB_s + L_t, \quad t \geq 0,
\end{equation}

where the sign function \(\text{sgn}(\cdot)\) is defined as\footnote{See: \cite{Chung_1990}}
\begin{equation}  \label{eq:4}
	\text{sgn}(a) = \left\{
	\begin{aligned}
		-1, & \quad \text{if } a < 0, \\
		 0, & \quad \text{if } a = 0, \\
		 1, & \quad \text{if } a > 0.
	\end{aligned}
	\right.
\end{equation}

\subsection{Discrete approximation of local time}

Define the unscaled simple random walk $X_k: = \sum_{j=1}^k \xi_j$ for $k \in \mathbb{N}_0$, $X_0 = 0$, and the discrete local time at zero $L_k: = \sum_{j=0}^k \mathbf{1}_{\{X_j = 0\}}$.
Consider the scaled pair $\left( X_{\frac{k}{n}}^{(n)}, L_{\frac{k}{n}}^{(n)} \right): 
= \left( \frac{1}{\sqrt{n}} X_k, \frac{1}{\sqrt{n}} L_k \right),$ and extend by linear interpolation to continuous time as above. 

\begin{Proposition} \label{prop:2.2} 
(CSYIP with local time )\footnote{See: Theorem 3.1 of \cite{Cherny_2003}.} 

As $n \to \infty$,  $\left( X_t^{(n)}, L_t^{(n)} \right)_{t\geq 0}$
converges in distribution to $\left( B_t, L_t \right)_{t\geq 0}$ in $D([0,\infty), \mathbb{R}^2)$. 

\end{Proposition}

\begin{remark}
    The scaling $\frac{1}{\sqrt{n}}$ for $ L_k$ is optimal: discrete occupation counts grow at order $n^{\frac{1}{2}}$ and converge jointly with the walk. 
\end{remark}

 \begin{figure}[htbp]
	\centering
     \includegraphics[width=0.48\linewidth] {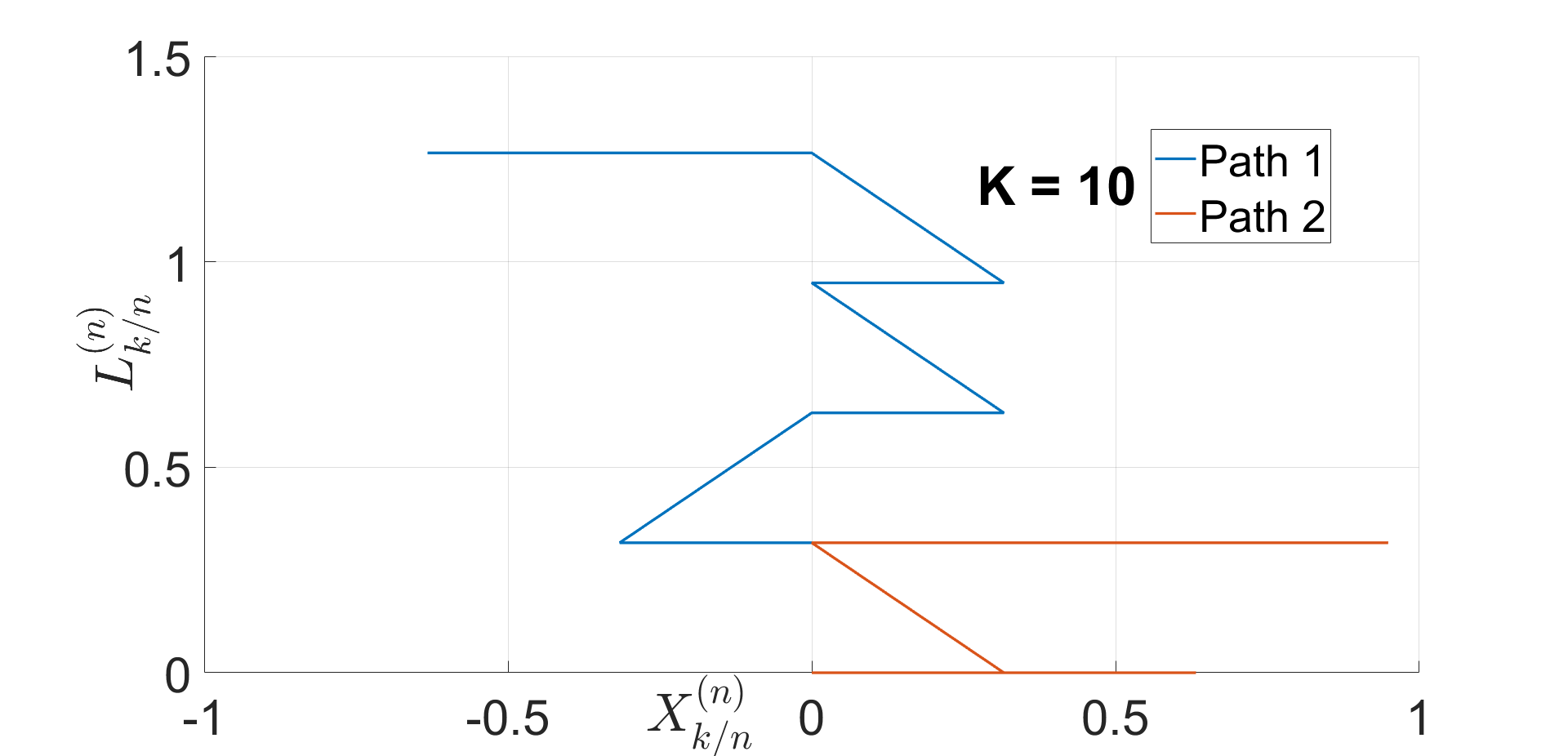}%
      \includegraphics[width=0.48\linewidth] {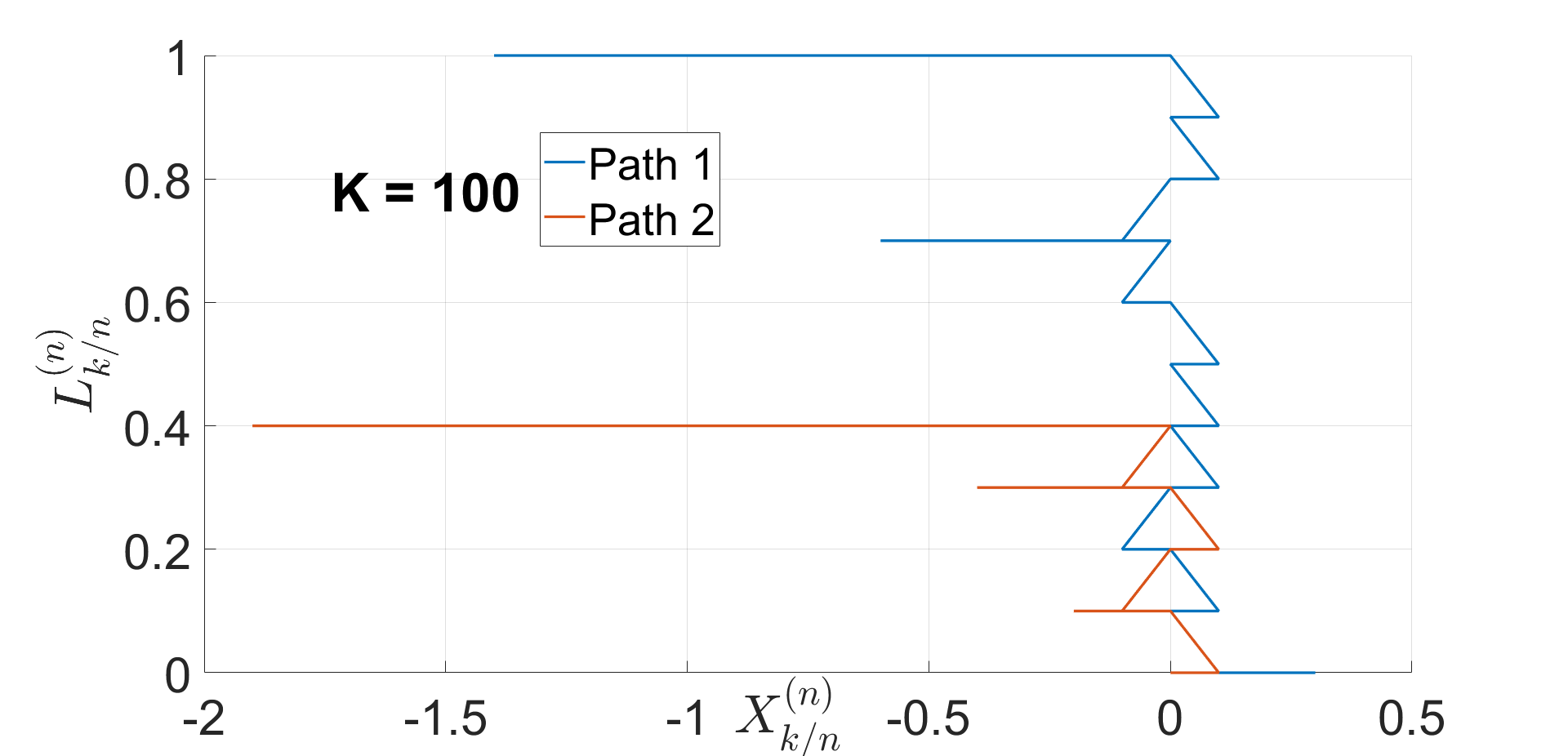}
      
      \includegraphics[width=0.88\linewidth] {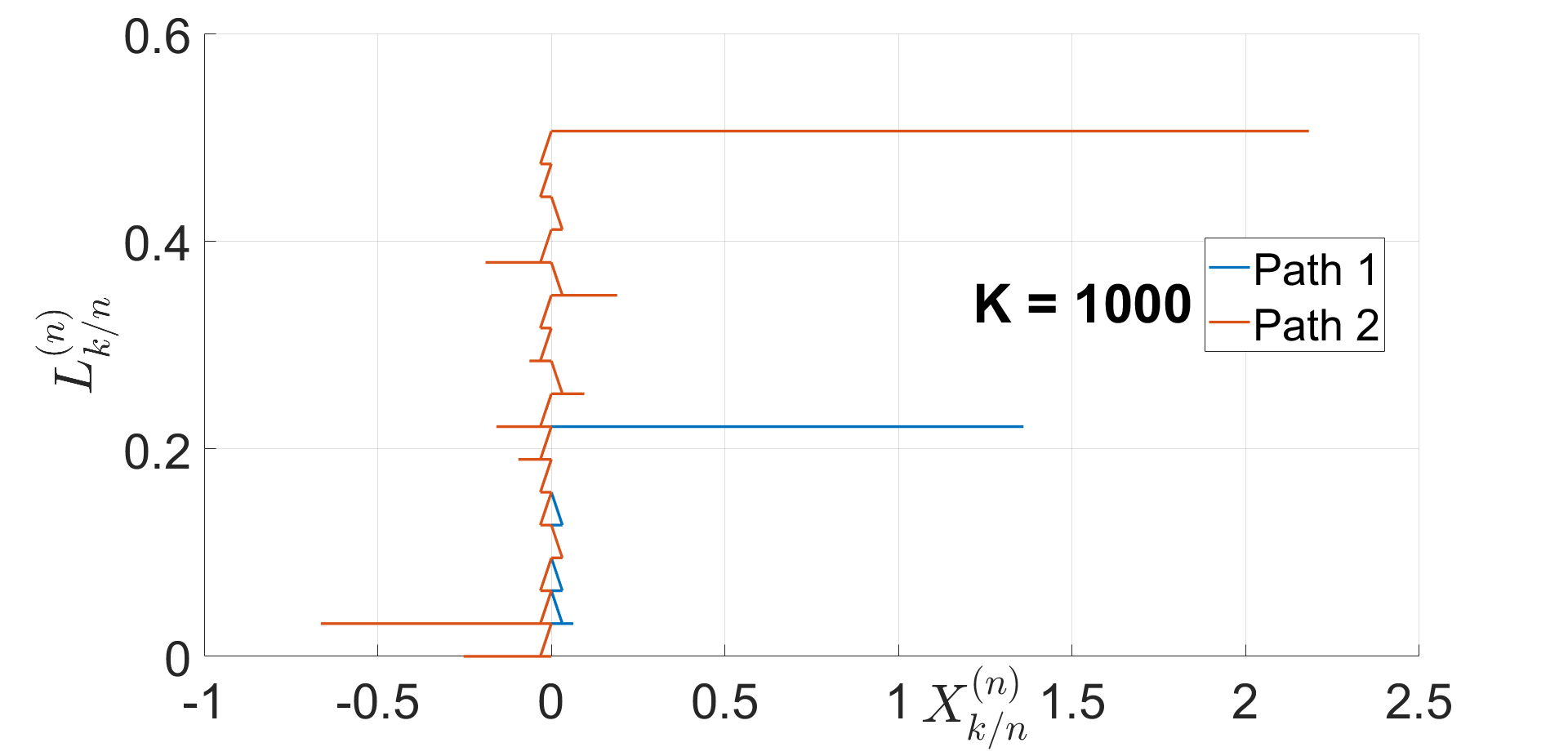}
    \caption{Sample paths of the normalized local time $L^{(n)}_{\frac{k}{n}}$ versus the corresponding positions $X^{(n)}_{\frac{k}{n}}$ for random walks, with varying time steps $K$.}

  	 \label{fig: (X,L)}
\end{figure}

Figure~\ref{fig: (X,L)} provides a detailed comparison of the normalized local time $L^{(n)}_{k/n}$ as a function of the normalized position $X^{(n)}_{k/n}$ for two independent random walk trajectories, with the panels arranged by step counts: $K=10$ (top left), $K=100$ (top right), and $K=1000$ (bottom). For $K=10$, the strongly discretized regime, the local time accumulates in coarse, piecewise-constant increments, and paths quickly plateau, reflecting the minimal recurrence possible at low resolution. At $K=100$, the structure becomes more intricate, with local time building through finer steps and pathwise differences emerging as the walks traverse the origin with varying frequency and dwelling times. In the high-resolution case $K=1000$ (bottom), the trajectories exhibit sharply concentrated local peaks at the origin, corresponding to extended visits and frequent returns; one path demonstrates rapid escalation in local time while the other builds more gradually, highlighting the effect of independent fluctuation sequences.

This progression with increasing $K$ visually demonstrates how the scaled local time converges toward its continuous counterpart: blocky increments gradually resolve into a dense cluster, tightly concentrated around $X=0$. The plots underscore the power of the chosen normalization, making the occupation density and recurrence profile of each random walk directly comparable to the theoretical Brownian local time. These visualizations robustly illustrate both the almost-sure empirical regularity and the random pathwise variability that underlie the joint convergence result of Proposition~\ref{prop:2.2}.

\subsection{The horizontal--vertical CSY random walk and asymmetric Brownian motion} \label{sec:2b}

In this section, we formulate CSYIP for the ``horizontal--vertical'' random walk as detailed in Section 5 of~\cite{Cherny_2003}. This construction generalizes the discrete-time random walk framework into a two-dimensional interacting process with asymmetric properties.

\subsubsection{Discrete-time definition of the horizontal--vertical random walk}
Let $\{\xi_k\}_{k \in \mathbb N_0}$ be i.i.d.\ random signs satisfying $\mathbb{P}(\xi_k = 1) = \mathbb{P}(\xi_k = -1) = \tfrac{1}{2}.$
Define the pair of discrete-time processes $\left( X_k^{(\mathrm{CSY})}, Y_k^{(\mathrm{CSY})} \right)_{k \geq 0}$ starting at zero: $X_0^{(\mathrm{CSY})} = Y_0^{(\mathrm{CSY})} = 0,$
with updates given by
\begin{subequations}
\begin{align}
X_{k+1}^{(\mathrm{CSY})} &= 
\begin{cases}
X_k^{(\mathrm{CSY})} + \xi_{k+1}, & \text{if } Y_k^{(\mathrm{CSY})} > X_k^{(\mathrm{CSY})}, \\
X_k^{(\mathrm{CSY})}, & \text{if } Y_k^{(\mathrm{CSY})} \leq X_k^{(\mathrm{CSY})},
\end{cases} \label{eq:5a} \\
Y_{k+1}^{(\mathrm{CSY})} &= 
\begin{cases}
Y_k^{(\mathrm{CSY})}, & \text{if } Y_k^{(\mathrm{CSY})} > X_k^{(\mathrm{CSY})}, \\
Y_k^{(\mathrm{CSY})} - \xi_{k+1}, & \text{if } Y_k^{(\mathrm{CSY})} \leq X_k^{(\mathrm{CSY})}.
\end{cases} \label{eq:5b}
\end{align}
\end{subequations}

This construction introduces an intricate coupling: the increment in one coordinate depends on the relative position of the two processes, encoding an inherent asymmetry and state-dependent dynamics.

\subsubsection{Continuous-time scaling and interpolation}
For $n \in \mathbb{N}$, define the scaled processes
\begin{equation} \label{eq:6}
X_{\frac{k}{n}}^{(\mathrm{CSY};n)} := \frac{1}{\sqrt{n}} X_k^{(\mathrm{CSY})}, \qquad
Y_{\frac{k}{n}}^{(\mathrm{CSY};n)} := \frac{1}{\sqrt{n}} Y_k^{(\mathrm{CSY})}.
\end{equation}
Extend these by linear interpolation to continuous time $t \geq 0$, producing the piecewise linear trajectories
$\mathbb{X}_t^{(\mathrm{CSY};n)}, \quad \mathbb{Y}_t^{(\mathrm{CSY};n)}.$

\begin{figure}[htbp] 
    \centering
    \includegraphics[width=0.49\textwidth]{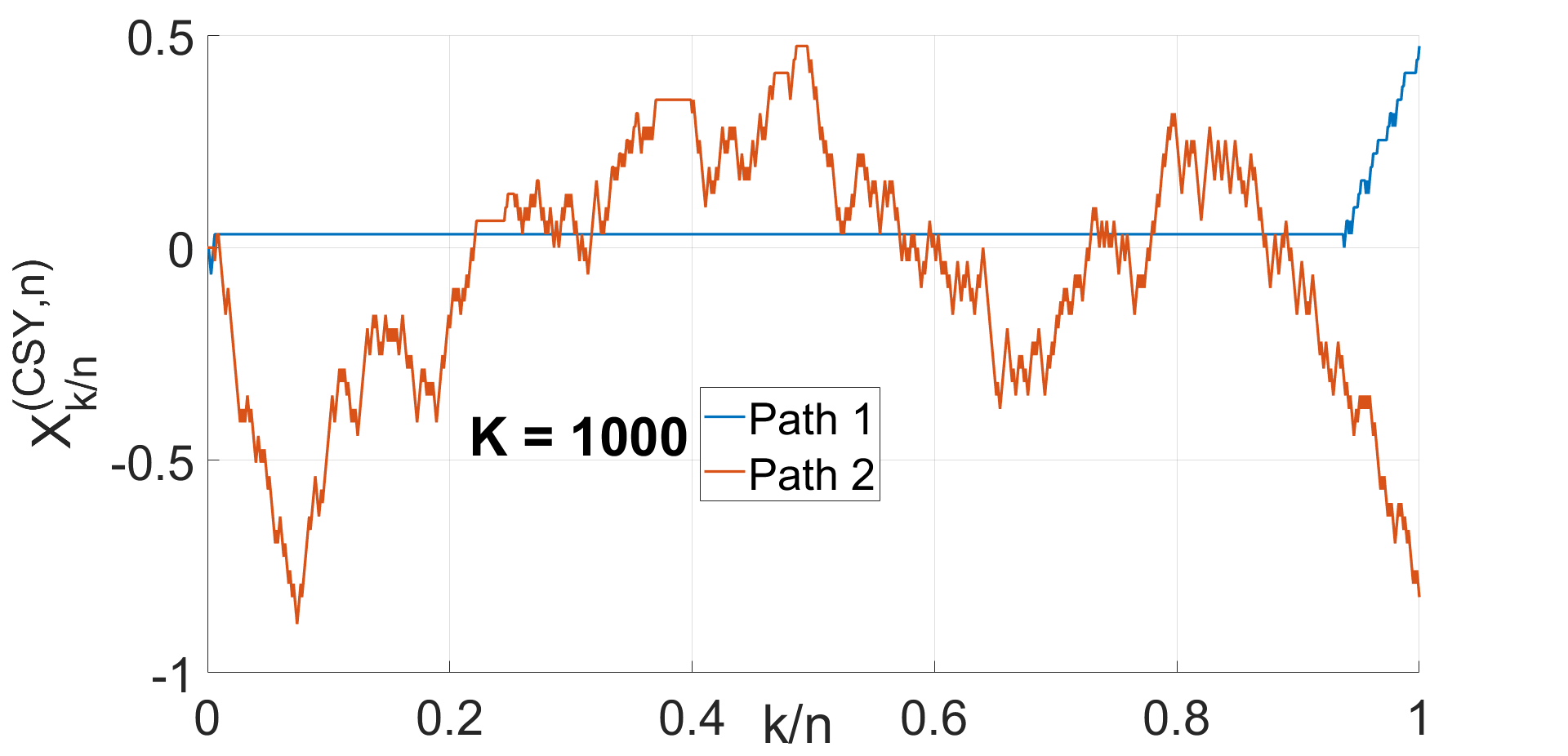}
   \includegraphics[width=0.49\linewidth]{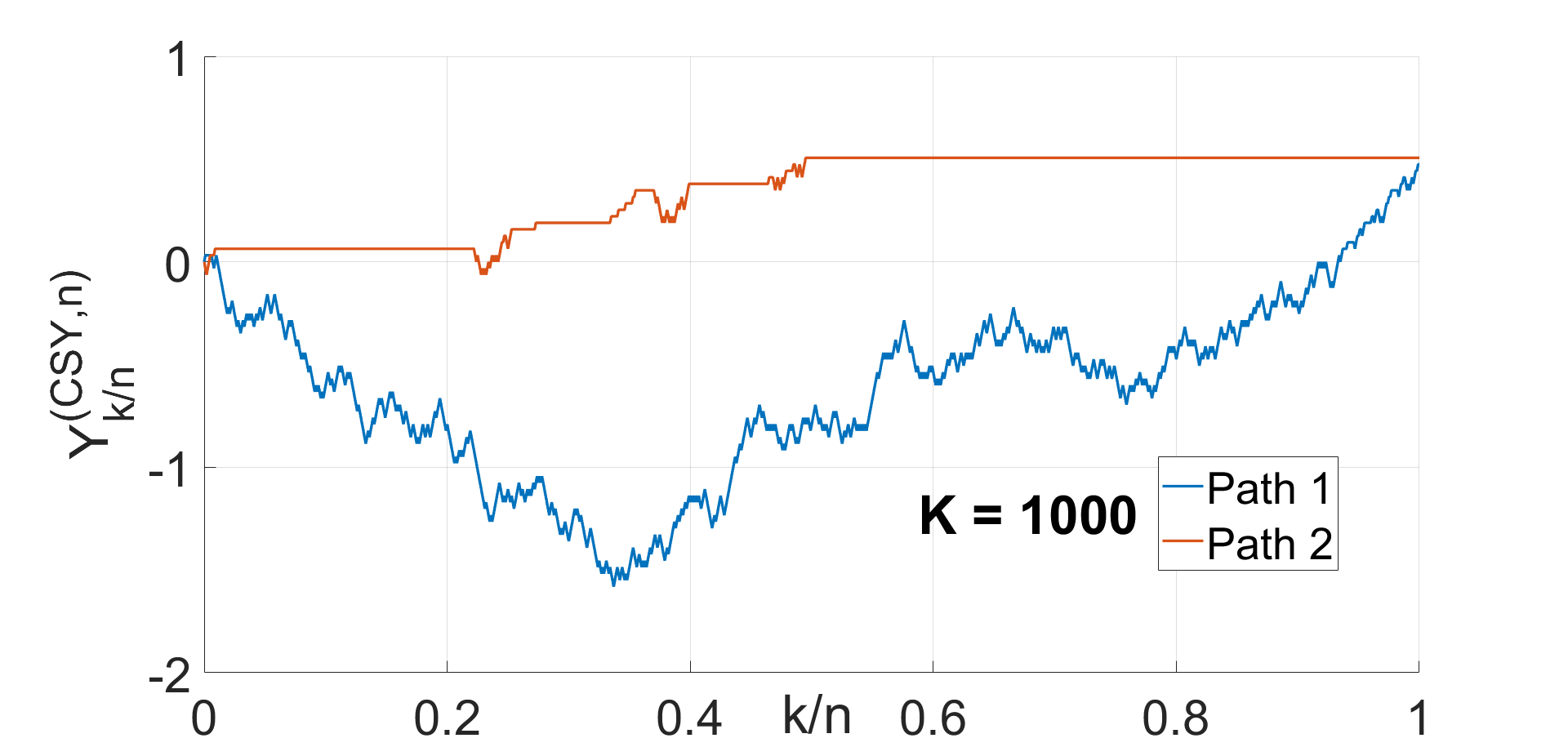}
    \includegraphics[width=0.90\linewidth]{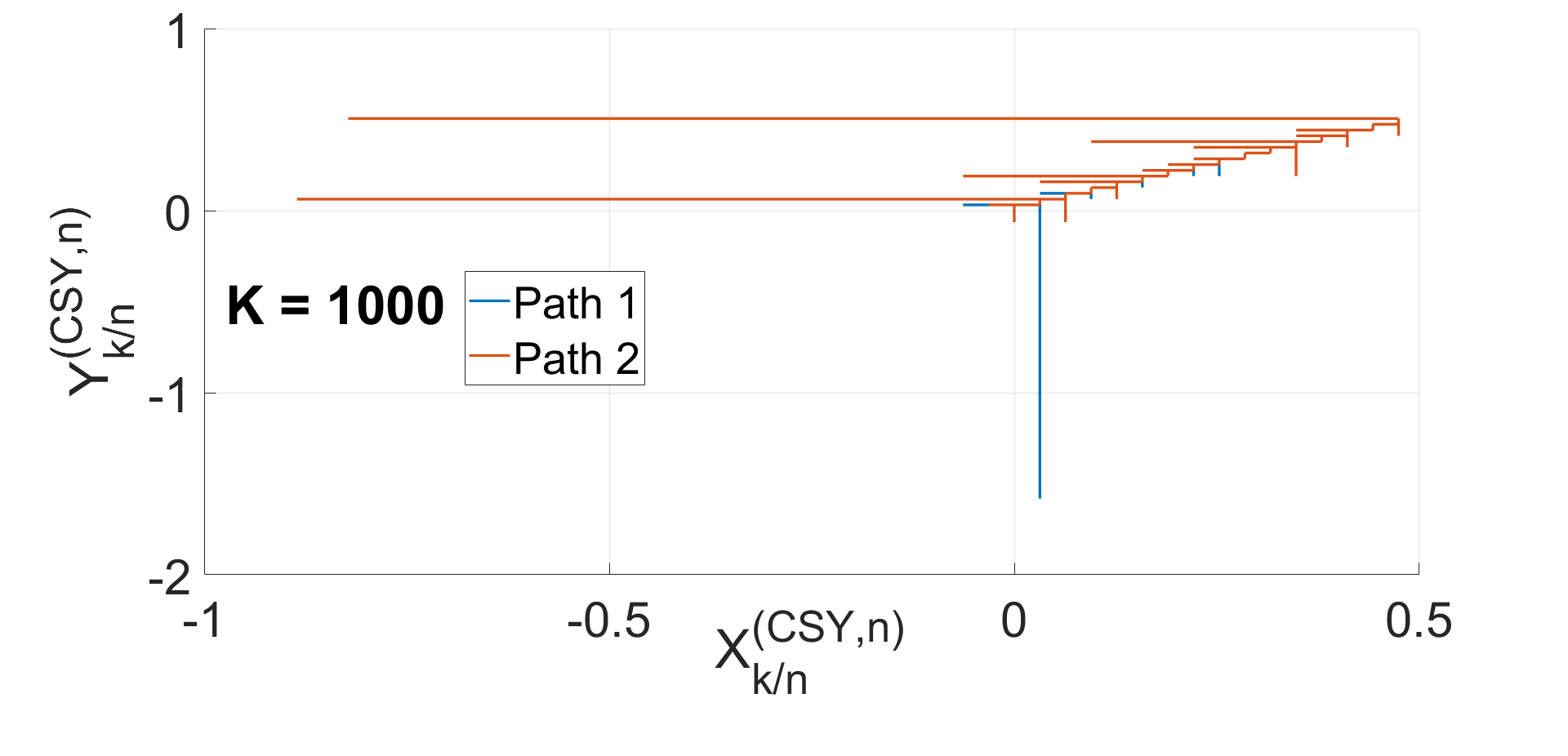}

\caption{Sample paths of the coupled processes $X_t^{(\mathrm{CSY};n)}$ and $Y_t^{(\mathrm{CSY};n)}$ for $K=1000$, showing (left) $X$ versus time, (right) $Y$ versus time, and (bottom) their joint evolution.The state-dependent dynamics yield piecewise constant segments and abrupt transitions characteristic of the horizontal--vertical random walk.}
    \label{fig:(X,Y)}    
\end{figure}
\begin{remark}
The scaling by \( \frac{1}{\sqrt{n}}\) aligns with the classical diffusive scaling in invariance principles, ensuring non-trivial limits and maintaining connection to Brownian motion.
\end{remark}

Figure~\ref{fig:(X,Y)} displays three panels representing two independent sample paths of the coupled, scaled CSY process: the top left panel shows $X_{k/n}^{(\mathrm{CSY};n)}$ versus time, the top right shows $Y_{k/n}^{(\mathrm{CSY};n)}$ versus time, and the bottom panel presents their joint trajectory in the $(X, Y)$ plane. In the top left, Path~1 exhibits a long period of near-constancy followed by a sharp increase, while Path~2 fluctuates more frequently, reflecting differences in when the process becomes dynamically “locked” or released by the coupling mechanism. In the top right, a striking contrast arises: Path~1 for $Y$ meanders with substantial negative excursions and gradual recovery, whereas Path~2 remains almost fixed, punctuated only by rare upward jumps—demonstrating how the state-dependent updates can freeze one coordinate for extensive intervals while the other evolves. The bottom panel synthesizes these behaviors, with each path alternating between horizontal and vertical segments unique to its random history—the sharply defined “staircase” structure particularly clear for Path~2, which spends extended periods on horizontal plateaus before rapid vertical shifts. These findings highlight not only the diversity afforded by the process’s probabilistic construction, but also how state and path dependence combine to yield highly asymmetric, non-generic behavior that is essential to understanding the model’s scaling limits.

\subsubsection{Limit theorem: horizontal--vertical CSY process}
\begin{Proposition}[CSYIP for the horizontal--vertical random walk~\cite{Cherny_2003}] \label{prop:2.3} \ \\
Let \(B = \{B_t\}_{t \ge 0}\) be a standard Brownian motion with local time \(L = \{L_t\}_{t \ge 0}\) at zero. Define the positive and negative parts:
$B_t^{\min} := \min(B_t, 0), \qquad B_t^{\max} := \max(B_t, 0).$
Then the limiting horizontal--vertical processes are
\begin{subequations}
\begin{align}
    \mathbb{X}_t^{(\mathrm{CSY})} &:= \int_0^t \mathbf{1}_{(-\infty,0]}(B_s) dB_s = \tfrac{1}{2}L_t + B_t^{\min}, \label{eq: 7a}     \\           
    \mathbb{Y}_t^{(\mathrm{CSY})} &:= -\int_0^t \mathbf{1}_{(0,\infty)}(B_s) dB_s = \tfrac{1}{2}L_t - B_t^{\max}. \label{eq: 7b}                  
\end{align}
\end{subequations}
The continuous-time scaled processes \(\left(\mathbb{X}_t^{(\mathrm{CSY};n)}, \mathbb{Y}_t^{(\mathrm{CSY};n)}\right)\) converge in law to \(\left(\mathbb{X}_t^{(\mathrm{CSY})}, \mathbb{Y}_t^{(\mathrm{CSY})}\right)\) as \(n \to \infty\).
\end{Proposition}

\begin{remark}
This proposition extends the classical Donsker invariance principle by incorporating a state-dependent decomposition of Brownian motion through its excursions and local time, revealing the asymmetric structure inherent in the CSY horizontal-vertical walk.
\end{remark}

Applying Tanaka's formula, the processes admit the representations
\begin{subequations}
\begin{align}
\mathbb{X}_t^{(\mathrm{CSY})} &= \tfrac{1}{2} \left( B_t - \int_0^t \mathrm{sgn}(B_s) dB_s \right), \\
\mathbb{Y}_t^{(\mathrm{CSY})} &= \tfrac{1}{2} \left( -B_t - \int_0^t \mathrm{sgn}(B_s) dB_s \right),
\end{align}
\end{subequations}
where \(\mathrm{sgn}(x)\) is defined as usual.

\subsubsection{Parametrization of asymmetric perturbations}
For any vector \(\beta := (\beta^{(-)}, \beta^{(L)}, \beta^{(+)}) \in \mathbb{R}^3\), define the asymmetrically perturbed process
\begin{equation} \label{eq:9}       
    \mathbb{Z}_t^{(\beta)} := B_t + \beta^{(-)} \mathbb{X}_t^{(\mathrm{CSY})} + \beta^{(L)} L_t + \beta^{(+)} \mathbb{Y}_t^{(\mathrm{CSY})}, \quad t \ge 0.
\end{equation}

Expressing in terms of extremal functionals and local time:
\[
\mathbb{Z}_t^{(\beta)} = B_t + \beta^{(-)} B_t^{\min} + \left( \tfrac{1}{2} \beta^{(-)} + \beta^{(L)} + \tfrac{1}{2} \beta^{(+)} \right) L_t - \beta^{(+)} B_t^{\max}.
\]

Reparameterizing as
\[
\gamma^{(-)} := \beta^{(-)}, \quad \gamma^{(+)} := -\beta^{(+)}, \quad \gamma^{(L)} := \tfrac{1}{2} \beta^{(-)} + \beta^{(L)} + \tfrac{1}{2} \beta^{(+)}
\]
yields a compact form:
\begin{equation} \label{eq:10}      
    \mathbb{Z}_t^{(\beta, \gamma)} := B_t + \gamma^{(-)} B_t^{\min} + \gamma^{(L)} L_t + \gamma^{(+)} B_t^{\max}, \quad t \geq 0.
\end{equation}

\begin{remark}
The parameters \(\gamma^{(-)}\) and \(\gamma^{(+)}\) regulate skewness on the negative and positive sides of Brownian motion's excursions, respectively, while \(\gamma^{(L)}\) modulates the influence of local time, creating a ``sticky'' behavior at zero--a key feature for asymmetric modeling in finance.
\end{remark}

\subsubsection{Extensions incorporating correlated Brownian motions}
Introduce the more general dynamics
\[
\mathbb{W}_t^{(\beta, \gamma)} = \gamma^{(B)} B_t^{(1)} + \gamma^{(-)} B_t^{(2,\min)} + \gamma^{(0)} L_t^{(4)} + \gamma^{(+)} B_t^{(3,\max)}, \quad t \geq 0,
\]
where \(B_t^{(i)}\), \(i=1,2,3\), are correlated standard Brownian motions with instantaneous covariance matrix \(\rho_{(i,j)} \in [0,1]\).

Note that since
\[
B_t^{(\max)} - B_t^{(\min)} = |B_t|,
\]
this framework includes the skew Brownian motion (~\cite{Corns_2007}):
\[
B_t^{(\mathrm{skew}, \delta)} = \sqrt{1 - \delta^2} B_t + \delta |W_t|, \quad \delta \in (0,1),
\]
with $B_t$ and $W_t$ independent standard Brownian motions.

\subsection{Alternate parametrization and connection to stochastic integral representation}
The parametrization in Eq.~\eqref{eq:10} admits the alternative form
\begin{equation} \label{eq:11}
    \mathbb{Z}_t^{(\beta, \gamma)} = \delta^{(-)} B_t^{(\min)} + \delta^{(0)} \int_0^t h^{(0)}(B_s) dB_s + \delta^{(+)} B_t^{(\max)}, \quad t \geq 0,
\end{equation}
with parameters
\[
\delta^{(-)} = \gamma^{(-)} - \gamma^{(L)}, \quad \delta^{(0)} = -\gamma^{(L)}, \quad \delta^{(+)} = \gamma^{(+)} + \gamma^{(L)},
\]
and $ h^{(0)} (x) = \text{sgn}(x)$.
Patch together the discrete-time counterpart:
\[
Y_{\frac{k}{n}}^{(h^{(0)},n)} = \sum_{j=1}^k h^{(0)}\left( \frac{X_{j-1}^{(n)}}{n} \right) \xi_j^{(n)},
\]
and we obtain, by Proposition~\ref{prop:2.1}, convergence of the piecewise linear interpolation $\left( X_t^{(n)}, Y_t^{(h^{(0)},n)} \right)$ weakly to $\left( B_t, \int_0^t h^{(0)}(B_s) dB_s \right)$.

Define the transformations:
\begin{equation} \label{eq:12}    
\tilde{\mathbb{X}}_t^{(\mathrm{CSY};n)} = \mathbb{X}_t^{(\mathrm{CSY};n)} - \mathbb{Y}_t^{(\mathrm{CSY};n)}, \quad
\tilde{\mathbb{Y}}_t^{(\mathrm{CSY};n)} = \mathbb{X}_t^{(\mathrm{CSY};n)} + \mathbb{Y}_t^{(\mathrm{CSY};n)}.
\end{equation}
By Proposition~\ref{prop:2.1}, the bivariate process $\left( \tilde{\mathbb{X}}_t^{(\mathrm{CSY};n)}, \tilde{\mathbb{Y}}_t^{(\mathrm{CSY};n)} \right)$ converges weakly to
\[
\left( B_t, - \int_0^t h^{(0)}(B_s) dB_s \right).
\]

Inverting these transformations gives the limit:
\[
\left( \mathbb{X}_t^{(\mathrm{CSY};n)}, \mathbb{Y}_t^{(\mathrm{CSY};n)} \right) \xrightarrow{d} \left( \tfrac{1}{2} \left( B_t - \int_0^t h^{(0)}(B_s) dB_s \right), -\tfrac{1}{2} \left( B_t + \int_0^t h^{(0)}(B_s) dB_s \right) \right).
\]

\subsection{Definition of asymmetric Brownian motion for the complete market case}
Fix parameters leading to an asymmetric Brownian motion, defined via
\begin{equation} \label{eq:13}  
\mathbb{A}_t^{(\alpha)} = B_t + \alpha \int_0^t h^{(0)}(B_s) dB_s, \quad t \geq 0,
\end{equation}
where \(\alpha = -\gamma^{(-)} = -\gamma^{(L)} = \gamma^{(+)} \in \mathbb{R}\).

\begin{remark}
The symmetric choice of the parameters yields a complete market dynamics with one risky asset and a riskless bank account driven by \(\mathbb{A}_t^{(\alpha)}\).
\end{remark}

\begin{figure}[htbp]
	\centering
	\includegraphics[width=0.49\linewidth]{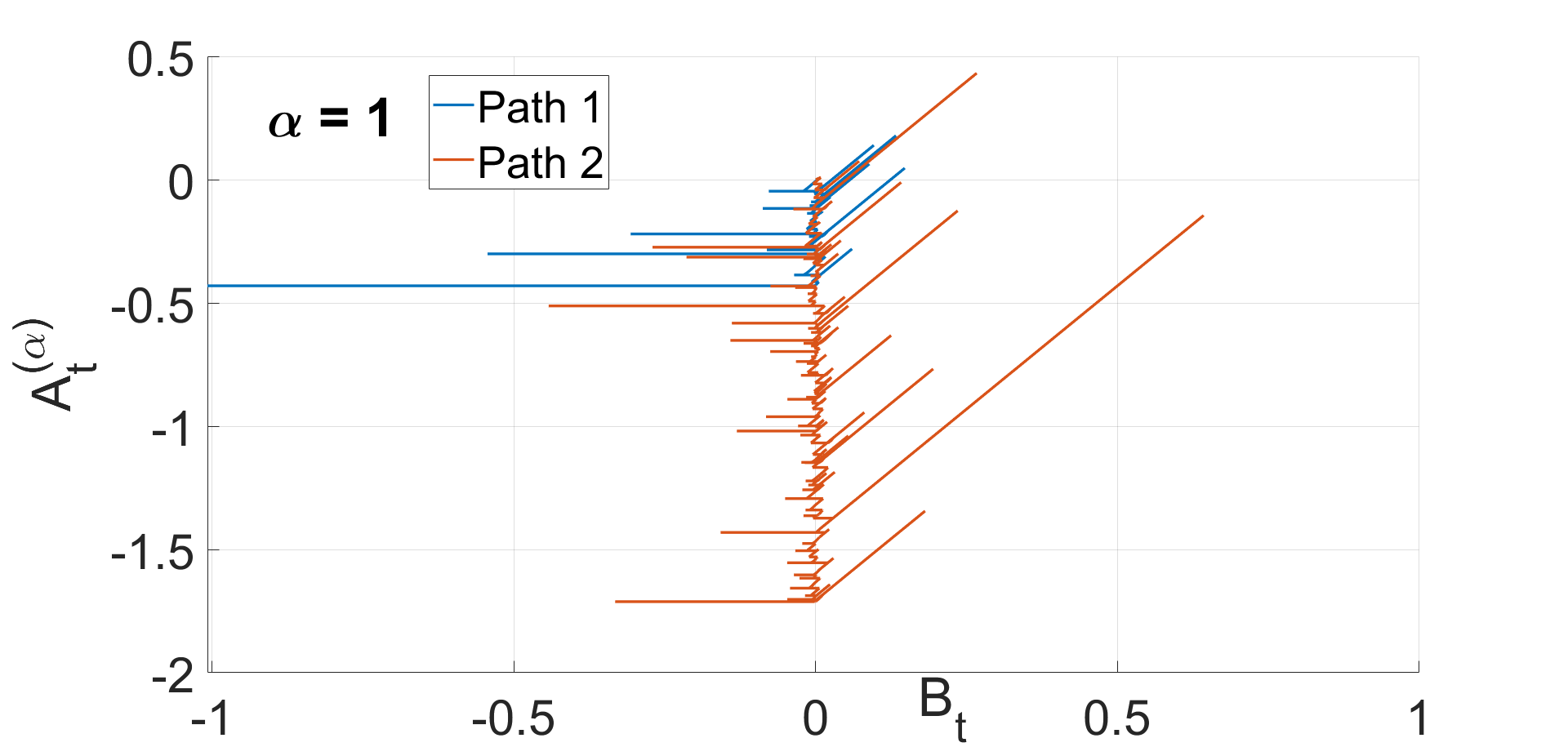}%
     \includegraphics[width=0.49\linewidth]{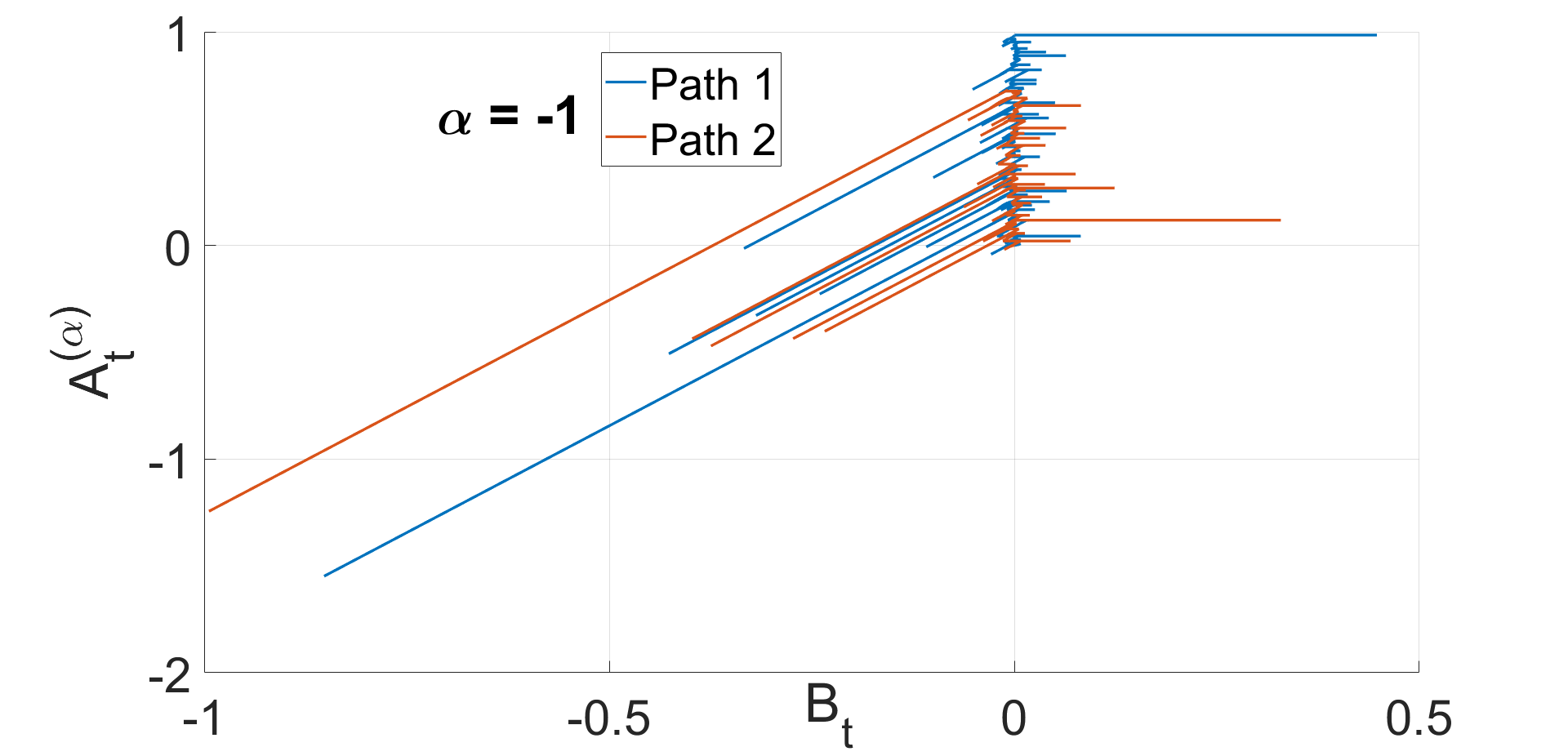}
     
  \caption{Sample trajectories of the bivariate process $\left( B_t, \mathbb{A}_t^{(\alpha)} \right)$ defined in Eq.~\eqref{eq:13}, simulated over the fixed horizon $T=1$,  with $\alpha = 1$ (left) and $\alpha = -1$ (right). Each subplot visualizes two independent realizations, highlighting the effect of the asymmetry parameter $\alpha$ on the joint path structure.}
  	 \label{fig: (t, B_t, A_t^(alpha) )}.
\end{figure}
The plots in Figure~\ref{fig: (t, B_t, A_t^(alpha) )} illustrate sample paths of the bivariate process $(B_t, \mathbb{A}_t^{(\alpha)})$ for $\alpha = 1$ and $\alpha = -1$, simulated over a fixed horizon $T = 1$. This ensures that  
 how the asymmetry parameter $\alpha$ modulates directional skewness through the sign-dependent drift term in Eq.~\eqref{eq:13}, rather than variations in time scaling. For $\alpha = 1$, upward deviations of $B_t$ are amplified, producing trajectories where positive excursions dominate. In contrast, $\alpha = -1$ reverses the bias, magnifying downward movements and generating predominantly negative trends in $\mathbb{A}_t^{(\alpha)}$. The observed range of variation—approximately $B_t \in [-1, 1]$ and $\mathbb{A}_t^{(\alpha)} \in [-2, 0.5]$ for $\alpha=1$, versus $B_t \in [-1, 0.5]$ and $\mathbb{A}_t^{(\alpha)} \in [-2, 1]$ for $\alpha=-1$--quantitatively confirms the asymmetric deformation induced by $\alpha$. The broken-line geometry apparent in both panels results from slope discontinuities occurring at the zero crossings of $B_t$, where $\mathrm{sgn}(B_s)$ flips. These kinks encode the sticky and directionally biased behavior predicted by the extended CSY limit theorem and the Tanaka--type stochastic representation.
 These features provide numerical and visual validation of the state-dependent deformation mechanism governed by $\alpha$, confirming the theoretical consistency of the asymmetric model and highlighting its suitability for capturing skewed or regime-sensitive behaviors in stochastic systems.

\subsection{Extension to incomplete market dynamics}
To model market incompleteness characterized by an additional skew component, define
\[
\mathbb{B}_t^{(\alpha,\delta)} = \sqrt{1 - \delta^2} B_t + \delta |W_t| + \alpha \int_0^t h^{(0)}(B_s) dB_s, \quad \delta \in (0,1),
\]
with $B_t$ and $W_t$ independent Brownian motions.

Hedging under this model typically requires two risky assets plus a riskless account, as discussed in~\cite{Hu_2024}, beyond the scope of the current work focused on complete markets.

\subsection{Flexible piecewise-continuous function modeling}
A more flexible piecewise-continuous function (PCF) for modeling asymmetry is given by
\begin{equation} \label{eq:14}   
    h^{(\Phi)}(x) = \begin{cases}
        \phi^{(-)}, & x < \phi \\
        \phi^{(0)}, & x = \phi \\
        \phi^{(+)}, & x > \phi,
    \end{cases}
\end{equation}
where \(\Phi = (\phi, \phi^{(-)}, \phi^{(0)}, \phi^{(+)}) \in \mathbb{R}^4\).

\begin{remark}
More general PCFs can be constructed as linear combinations of such step-like functions, preserving piecewise continuity; similar limit results for option pricing as developed in Sections~\ref{sec:3} and~\ref{sec:4} hold under these generalizations.
\end{remark}

For corresponding discrete-time increments,
$Y_{\frac{k}{n}}^{(h^{(\Phi)},n)} = \sum_{j=1}^k h^{(\Phi)} \left( X_{\frac{j-1}{n}}^{(n)} \right) \xi_j^{(n)},$
we establish convergence in law:
$\left( \mathbb{X}_t^{(n)}, \mathbb{Y}_t^{(h^{(\Phi)}; n)} \right) \xrightarrow{d} \left( B_t, \int_0^t h^{(\Phi)}(B_s) dB_s \right).$

\subsection{NIG density as a PCF for heavy-tailed asymmetry}
To capture heavy tail and flexible asymmetry, we consider the NIG density as a PCF:

\begin{equation}  \label{eq:15}            
h^{\Phi}(\mu, \alpha, \beta, \delta;x)= \frac{\alpha K_1 \left( \alpha \sqrt{\delta^2 + (x-\mu)^2} \right)}{\pi \sqrt{\delta^2 + (x-\mu)^2}}
\exp\left( \delta \sqrt{\alpha^2 - \beta^2} + \beta (x-\mu) \right),
\end{equation}

where $\mu$ locates the distribution, $\alpha>0$ controls tail heaviness, $\beta \in (-\alpha, \alpha)$ controls skewness, $\delta$ is a scale parameter, and $K_1$ is the modified Bessel function of the second kind.

\begin{remark}
The NIG density PCF provides a principled, flexible tool for calibrating asymmetry and tail behavior in financial return models, capturing effects beyond the simple piecewise step functions.
\end{remark}

\begin{figure}[htbp]
   \centering
   \includegraphics[width=0.90\linewidth]{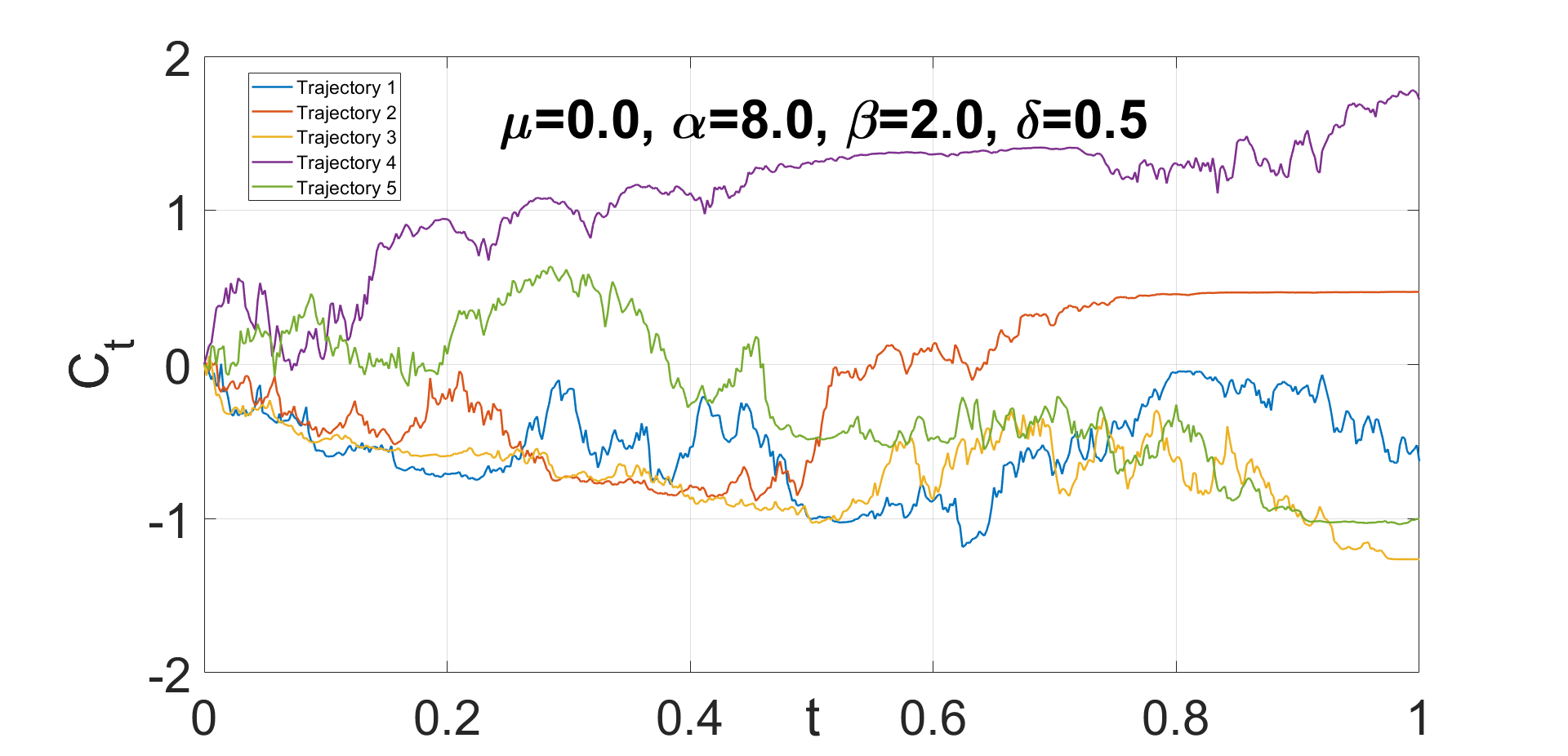}
\caption{Sample trajectories of the stochastic integral process 
$C_t{(\mu, \alpha, \beta, \delta)} = \int_0^t h^{\Phi}(B_s) \, dB_s$ for $t \in [0,1]$, where $h^{(\Phi)}$ is the NIG density with parameters $\mu=0.0$, $\alpha=8.0$, $\beta=2.0$, and $\delta=0.5$, illustrating heavy-tailed, asymmetric path behavior.}
\label{fig: numerical 7_trajectories}
\end{figure}

Figure~\ref{fig: numerical 7_trajectories} displays five sample trajectories of the stochastic integral process $C_t = \int_0^t h^{\Phi}(B_s)\, dB_s$ for $t \in [0,1]$, where $h^{\Phi}$ is the NIG density with parameters $\mu=0.0$, $\alpha=8.0$, $\beta=2.0$, and $\delta=0.5$. The plotted paths reveal pronounced heterogeneity and asymmetric features, stemming from the heavy-tailed and positively skewed nature of the control function $h^{\Phi}$. Specifically, the relatively high value of $\alpha$ amplifies the likelihood of large fluctuations, while the positive $\beta$ imparts a consistent upward skew, resulting in several trajectories that trend sharply positive or reach unusually high values over short intervals. This asymmetric, heavy-tailed modulation leads to trajectories exhibiting abrupt excursions and highly variable magnitudes, as observed in the rapid ascent of trajectory~4 and the more moderate but persistent growth in others.

Distinct pathwise differences are apparent: some trajectories remain centered or oscillatory, while others display sustained upward drifts or abrupt, isolated spikes. The parameter regime chosen effectively magnifies the stochastic variability inherent in Brownian motion, but channels this variability into persistently asymmetric, heavy-tailed outcomes—a property unattainable with Gaussian or symmetric controls. These visual results demonstrate the power of the NIG density as a parametric control, enabling fine tuning of both skewness and kurtosis in the induced stochastic integral, and providing a compelling explanation for its suitability in modeling financial data characterized by sudden shocks and persistent trends.

\section{Option Pricing with Geometric Asymmetric Brownian Motion} \label{sec:3}

In this section, we apply the CSYIP methodology to extend the classical BSM framework. We consider a frictionless market with a risky asset \(\mathcal{S}\) whose price process follows a \textit{Geometric Asymmetric Brownian Motion} (GABM):
\begin{equation} \label{eq:3.1}
    S_t^{(h)} = S_0 \exp\left( \nu t + \sigma B_t + \gamma C_t \right), \quad S_0 > 0, \quad \nu \in \mathbb{R}, \quad \sigma \neq 0, \quad \gamma \in \mathbb{R},
\end{equation}
defined on the filtered probability space \(\left( \Omega, \mathbb{F} = (\mathcal{F}_t = \sigma(B_u; u \le t)), \mathbb{P} \right)\). Here, 
\begin{equation} \label {eq:3.2}
    C_t = \int_0^t h(B_s) \, dB_s,
\end{equation}
where \(h(x) = h^{\Phi}(\mu, \alpha, \beta, \delta; x)\) is a probability density function as introduced in Eq.~\eqref{eq:15}. The riskless asset (bank account) evolves deterministically as
\begin{equation} \label{eq:3.3}
    \beta_t = \beta_0 e^{r t}, \quad \beta_0 > 0,
\end{equation}
with fixed interest rate \(r \geq 0\). Additionally, we consider a European contingent claim (option) \(\mathcal{C}\) with price process 
\begin{equation}\label{eq:3.4}
    f_t = f\left( S_t^{(h)}, t \right), \quad t \in [0,T),
\end{equation}
and terminal payoff \(f_T = g(S_T^{(h)})\), where \(f\) and \(g\) satisfy the regularity conditions summarized in Appendix E of~\cite{duffie_2001}.

\subsection{No-arbitrage and equivalent martingale measure}
The market \((\mathcal{S}, \mathcal{B}, \mathcal{C})\) is arbitrage-free and complete, as there is only a single Brownian driver and the asset price \(S_t^{(h)}\) is constructed as a functional extension of the classical BSM model.

To price derivatives, we seek an equivalent martingale measure \(\mathbb{Q} \sim \mathbb{P}\) such that the discounted price process $Z_t = \frac{S_t^{(h)}}{\beta_t}$
is a \(\mathbb{Q}\)-martingale. This requires identifying the market price of risk \(\theta_t\) such that
\[
B_t^{(\mathbb{Q})} := B_t + \int_0^t \theta_s ds
\]
is a Brownian motion under \(\mathbb{Q}\).

\begin{Lemma} \label{lemma:3.1}
Applying It\^o's lemma to \(S_t^{(h)}\) defined by Eq.~\eqref{eq:3.1} and Eq.~\eqref{eq:3.2} yields:
\begin{equation} \label{eq:3.5}
\begin{aligned}
    dS_t^{(h)} = S_t^{(h)} \Big[
        \nu + \tfrac{1}{2} \sigma^2 + \sigma \gamma h(B_t) + \tfrac{1}{2} \gamma^2 h(B_t)^2
    \Big] dt
    + S_t^{(h)} \Big[
        \sigma + \gamma h(B_t)
    \Big] dB_t.
\end{aligned}
\end{equation}
\end{Lemma}

\noindent\textit{Proof: See Appendix~\ref{appendix:Lemma_3.1}.}
\vspace{0.2cm}

\begin{remark}
The drift adjustment \(\sigma \gamma h(B_t) + \frac{1}{2} \gamma^2 h(B_t)^2\) arises due to the stochastic integral term \(C_t\), reflecting the path-dependent asymmetric perturbation.
\end{remark}

\subsection{Market price of risk and dynamics under \(\mathbb{Q}\)}

To ensure that \(Z_t\) is a \(\mathbb{Q}\)-martingale, the market price of risk is formally given by\footnote{See Appendix:~\ref{appendix:EMM}}
\begin{equation} \label{eq:3.6}
    \theta_t = \frac{\nu - r + \frac{1}{2} \sigma^2 + \sigma \gamma h(B_t) + \frac{1}{2} \gamma^2 h(B_t)^2}{\sigma + \gamma h(B_t)}.
\end{equation}

Under \(\mathbb{Q}\), using Girsanov's theorem, \(S_t^{(h)}\) satisfies
\begin{equation} \label{eq:3.7}
    dS_t^{(h)} = r S_t^{(h)} dt + (\sigma + \gamma h(B_t)) S_t^{(h)} dB_t^{(\mathbb{Q})},
\end{equation}
where \(B_t^{(\mathbb{Q})}\) is a \(\mathbb{Q}\)-Brownian motion.

The unique equivalent martingale measure \(\mathbb{Q}\) allows pricing the contingent claim by
\begin{equation} \label{eq:3.8}
    f_0 = e^{-rT} \mathbb{E}^{\mathbb{Q}}[g(S_T^{(h)})].
\end{equation}

\subsection{Replicating portfolio and hedging}

Consider a self-financing portfolio \((a_t, b_t)\) replicating the option \(\mathcal{C}\):
\begin{equation} \label{eq:3.9}
    P_t = a_t S_t^{(h)} + b_t \beta_t = f_t = f(S_t^{(h)}, t), \quad t \in [0,T).
\end{equation}

\begin{Lemma} \label{emma:3.2} \label{lemma:3.2}
By It\^o's formula, the option price satisfies
\begin{equation} \label{eq:3.10}
\begin{aligned}
    df_t = {}& \frac{\partial f}{\partial t} dt + \frac{\partial f}{\partial x} \left[ \left( \nu + \tfrac{1}{2} \sigma^{2} + \sigma \gamma h(B_t) + \tfrac{1}{2} \gamma^{2} h(B_t)^2 \right) S_t^{(h)} dt + (\sigma + \gamma h(B_t)) S_t^{(h)} dB_t \right] \\
    & + \frac{1}{2} \frac{\partial^2 f}{\partial x^{2}} (\sigma + \gamma h(B_t))^2 (S_t^{(h)})^2 dt.
\end{aligned}
\end{equation}
\end{Lemma}

\noindent\textit{Proof: See Appendix~\ref{appendix:Lemma_3.2}.}
\vspace{0.2cm}

The replicating strategy satisfies
\begin{align}
    a_t &= \frac{\partial f}{\partial x}(S_t^{(h)}, t), \label{eq:3.11} \\
    b_t &= \frac{1}{\beta_t} \left[ f(S_t^{(h)}, t) - a_t S_t^{(h)} \right]. \label{eq:3.12}
\end{align}

\subsection{Modified Feynman--Kac representation}

Setting \(x = S_t^{(h)}\) and \(b(t) = B_t\), the option price solves the parabolic PDE:
\begin{equation} \label{eq:3.13}
    \frac{\partial f}{\partial t} + r x \frac{\partial f}{\partial x} - r f + \frac{1}{2} [\sigma + \gamma h(b(t))]^2 x^2 \frac{\partial^2 f}{\partial x^2} = 0, \quad x > 0, \ t \in [0,T),
\end{equation}
with terminal condition 
\[
f(x,T) = g(x).
\]

The observable market price of risk at time \(t\) is
\begin{equation} \label{eq:3.14}
    \theta_t(b(t)) = \frac{\nu - r + \frac{1}{2} \sigma^2 + \sigma \gamma h(b(t)) + \frac{1}{2} \gamma^2 h(b(t))^2}{\sigma + \gamma h(b(t))}.
\end{equation}

Compared to the classical BSM assumption where only the price \(S_t^{(h)}\) is observed, here the PDE explicitly depends on the Brownian motion state \(B_t\), adding path-dependence and breaking Markovianity in price alone.

By the Feynman--Kac theorem, the unique solution admits the conditional expectation representation
\begin{equation} \label{eq:3.15}
    f(x,t) = \mathbb{E}^{X_t = x}\left[ e^{-r(T-t)} g(X_T) \right],
\end{equation}
where \(X_s\) satisfies the stochastic differential equation
\begin{equation} \label{eq:3.16}
    dX_s = r X_s ds + [\sigma + \gamma h(b(s))] X_s dB_s^{(\mathbb{Q})}, \quad s \in [t,T).
\end{equation}

\subsection{Limitations of discounting and multi-asset considerations}

Valuation via the GABM-PDE presents a conceptual challenge because the stochastic discount factor generated by Eq.~\eqref{eq:3.3} is no longer riskless due to its implicit dependence on \(B_t\).

Consider an extended market with two risky assets \(\mathcal{S}_1\) and \(\mathcal{S}_2\), modeled by
\begin{equation} \label{eq:3.17}
    dS_{i,t}^{(h)} = \mu_{i,h}(t) dt + \sigma_{i,h}(t) dB_t, \quad i=1,2,
\end{equation}
where the drift and diffusion coefficients depend on \(B_t\) through:
\begin{align}
    \mu_{i,h}(t) &= \left( \nu_i + \tfrac{1}{2} \sigma_i^2 + \sigma_i \gamma_i h(B_t) + \tfrac{1}{2} \gamma_i^2 h(B_t)^2 \right) S_{i,t}^{(h)}, \label{eq:3.18} \\
    \sigma_{i,h}(t) &= \left( \sigma_i + \gamma_i h(B_t) \right) S_{i,t}^{(h)}. \label{eq:3.19}
\end{align}

If \(\sigma_{1,h}(t) \neq \sigma_{2,h}(t)\), the implied riskless rate \(r_t\) from the unique market price of risk \(\theta_t^{(S_1,S_2)}\):
\begin{equation} \label{eq:3.20}
    \theta_t^{(S_1,S_2)} := \frac{\mu_{1,h}(t) - r_t}{\sigma_{1,h}(t)} = \frac{\mu_{2,h}(t) - r_t}{\sigma_{2,h}(t)}
\end{equation}
solves
\begin{equation} \label{eq:3.21}
    r_t = \frac{\mu_{1,h}(t) \sigma_{2,h}(t) - \mu_{2,h}(t) \sigma_{1,h}(t)}{\sigma_{2,h}(t) - \sigma_{1,h}(t)}.
\end{equation}

\begin{remark}
Due to dependence on \(B_t\), \(r_t\) is stochastic and hence the benchmark rate is not risk-free, invalidating the use of \(\beta_t\) from Eq.~\eqref{eq:3.3} as discount factor in general.
\end{remark}

In the single asset setting \eqref{eq:3.5}--\eqref{eq:3.6}, the diffusion coefficient \(\sigma_h(t)\) vanishing would imply \(\sigma = \gamma = 0\), contradicting the model. This confirms that the classical riskless asset \( \beta \) remains unique with dynamics given by Eq.~\eqref{eq:3.3}.

The no-arbitrage condition on the market price of risk requires
\begin{equation} \label{eq:3.22}
    \theta_t = \frac{\mu_h(t) - r}{\sigma_h(t)} \neq 0, \quad \mathbb{P}\text{-a.s.}.
\end{equation}

This section rigorously extends BSM pricing to incorporate the asymmetric Brownian perturbations via CSYIP, carefully addressing the challenges introduced by path-dependence and stochastic discounting.

\section{Option Pricing in a Market with Two Perfectly Correlated Geometric Asymmetric Brownian Motions and a Risk-less Asset} \label{sec:4}

This section extends the preceding analysis to a market model \((\mathcal{S}_1, \mathcal{S}_2, \mathcal{B}, \mathcal{C})\) comprising two risky assets driven by perfectly correlated GABMs, a risk-less asset, and a European contingent claim (ECC).

\subsection{Market model setup}

The price dynamics for the risky assets \( \mathcal{S}_i, i=1,2 \), are as given in Eqs.~\eqref{eq:3.17}, \eqref{eq:3.18}, and \eqref{eq:3.19}, specifically
\begin{equation} \label{eq:4.1}
    \mathcal{S}_{i,t}^{(h)} = S_0 \exp\left( \nu_i t + \sigma_i B_t + \gamma_i C_t \right), \quad \mathcal{S}_{i,0}^{(h)} = S_0 > 0,
\end{equation}
with \( \nu_i \in \mathbb{R}, \sigma_i \neq 0, \gamma_i \in \mathbb{R} \), and \( C_t = \int_0^t h(B_s) dB_s \) as before. The risk-less asset \(\mathcal{B}\) evolves with dynamics
\begin{equation} \label{eq:4.2}      
    \beta_t = \beta_0 e^{r t}, \quad \beta_0 > 0.
\end{equation}

The price \(f_t\) of the ECC \(\mathcal{C}\) satisfies
\begin{equation} \label{eq:4.3}
    f_t = f(\mathcal{S}_{1,t}^{(h)}, \mathcal{S}_{2,t}^{(h)}, \beta_t, t), \quad t \in [0,T),
\end{equation}
with terminal payoff
$f_T = g(S_{1,T}^{(h)}, S_{2,T}^{(h)}, \beta_T),$
where \(g\) and \(f\) satisfy the regularity conditions in Section 5 of~\cite{duffie_2001}.

This two-asset model aligns with the one-factor market models described in Section 6D of~\cite{duffie_2001}, where a single Brownian driver \(B_t\) represents aggregated market shocks impacting both assets.

\begin{Lemma} (Dynamics of the pricing function) \label{Lemma:4.1}

  Applying Itô's formula to \(f_t\), under assumed differentiability, yields the SDE
\begin{equation} \label{eq:4.4}
    df_t = \mu^{(f)}(t) dt + \sigma^{(f)}(t) dB_t + \gamma^{(f)}(t) dC_t,
\end{equation}
with explicit expressions for \(\mu^{(f)}, \sigma^{(f)}, \gamma^{(f)}\) provided in Appendix~\ref{Appendix:para_4.1}.  
\end{Lemma}

\subsection{Replicating portfolio}

Define a self-financing portfolio \((a_{1,t}, a_{2,t}, b_t)\) replicating \(\mathcal{C}\):
\begin{equation} \label{eq:4.5}
    P_t = a_{1,t} \mathcal{S}_{1,t}^{(h)} + a_{2,t} \mathcal{S}_{2,t}^{(h)} + b_t \beta_t = f_t, \quad t \in [0,T).
\end{equation}

The self-financing condition implies
\begin{equation} \label{eq:4.6}
    dP_t = a_{1,t} d\mathcal{S}_{1,t}^{(h)} + a_{2,t} d\mathcal{S}_{2,t}^{(h)} + b_t d\beta_t = df_t.
\end{equation}

Using the dynamics of \(\mathcal{S}_i^{(h)}\):
\begin{equation} \label{eq:4.7}
\begin{split}
    df_t = {}& a_{1,t} \left( \mu_1(B_t) \mathcal{S}_{1,t}^{(h)} dt + \sigma_1 \mathcal{S}_{1,t}^{(h)} dB_t + \gamma_1 \mathcal{S}_{1,t}^{(h)} dC_t \right) \\
    & + a_{2,t} \left( \mu_2(B_t) \mathcal{S}_{2,t}^{(h)} dt + \sigma_2 \mathcal{S}_{2,t}^{(h)} dB_t + \gamma_2 \mathcal{S}_{2,t}^{(h)} dC_t \right) \\
    & + b_t r_t \beta_t dt,
\end{split}
\end{equation}
where \(\mu_i(B_t), \sigma_i, \gamma_i\) correspond to the drift, volatility, and asymmetry parameters.

Matching coefficients in \eqref{eq:4.4}-\eqref{eq:4.7} leads to
\begin{equation} \label{eq:4.8}
    a_{1,t} = \frac{\partial f}{\partial x_1}, \quad a_{2,t} = \frac{\partial f}{\partial x_2},
\end{equation}
and, from \eqref{eq:4.5},
\begin{equation} \label{eq:4.9}
    b_t = \frac{1}{\beta_t} \left[ f - \frac{\partial f}{\partial x_1} \mathcal{S}_{1,t}^{(h)} - \frac{\partial f}{\partial x_2} \mathcal{S}_{2,t}^{(h)} \right].
\end{equation}

Setting \(\mathcal{S}_{1,t}^{(h)} = x_1\), \(\mathcal{S}_{2,t}^{(h)}=x_2\), \(\beta_t = y\), and \(B_t = b(t)\), the pricing function \(f(x_1, x_2, y, t)\) satisfies the PDE:
\begin{equation} \label{eq:4.10}
\begin{split}
    \frac{\partial f}{\partial t} &- r_t f + r_t x_1 \frac{\partial f}{\partial x_1} + r_t x_2 \frac{\partial f}{\partial x_2} + \frac{\partial f}{\partial y} r_t \beta_0 e^{\int_0^t r_u du} \\
    & + \frac{1}{2} \frac{\partial^2 f}{\partial x_1^2} \left( \sigma_1^2 + \gamma_1^2 h^2 + 2 \sigma_1 \gamma_1 h \right) x_1^2 + \frac{1}{2} \frac{\partial^2 f}{\partial x_2^2} \left( \sigma_2^2 + \gamma_2^2 h^2 + 2 \sigma_2 \gamma_2 h \right) x_2^2 \\
    & + \frac{\partial^2 f}{\partial x_1 \partial x_2} \left( \sigma_1 \sigma_2 + (\sigma_1 \gamma_2 + \gamma_1 \sigma_2) h + \gamma_1 \gamma_2 h^2 \right) x_1 x_2 = 0,
\end{split}
\end{equation}
with \(h = h(b(t))\).

\begin{remark}
When the constant risk-free rate \(r\) is replaced by the stochastic rate \(r_t\) defined in Eq.~\eqref{eq:3.21}, the PDE generalizes to Eq.~\eqref{eq:4.10}. The time dependence of \(r_t\) and its parameter dependencies are detailed in Appendix~\ref{appendix:time_dependent_rt}. If the risk-less asset \(\mathcal{B}\) with these dynamics is unavailable for trade, this PDE governs option pricing~\cite{Rachev_2017}.
\end{remark}

Although the model \((\mathcal{S}_1, \mathcal{S}_2, \mathcal{B}, \mathcal{C})\) is complete and arbitrage-free, the dependence of the option PDE on the path \(B_t\) remains despite the presence of two risky assets driven by the same Brownian motion. The risk-less asset’s dynamics are induced by the dynamics of \(\mathcal{S}_1\) and \(\mathcal{S}_2\), rendering the explicit use of \(\mathcal{B}\) for replication redundant. This demonstrates that, while the two GABM assets and a risk-free security create a formally complete market, the natural filtration includes the Brownian path \(B_t\), which cannot be suffiently summarized by the asset prices alone to eliminate path-dependence from the pricing PDE.

\section{Option Pricing in a Market with Two Perfectly Correlated Geometric Asymmetric Brownian Motions without a Riskless Asset} \label{sec:5}

We consider the market model \((\mathcal{S}_1, \mathcal{S}_2, \mathcal{C})\) where the riskless asset \(\mathcal{B}\) is unavailable for hedging purposes. The dynamics of the tradable assets are specified as in Eqs.~\eqref{eq:3.1}--\eqref{eq:3.2}, explicitly:
\begin{equation} \label{Eq:5.1}
   S_{i,t}^{(h)} = S_0 \exp\left( \nu_i t + \sigma_i B_t + \gamma_i C_t \right), \quad S_{i,0}^{(h)} = S_0 > 0,
\end{equation}
with parameters \(0 \neq \sigma_i \in \mathbb{R}\), and \(\nu_i, \gamma_i \in \mathbb{R}\), where  
\[
C_t = \int_0^t h(B_s) \, dB_s,
\] 
with the function \(h\) as defined previously.

\begin{Lemma} \label{Lemma:5.1}
For a sufficiently smooth contingent claim \(f_t = f(S_{1,t}^{(h)}, S_{2,t}^{(h)}, t)\), It\^o's formula gives the dynamics
\begin{equation} \label{eq:5.2}
    df_t = \mu^{(f)}(t) dt + \sigma^{(f)}(t) dB_t + \gamma^{(f)}(t) dC_t,
\end{equation}
where explicit formulations for \(\mu^{(f)}(t), \sigma^{(f)}(t), \gamma^{(f)}(t)\) are provided in Appendix~\ref{appendix:5}.
\end{Lemma}

A self-financing portfolio with holdings \(a_{1,t}, a_{2,t}\) in \(S_{1,t}^{(h)}\) and \(S_{2,t}^{(h)}\) respectively, and no position in a riskless asset, satisfies
\begin{equation} \label{eq:5.3}
    P_t = a_{1,t} S_{1,t}^{(h)} + a_{2,t} S_{2,t}^{(h)} = f(S_{1,t}^{(h)}, S_{2,t}^{(h)}, t).
\end{equation}

Matching stochastic terms of \(df_t\) in Eq.~\eqref{eq:5.2} with the portfolio dynamics yields the unique admissible trading strategy:
\begin{equation} \label{eq:5.4}
    a_{1,t} = \frac{\partial f(S_{1,t}^{(h)}, S_{2,t}^{(h)}, t)}{\partial x_1}, \quad
    a_{2,t} = \frac{\partial f(S_{1,t}^{(h)}, S_{2,t}^{(h)}, t)}{\partial x_2}.
\end{equation}
This implies the replicating condition
\[
\frac{\partial f}{\partial x_1} S_{1,t}^{(h)} + \frac{\partial f}{\partial x_2} S_{2,t}^{(h)} = f(S_{1,t}^{(h)}, S_{2,t}^{(h)}, t).
\]

Following the arguments in Section~\ref{sec:4}, the value function \(f(x_1, x_2, t)\) satisfies the nonlinear PDE
\begin{equation} \label{eq:5.5}
\begin{aligned}
    \frac{\partial f}{\partial t} 
    &+ \tfrac{1}{2} \frac{\partial^2 f}{\partial x_1^2}
    \left[ \sigma_1^2 + \gamma_1^2 h(b(t))^2 + 2 \sigma_1 \gamma_1 h(b(t)) \right] x_1^2 \\
    &+ \tfrac{1}{2} \frac{\partial^2 f}{\partial x_2^2}
    \left[ \sigma_2^2 + \gamma_2^2 h(b(t))^2 + 2 \sigma_2 \gamma_2 h(b(t)) \right] x_2^2 \\
    &+ \frac{\partial^2 f}{\partial x_1 \partial x_2}
    \left[ \sigma_1 \sigma_2 + (\sigma_1 \gamma_2 + \sigma_2 \gamma_1) h(b(t)) + \gamma_1 \gamma_2 h(b(t))^2 \right] x_1 x_2 = 0,
\end{aligned}
\end{equation}
with terminal condition $f(x_1, x_2, T) = g(x_1, x_2).$

The Feynman--Kac formula characterizes the unique solution to Eq.~\eqref{eq:5.5}. Notably, this market model lacks an explicit risk-free rate because no riskless asset is available for hedging. Instead, hedging occurs solely through the two perfectly correlated risky assets.

\begin{remark}
Despite absence of an explicit discount rate in the pricing PDE, an endogenous riskless rate can be recovered from the dynamics of the assets and option, as shown in Eq.~\eqref{eq: 4.26_appendix}.
\end{remark}

\section{Binomial Option Pricing under Asymmetric Random Walk} \label{sec:6}

This section develops a binomial option pricing framework under an asymmetric random walk constructed via the CSYIP as established in Proposition~\ref{prop:2.1}, drawing on the approaches in \cite{Hu_2020a} \& \cite{Hu_2020b}. The model encompasses classical Bachelier and Black--Scholes--Merton (BSM) models as special or limiting cases, extended by an asymmetric component that captures skewness beyond standard formulations.

\subsection{Asset price dynamics}

Define the discrete risky asset price on the time grid \([0,T]\) with step size \(\Delta t = \frac{T}{n}\), \(n \in \mathbb{N}\):
\begin{equation} \label{eq:6.1}
    S_t^{(n,h)} = S_0 \exp \left( \nu k \Delta t + \sigma \sqrt{\Delta t} \sum_{j=1}^k \xi_j + \gamma \sqrt{\Delta t} \sum_{j=1}^k h_T\left( \sqrt{\Delta t} \sum_{i=1}^{j-1} \xi_i \right) \xi_j \right), 
\end{equation}
where  \(t \in [k\Delta t, (k+1)\Delta t),\) and \(\{\xi_j\}_{j=1}^n\) are i.i.d.\ symmetric Bernoulli random variables (\(\mathbb{P}(\xi_j = \pm 1) = \frac{1}{2}\)) and \(h_T\) is a scaled piecewise-continuous function encoding asymmetric, state-dependent effects. By the CSYIP, the process \(S_t^{(n,h)}\) converges weakly to the GABM:
\begin{equation} \label{eq:6.2}
    S_t^{(h)} = S_0 \exp \left( \nu t + \sigma B_t + \gamma \int_0^t h_T(B_v) \, dB_v \right),
\end{equation}
where \(B_t\) is standard Brownian motion.

\begin{remark}
Setting \(\gamma=0\) recovers the classical BSM model, while alternative scalings and choices of \(h_T\) correspond to the Bachelier model. Thus, the framework unifies and generalizes prominent pricing models incorporating asymmetric, path-dependent volatility.
\end{remark}

\subsection{Option prices formula}

The model exhibits path-dependent discrete-time volatility at step \(k\):
\begin{equation} \label{eq:path_dep_Vol_MB}
    \eta_{k,\Delta t} = \sigma + \gamma h_T \left( \sqrt{\Delta t} \sum_{i=1}^k \xi_i \right),
\end{equation}
reflecting dependence on the cumulative history of the random signs.

Define a self-financing hedging portfolio \(P_{k\Delta t}^{(n,h)}\) at time \(k\Delta t\) by
\[
P_{k \Delta t}^{(n,h)} = D_{k \Delta t}^{(n,h)} S_{k \Delta t}^{(n,h)} - f_{k \Delta t}^{(n,h)},
\]
where the hedge ratio is given by
\begin{equation} \label{eq:6.6_MB}
    D_{k \Delta t}^{(n,h)} = \frac{f_{(k+1) \Delta t}^{(n,h,u)} - f_{(k+1) \Delta t}^{(n,h,d)}}{S_{k \Delta t}^{(n,h)} \exp(\nu \Delta t) \left( \exp(\eta_{k,\Delta t} \sqrt{\Delta t}) - \exp(-\eta_{k,\Delta t} \sqrt{\Delta t}) \right)},
\end{equation}
with the backward recurrence for option prices:
\begin{equation} \label{eq:6.7_MB}
    f_{k \Delta t}^{(n,h)} = e^{-r \Delta t} \left[ q_{k \Delta t}^{(n,h)} f_{(k+1) \Delta t}^{(n,h,u)} + (1 - q_{k \Delta t}^{(n,h)}) f_{(k+1) \Delta t}^{(n,h,d)} \right].
\end{equation}

The risk-neutral probabilities are state-dependent:
\begin{equation} \label{eq:6.8_MB}
    q_{k \Delta t}^{(n,h)} = \frac{e^{r \Delta t} - e^{-\eta_{k,\Delta t} \sqrt{\Delta t}}}{e^{\eta_{k,\Delta t} \sqrt{\Delta t}} - e^{-\eta_{k,\Delta t} \sqrt{\Delta t}}}.
\end{equation}

As \(n \to \infty\), the discrete-time \(\mathbb{Q}\)-price process \(S_t^{(n,h,\mathbb{Q})}\) converges in the Skorokhod space \(\mathcal{D}[0,T]\) to the continuous-time GABM \(S_t^{(h)}\) of Eq.~\eqref{eq:6.2}, realizing a diffusion with path-dependent volatility.

This approach enables effective numerical pricing of European contingent claims under asymmetric volatility effects driven by path-dependent Brownian functionals.

\begin{remark}
When \(\gamma=0\), the model reduces to the classical Cox--Ross--Rubinstein (CRR) binomial tree. The asymmetric extension captures skewness and volatility state-dependence, facilitating more realistic option valuation.
\end{remark}

\section{Application to Empirical Data} \label{subsec:empirical}

This section implements the empirical component of our study by applying the GABM framework to real-world options data. The objective is to evaluate both the numerical tractability of the model and its market relevance.

\subsection{Construction of price surfaces and calibration pipeline}

We construct theoretical European call and put price surfaces on a binomial sign-path lattice with nodewise, state-dependent volatility as described in Eq.~\eqref{eq:path_dep_Vol_MB}. At each maturity–strike node \((T_i, K_j)\), we recover a state-dependent implied-volatility surface by minimizing relative pricing errors against observed market prices.

The asymmetric diffusion component is encoded by the function \(h(\cdot)\) defined in Eq.~\eqref{eq:15}, with baseline parameters fixed to 
$\mu=0, \quad \alpha=8, \quad \beta=2, \quad \delta=0.5, \quad \gamma=0.35,$
chosen to represent stylized market features such as smile and skew effects.
This calibration scheme tests whether the state-dependent asymmetry captured by \(h(\cdot)\) can replicate important volatility surface phenomena (smiles, skews, and term structure) while preserving internal consistency such as absence of static arbitrage across strikes and maturities.

\subsection{Data sources and preprocessing}

The empirical evaluation is implemented on options data for three liquid U.S. equities: Apple (AAPL), Amazon (AMZN), and Microsoft (MSFT), selected from the “Magnificent Seven” roster to ensure cross-asset robustness within a homogeneous large-cap cohort. Spot prices \(S_0\) correspond to closing prices on 21~February~2025\footnote{The parameters were computed using an eight-year historical window.}. Time to maturities are computed on an ACT/252 basis unless otherwise specified. The risk-free interest rate is proxied by the continuously compounded yield of the 10-year U.S. Treasury on the valuation date, held constant across maturities for simplicity. Extension to a full term structure, \(r(T)\), is straightforward and does not alter the methodology.

Calls are used exclusively due to parity and exercise considerations—European and American call option prices coincide whereas put prices may differ. Supplementary data snapshots on 22 February 2025 support the cross-section construction for 21 February 2025.\footnote{Option and stock price data accessed via Bloomberg Professional Services~\cite{Gnawali_2025}.}

The raw option price matrix is refined by:
\begin{itemize}
    \item Removing stale and zero-volume quotes, as well as outliers.
    \item Applying shape-preserving piecewise cubic Hermite interpolating polynomial (PCHIP) interpolation firstly along maturities, then strikes, to overcome sparsity.
    \item Enforcing no-arbitrage constraints across strikes and maturities.
    \item Mapping prices onto a common grid of \((T,K)\) values, equivalently \((T,M)\) with moneyness \(M = \frac{K}{S_0}\).
\end{itemize}

Recall the continuous-time GABM price dynamics from Eq.~\eqref{eq:6.2}:
\[
S_t^{(h)} = S_0 \exp \left( \nu t + \sigma B_t + \gamma C_t \right), \quad C_t = \int_0^t h(B_s) \, dB_s,
\]
with lattice approximation specified by Eq.~\eqref{eq:6.1}:
\begin{equation}
S^{(n,h)}_{(k+1)\Delta t} = S^{(n,h)}_{k\Delta t} \exp\left( \nu \Delta t + \eta_{k,\Delta t} \sqrt{\Delta t} \, \xi_{k+1} \right),
\quad \eta_{k,\Delta t} = \sigma + \gamma h_T\left( \sqrt{\Delta t} \sum_{i=1}^k \xi_i \right),
\label{eq:lattice-dynamics-recall}
\end{equation}
where \(\xi_i \in \{ -1, +1 \}\) are i.i.d. symmetric signs, and up/down factors at node \((k,m)\) derive from $
u_k = \exp(\eta_{k,\Delta t} \sqrt{\Delta t}), \quad d_k = \exp(-\eta_{k,\Delta t} \sqrt{\Delta t}).$

The risk-neutral nodewise probabilities, which preserve the martingale property of discounted asset prices, are given by Eq.~\eqref{eq:6.8_MB}:
\begin{equation}  \label {eq:q-node}
   q_{k \Delta t} = \frac{e^{r \Delta t} - e^{-\eta_{k,\Delta t} \sqrt{\Delta t}}}{e^{\eta_{k,\Delta t} \sqrt{\Delta t}} - e^{-\eta_{k,\Delta t} \sqrt{\Delta t}}} \approx \frac{1}{2} + \frac{r - \eta_{k,\Delta t}^2/2}{2 \eta_{k,\Delta t}} \sqrt{\Delta t}, 
\end{equation}

where the approximation holds for small \(\Delta t\).

\subsection{Recombining lattice and numerical stability}

The sum of the signs \(\sum_{i=1}^k \xi_i = 2m - k\) is uniquely determined by the node index \((k,m)\), ensuring that the lattice is recombining when indexed by \((k,m)\). This structural property enables tractable dynamic programming and reduces computational complexity required for pricing.
The asymmetric random walk-based binomial tree thus efficiently captures path-dependent volatility skew while preserving computational feasibility similar to classical binomial models.

This numerical pipeline allows us to build flexible implied volatility surfaces reflecting skewness and state dependence observed in market data, while maintaining no-arbitrage conditions and stability in calibration. It provides a practically implementable tool to test the empirical relevance of the GABM model in cross-sectional option markets.

\subsection{Option prices}  \label{sec:OpPrices}

In computations on large \((T,K)\) grids, we price European options by Monte Carlo simulation of binomial sign paths on the risk–neutral lattice defined by Eqs.~\eqref{eq:lattice-dynamics-recall}–\eqref{eq:q-node}. For each \((T,K)\) we simulate \(N_{\text{paths}}\) sequences \(\{\xi_i\}_{i=1}^n\), evolve \(S^{(n,h)}\) with the corresponding nodewise \(\eta_{k,\Delta t}\), and estimate discounted payoffs,

\begin{align*}
C(S_0,K,T)
  &\approx e^{-rT}\,\frac{1}{N_{\text{paths}}}
     \sum_{\ell=1}^{N_{\text{paths}}}\!\big(S_T^{(\ell)}-K\big)^+,\\[1ex]
P(S_0,K,T)
  &\approx e^{-rT}\,\frac{1}{N_{\text{paths}}}
     \sum_{\ell=1}^{N_{\text{paths}}}\!\big(K-S_T^{(\ell)}\big)^+ .
\end{align*}

Indexing nodes by \((k,m)\) (time level \(k\) and number of up moves \(m\)), the local volatility becomes \(\eta_{k,m}\), with \(u_{k,m}=\exp(\eta_{k,m}\sqrt{\Delta t})\), \(d_{k,m}=\exp(-\eta_{k,m}\sqrt{\Delta t})\), and \(q_{k,m}\) as in Eq.~\eqref{eq:q-node}. European options can then be priced by backward induction on the \((n{+}1)\)-level \emph{recombining} tree. For a call with strike \(K\) and \(T=n\Delta t\), terminal payoffs are \(C_{n,m}=\max\{S_{n,m}-K,0\}\), and interior values satisfy
\[
C_{k,m}=e^{-r\Delta t}\Big(q_{k,m}\,C_{k+1,m+1}+(1-q_{k,m})\,C_{k+1,m}\Big), \qquad k=n-1,\ldots,0,
\]
with the analogous recursion for puts,
\[
P_{n,m}=\max\{K-S_{n,m},0\}, \quad  
P_{k,m}=e^{-r\Delta t}\Big(q_{k,m}\,P_{k+1,m+1}+(1-q_{k,m})\,P_{k+1,m}\Big).
\]
European put prices may equivalently be recovered from put-call parity
\[C^{E}(S_0,K,T)-P^{E}(S_0,K,T)=S_0-K e^{-rT},\] 
and the standard max operator applies nodewise when early exercise is relevant. In this paper we produce figures via Monte Carlo over sign paths for scalability on large \((T,K)\) grids; the backward-induction formulation provides an equivalent dynamic-programming view of the same lattice.

To stabilize estimates we cap per-step log-moves \(|a_k|:=|\eta_{k,\Delta t}\sqrt{\Delta t}|\le A_{\max}\), anchor the \(T{=}0\) boundary with exact payoffs, and apply a product Gaussian kernel over \((T,M)\) to regularize simulation/quote noise. No-arbitrage clamps (e.g., \(0\le C\le S_0\), \(0\le P\le K e^{-rT}\), convexity in \(K\)) are enforced post-smoothing to preserve economic shape. An equivalent backward–induction formulation can also be written by indexing nodes as \((k,m)\), defining \(\eta_{k,m}\), \(u_{k,m}\), \(d_{k,m}\), and \(q_{k,m}\), and applying the usual dynamic-programming recursions for calls and puts; in this paper, the figures are produced by the Monte Carlo method, while the dynamic–programming scheme provides an alternative view of the same lattice.

Let $C^{(\text{emp})}(S_0,T_i,K_j)$, $i = 1, ..., I$, $j = 1, ..., J$,
denote published prices for a call option having $\mathcal{S}$ as the underlying.
Let
$C^{(\text{th})}\!\big(S_0,T_i,K_j;\,\eta_{k,\Delta t},\, r\big)$
denote the respective theoretical option prices computed from the model. 
We computed theoretical call option prices on $t = $ 02/21/2025
for maturity times corresponding to trading dates $t+T$, $T = 1, ..., T_I$ 
and strike prices $K \in \{K_1, ..., K_J\}$.

\begin{table}
\caption{Parameter values computed from the historical data}
	\label{tab:params}
\centering
\begin{tabular}{l c c c}
\toprule
Stock   & $S_0$ & $\sigma$ & $r$ \\
\midrule
AAPL   &  245.55 & 0.0180 & 0.0432 \\
AMZN   & 216.58 &  0.0202 & \omit \\
MSFT   & 408.21 &  0.0170 & \omit \\
\bottomrule
\end{tabular}
\end{table}

\begin{figure}[htbp]
  \centering
  \includegraphics[width=0.32\linewidth]{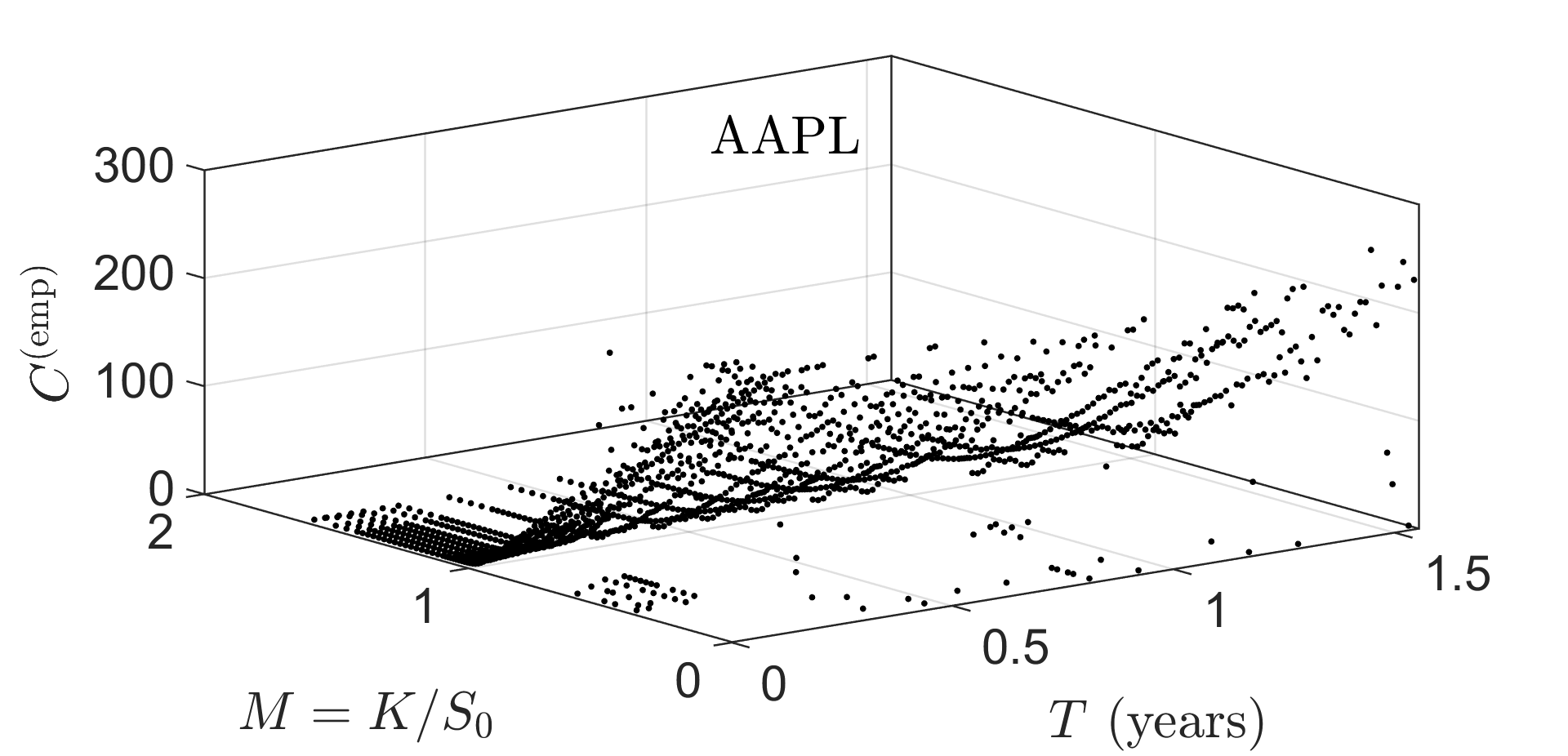}
  \includegraphics[width=0.32\linewidth]{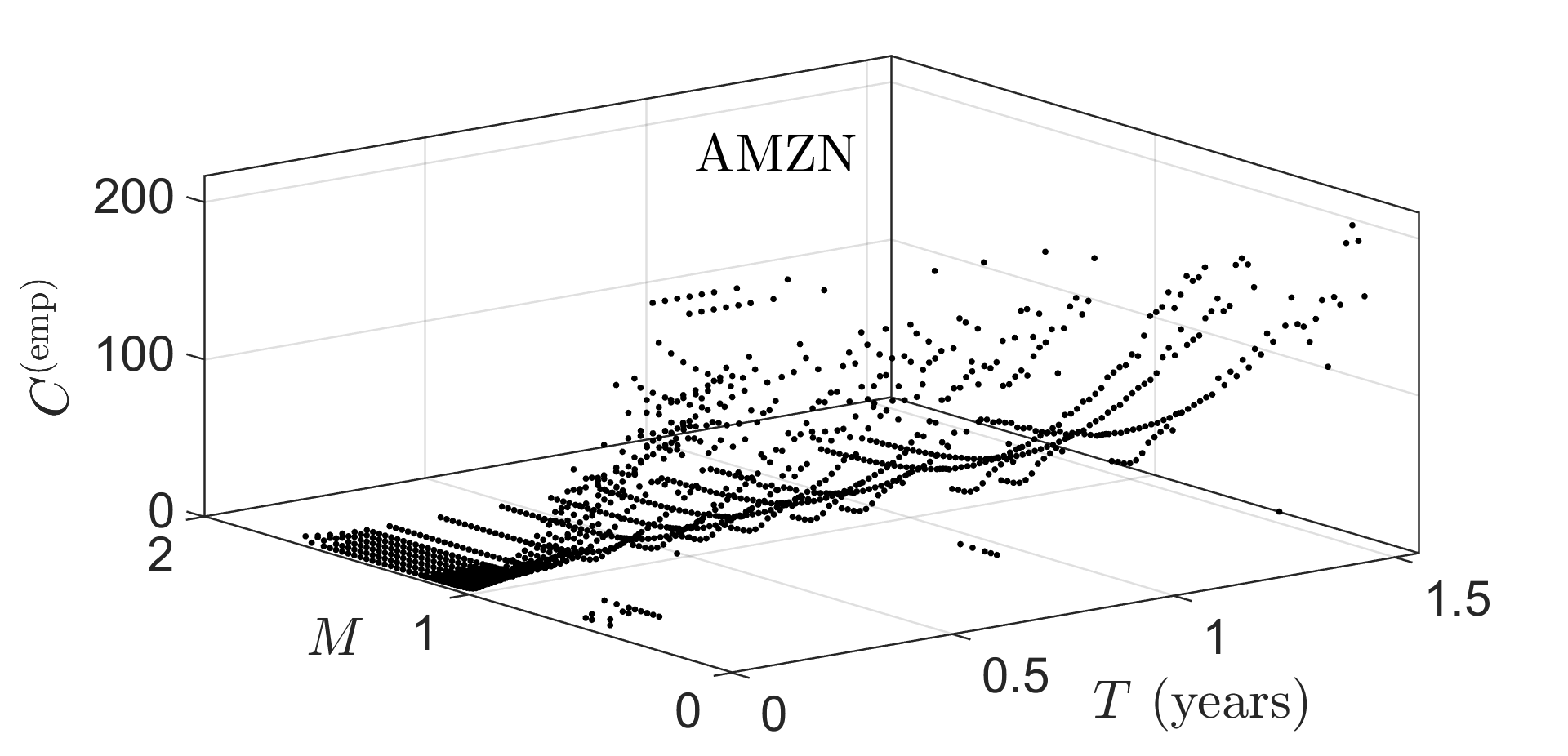}
  \includegraphics[width=0.32\linewidth]{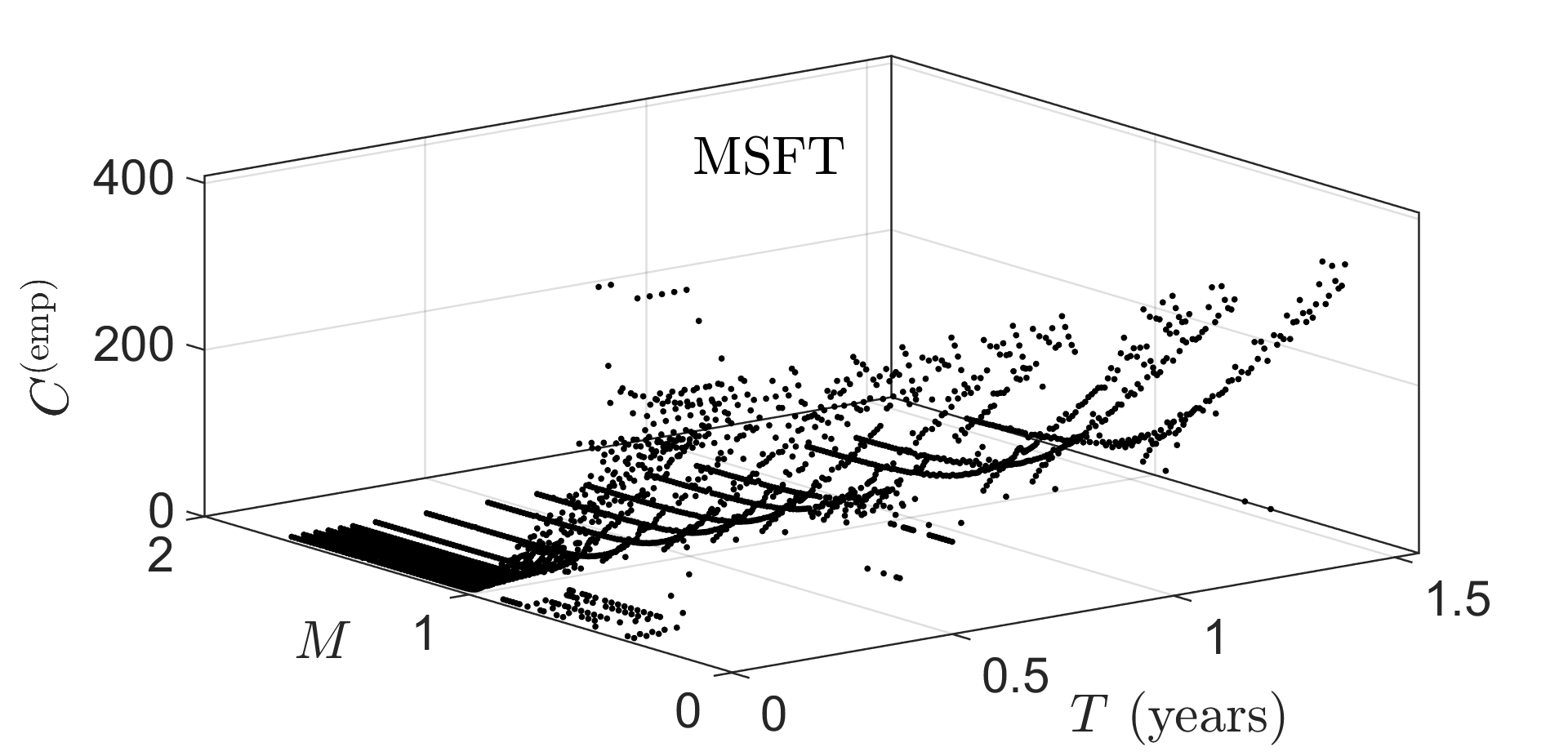}
  \includegraphics[width=0.32\linewidth]{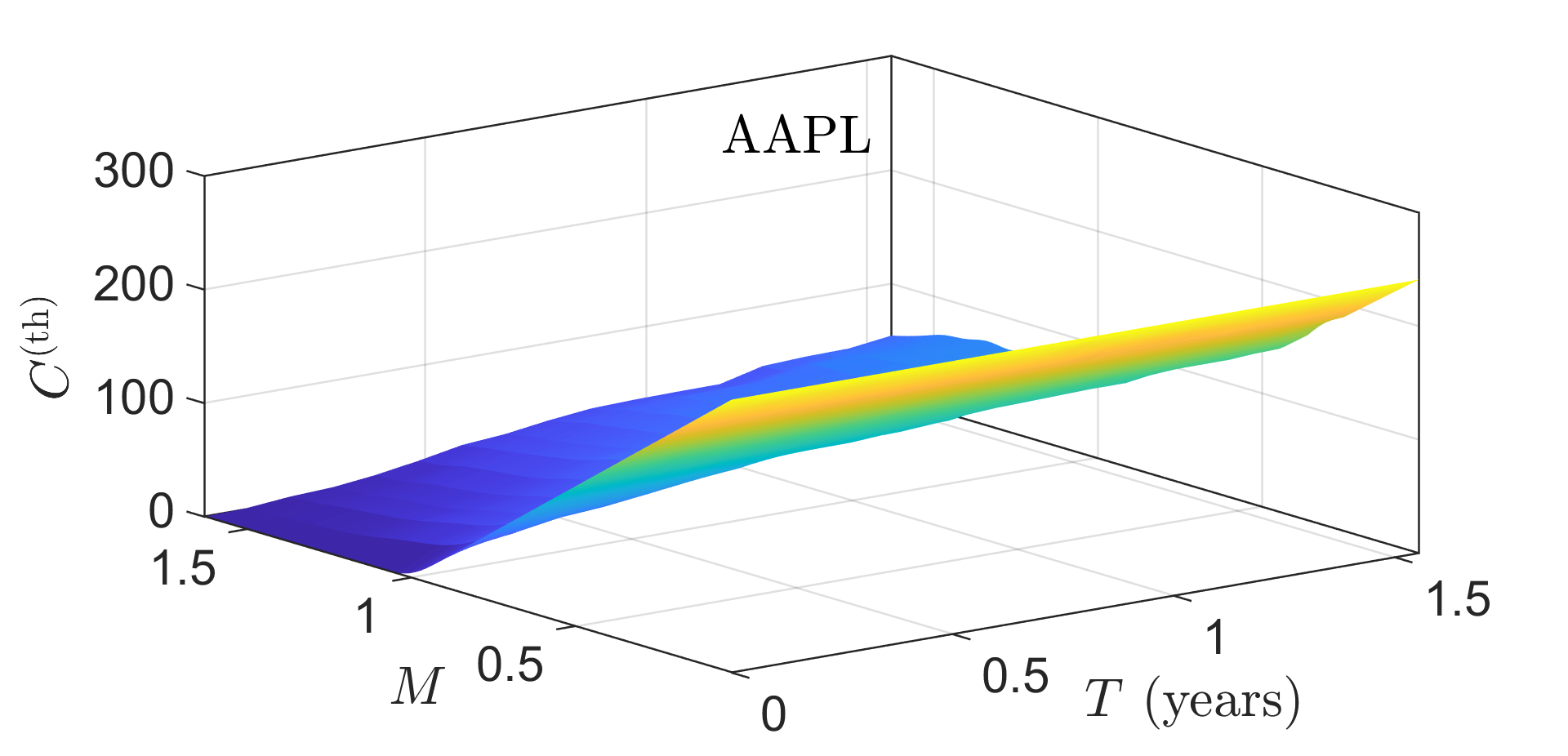}
  \includegraphics[width=0.32\linewidth]{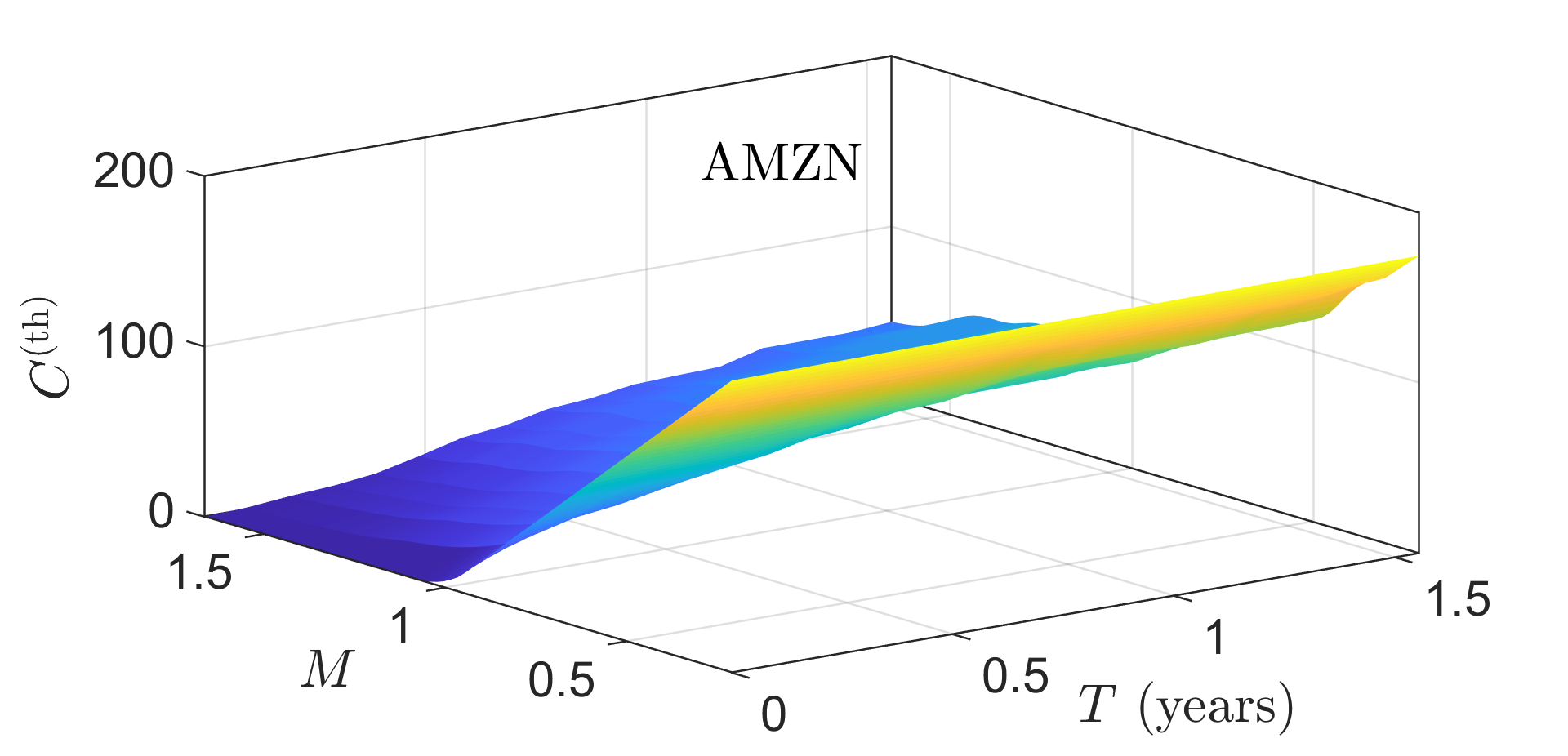}
   \includegraphics[width=0.32\linewidth]{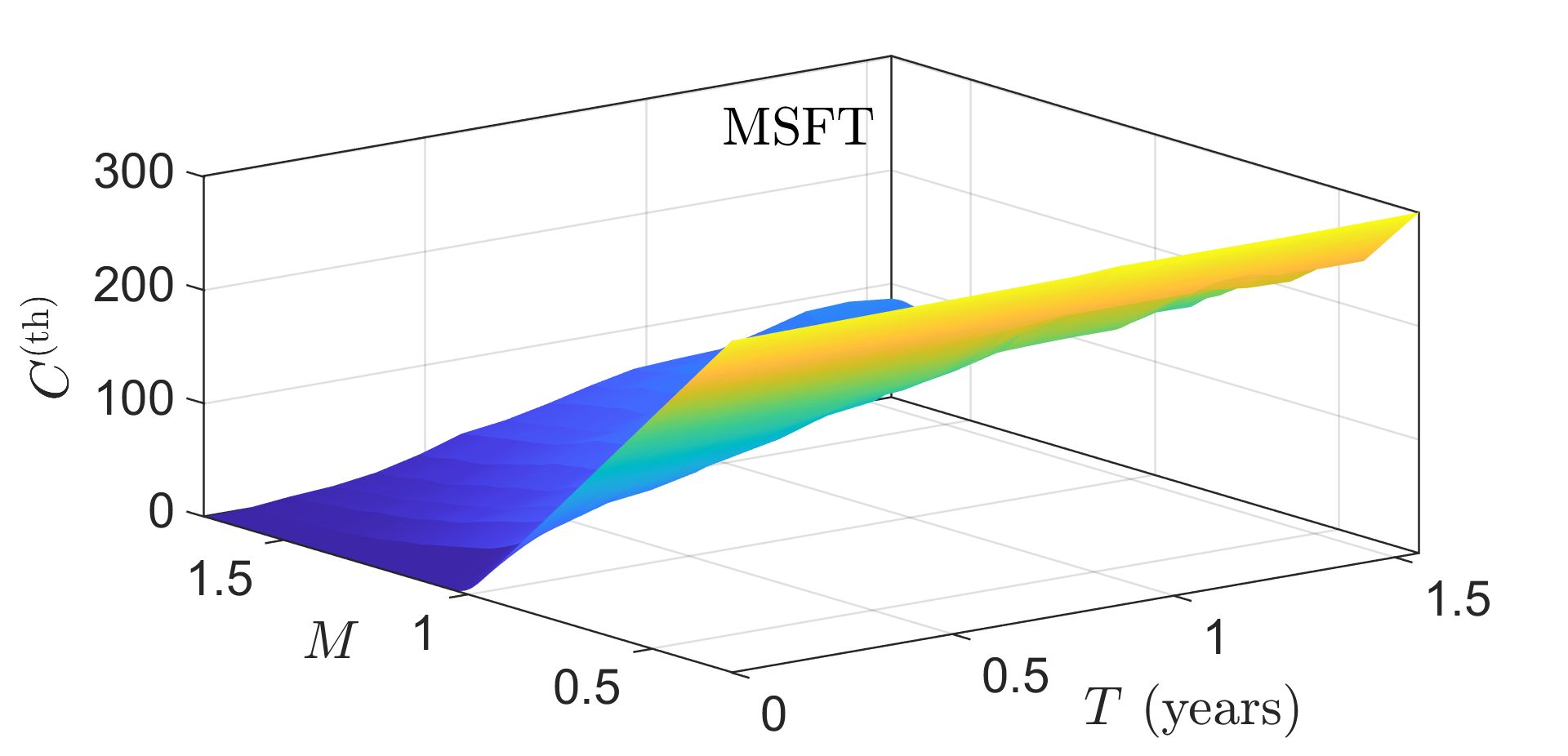}
 \includegraphics[width=0.32\linewidth, height=0.15\textheight]{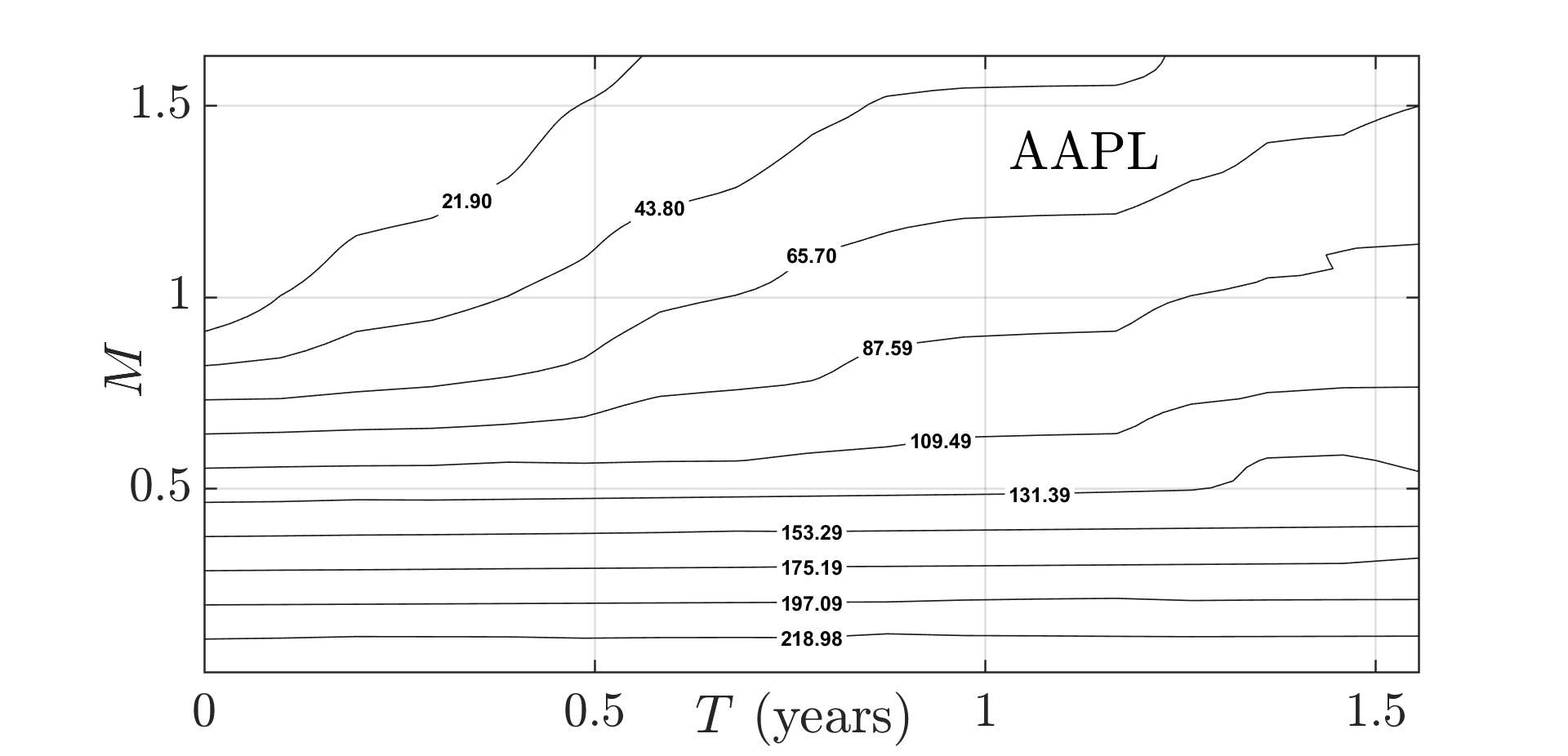}
  \includegraphics[width=0.32\linewidth, height=0.15\textheight]{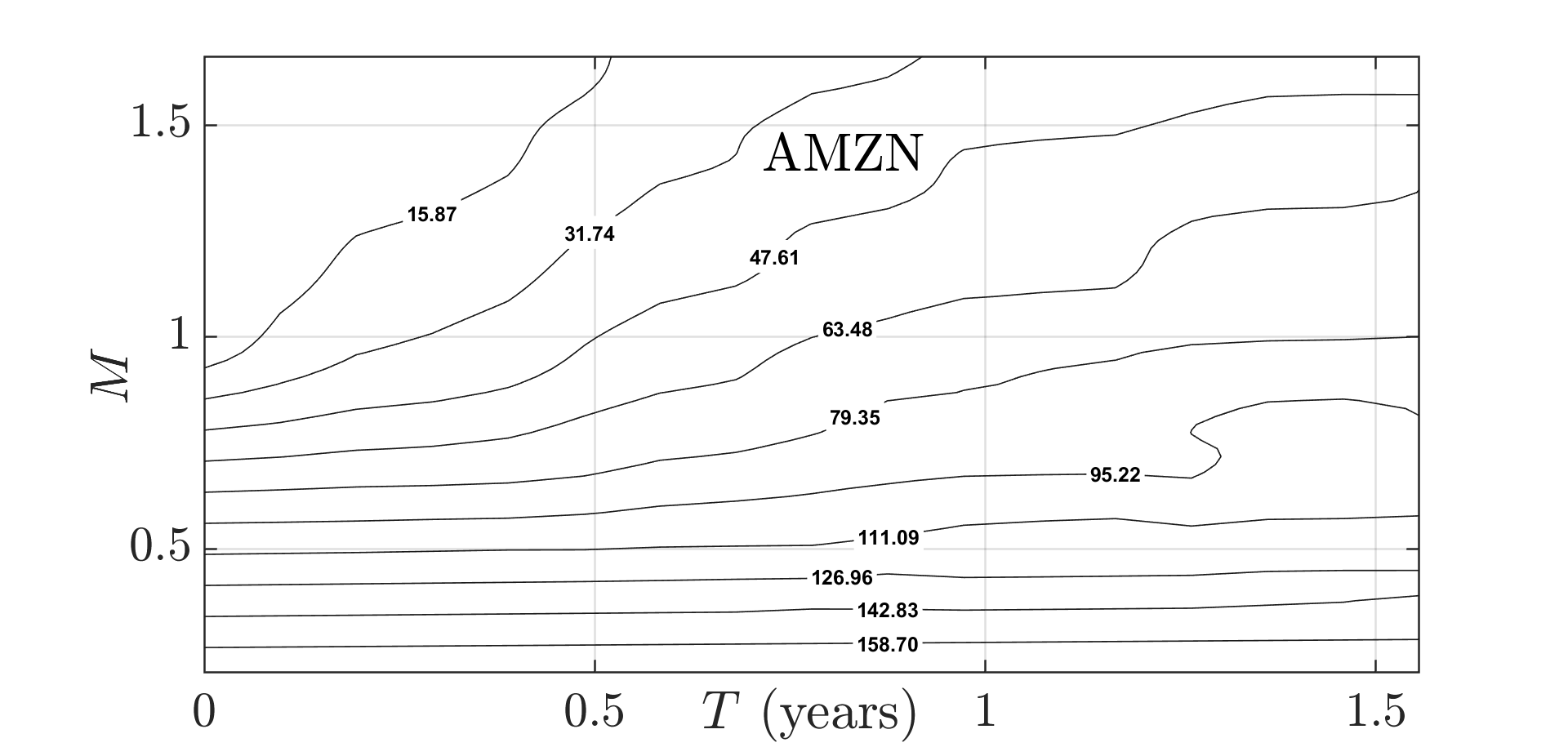}
   \includegraphics[width=0.32\linewidth, height=0.15\textheight]{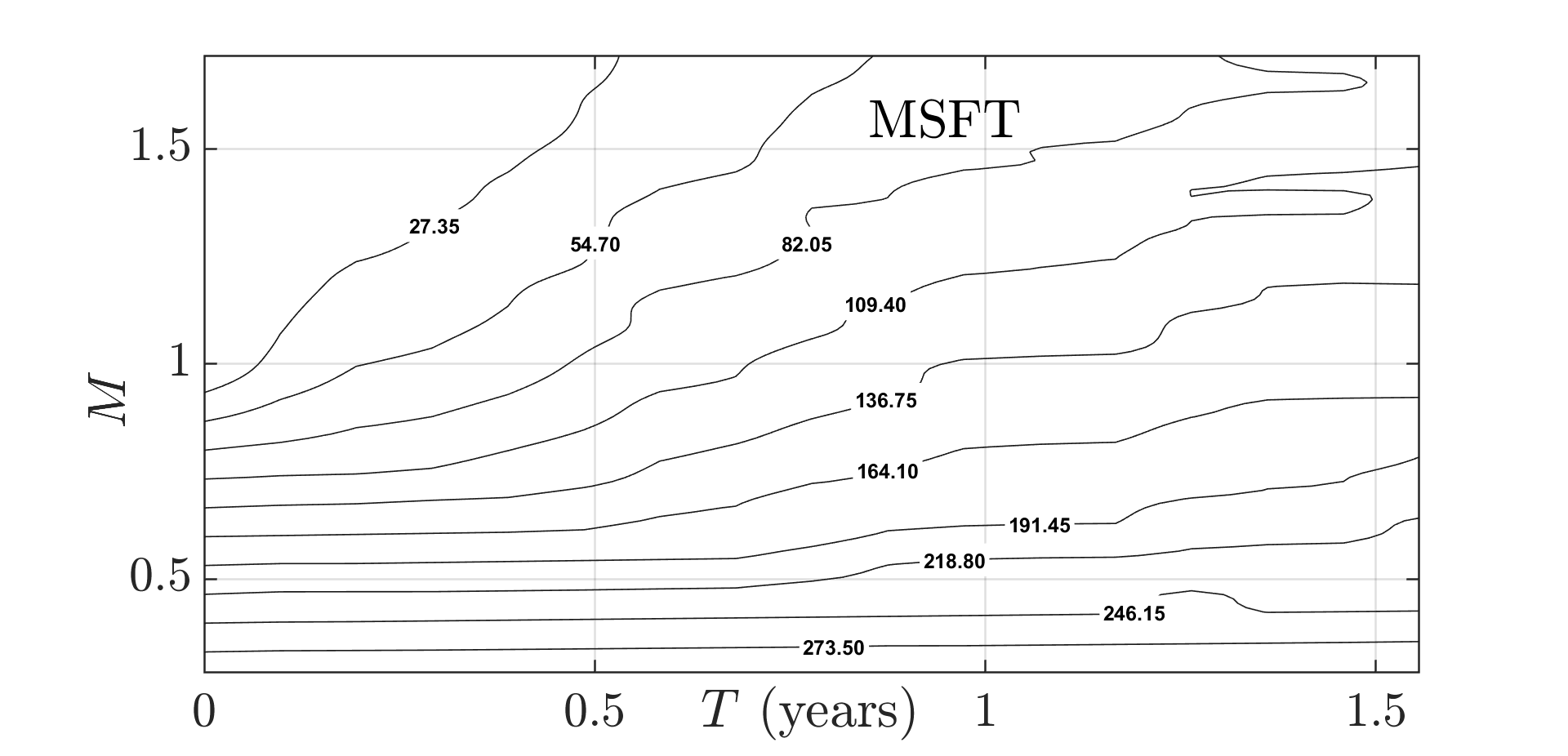}

   \caption{(top row) Empirical call option prices for the three assets.
		(middle row) Computed (theoretical) call option prices.
		(bottom row) Surface contours of the theoretical price surface projected on the ($T,M$) plane.
	}
  \label{fig:Call_surface}
\end{figure} 

Figure~\ref{fig:Call_surface} reports European call prices for AAPL, AMZN, and MSFT on 21~Feb~2025 over a grid in moneyness \(M = K/S_0\) and maturity \(T\) (years). 
The empirical point clouds in the top row show the usual no-arbitrage geometry: prices are decreasing in \(M\), reflecting declining intrinsic value as strikes rise, and increasing in \(T\), reflecting the build-up of time value and risk-neutral variance. 
These raw surfaces are granular and uneven, particularly for short maturities and far-from-the-money options, due to discrete strikes, quotation noise, and sparse observations.

The middle row displays the corresponding model-implied prices obtained from Monte Carlo simulation of the stochastic lattice in Eqs.~\eqref{eq:lattice-dynamics-recall}--\eqref{eq:q-node}, followed by Gaussian kernel smoothing under monotonicity, convexity, and bound constraints. 
The resulting surfaces are visually smooth in both arguments, yet preserve the economically required shape in \(K\) and \(M\): they remain decreasing and convex in \(K\), increasing in \(T\), and bounded between zero and \(S_0\). 
In levels and slopes, AAPL and MSFT exhibit relatively high near-the-money prices and a steep negative gradient in \(M\) around \(M \approx 1\), whereas AMZN shows a flatter profile and milder dependence on moneyness, consistent with a lower effective volatility in this sample.

The bottom row plots contours of the theoretical surfaces in the \((T,M)\) plane over a wide range of moneyness, including deep in-the-money \((M < 1)\) and out-of-the-money \((M > 1)\) strikes. 
For AAPL and MSFT, contour lines are relatively flat and widely spaced at low \(M\), reflecting slow variation in price when the option is deep in-the-money and dominated by intrinsic value, but they become much more tightly packed as \(M\) approaches and exceeds \(1\), indicating a steep drop in call value and elevated gamma in the transition from in-the-money to out-of-the-money. 
For AMZN, contour spacing is more uniform across \(M\), consistent with a smoother cross-sectional response and a less pronounced skew. 
The asymmetric, non-parallel bending of contours as \(T\) increases is characteristic of a state-dependent local volatility surface, parametrized by \(\eta_{k,\Delta t}\), rather than the flat term structure implied by constant-volatility models. 
Taken together, these contour plots confirm that the calibrated lattice with nodewise local volatility and regularization reproduces both the strong moneyness sensitivity near \(M \approx 1\) and the slower variation in deep in-the-money regions, while respecting the global no-arbitrage shape of the surface.

Figure~\ref{fig:Put-surface} reports model-implied European put prices for AAPL, AMZN, and MSFT over the same \((T,M)\) grid as in the call analysis, with \(M = K/S_0\) and \(T\) measured in years. The three upper panels show smoothed put price surfaces obtained via put--call parity applied to the calibrated call surfaces, ensuring internal consistency with the underlying stochastic lattice and respecting static no-arbitrage across strikes and maturities.
Across all assets, prices are essentially zero for short-dated, out-of-the-money contracts with \(M<1\), then rise steadily as either maturity or moneyness increases, reflecting the growing value of downside protection as strikes move further above spot and as the horizon over which adverse moves can realize becomes longer.

The lower panels display contour plots of the same surfaces in the \((T,M)\) plane and make the cross-sectional geometry more transparent.
For each fixed maturity, contour levels are widely spaced at low \(M\), indicating slow variation in the deep out-of-the-money region, and become progressively more compressed as \(M\) moves above one, signalling a sharp increase in put value as the option transitions into and further in-the-money.

MSFT exhibits the highest overall put price levels and the tightest contour spacing for \(M>1\), consistent with a stronger effective volatility and a steeper downside skew, whereas AAPL and AMZN display lower peaks and slightly more uniform contour spacing, corresponding to a more moderate premium for crash protection. Along the maturity dimension, contours bend upward but flatten somewhat at longer tenors, indicating that the marginal increase in put value with additional time decays as \(T\) grows, in line with a non-flat term structure of variance.

Overall, the combined evidence from the empirical and model-implied call and put surfaces shows that the calibrated GABM lattice, together with Gaussian smoothing and no-arbitrage constraints, produces economically coherent option price geometries across moneyness and maturity for all three underlyings.

\subsection{Calibration: state-dependent implied volatility}\label{subsubsec:implied}

The aim of this subsection is to \emph{invert} observed market quotes into a state-dependent implied-volatility field that is internally consistent with the GABM lattice. 
Holding fixed the state dependent component \(h(\cdot)\) (estimated in the previous subsection) and the risk-neutral node mechanics in Eqs.~\eqref{eq:lattice-dynamics-recall}-\eqref{eq:q-node}, we recover at each observed maturity-strike node \((T_i,K_j)\) a scalar volatility \(\sigma^{(\mathrm{imp})}(T_i,K_j)\) such that the model price matches the market quote as closely as possible. We adopt a \emph{relative} error criterion to balance liquid/cheap options and ensure numerical stability across the grid.

For each observed node \((T_i,K_j)\) we therefore solve
\begin{equation}
\sigma^{(\text{imp})}(T_i,K_j)
=\arg\min_{\sigma}\,
\left(
\frac{
C^{(\text{th})}\!\big(S_0,T_i,K_j;\,\eta_{k,\Delta t},\, r\big)
-
C^{(\text{emp})}\!\big(S_0,T_i,K_j\big)}
{C^{(\text{emp})}\!\big(S_0,T_i,K_j\big)}
\right)^{\!2},
\label{eq:opt-impvol}
\end{equation}
where \(G^{(\text{th})}\) denotes the lattice price obtained with nodewise local volatility \(\eta_{k,\Delta t}\) (which depends on the fitted \(h(\cdot)\) and current \(\sigma\)), and \(G^{(\text{emp})}\) is the observed market price. In practice, we compute \(\sigma^{(\mathrm{imp})}\) only at nodes with quotes, using a bounded one-dimensional optimizer with sensible bracketing (e.g., \(\sigma\in[0.01,\,3.00]\) in annualized units), initialized by the at-the-money historical volatility or a BSM ATM proxy. Because call prices are (strictly) increasing in volatility on the lattice, the minimization is well–posed and numerically stable for calls; puts can be handled analogously or checked via put-call parity.

The discrete set \(\{\sigma^{(\text{imp})}(T_i,K_j)\}\) is then smoothed over \((T,K)\) with a Gaussian kernel to obtain a continuous implied-volatility surface. In practice, implied values are first computed only at nodes with available quotes; the kernel smoother then interpolates across the full \((T,K)\)-grid. Bandwidth is selected to balance faithfulness to observed quotes and surface regularity, and we verify that the smoothed surface preserves standard shape properties (e.g., no static-arbitrage violations) before proceeding to diagnostics and cross-asset comparisons.

\begin{figure}[htbp]
  \centering
   \includegraphics[width=0.32\linewidth]{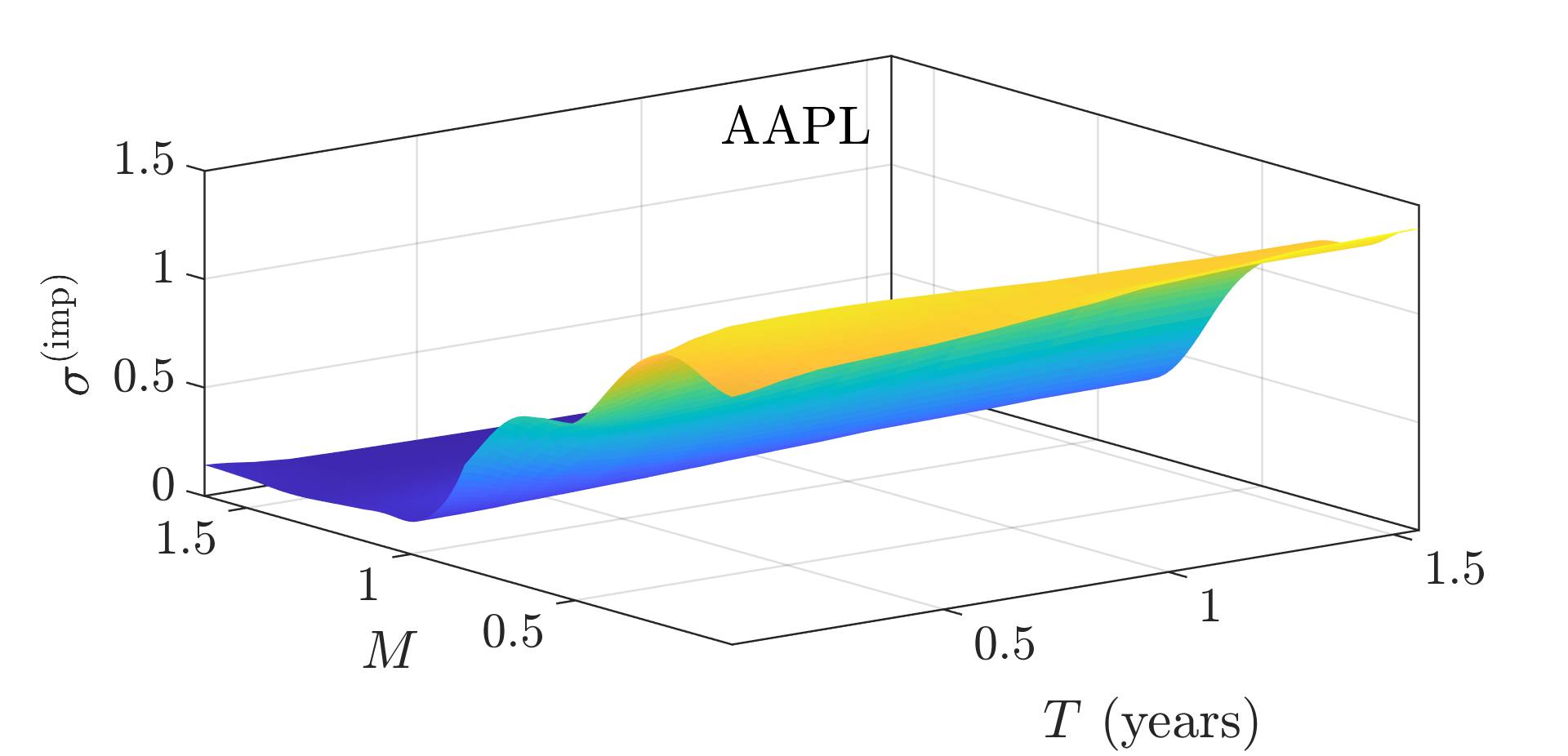}
   \includegraphics[width=0.32\linewidth]{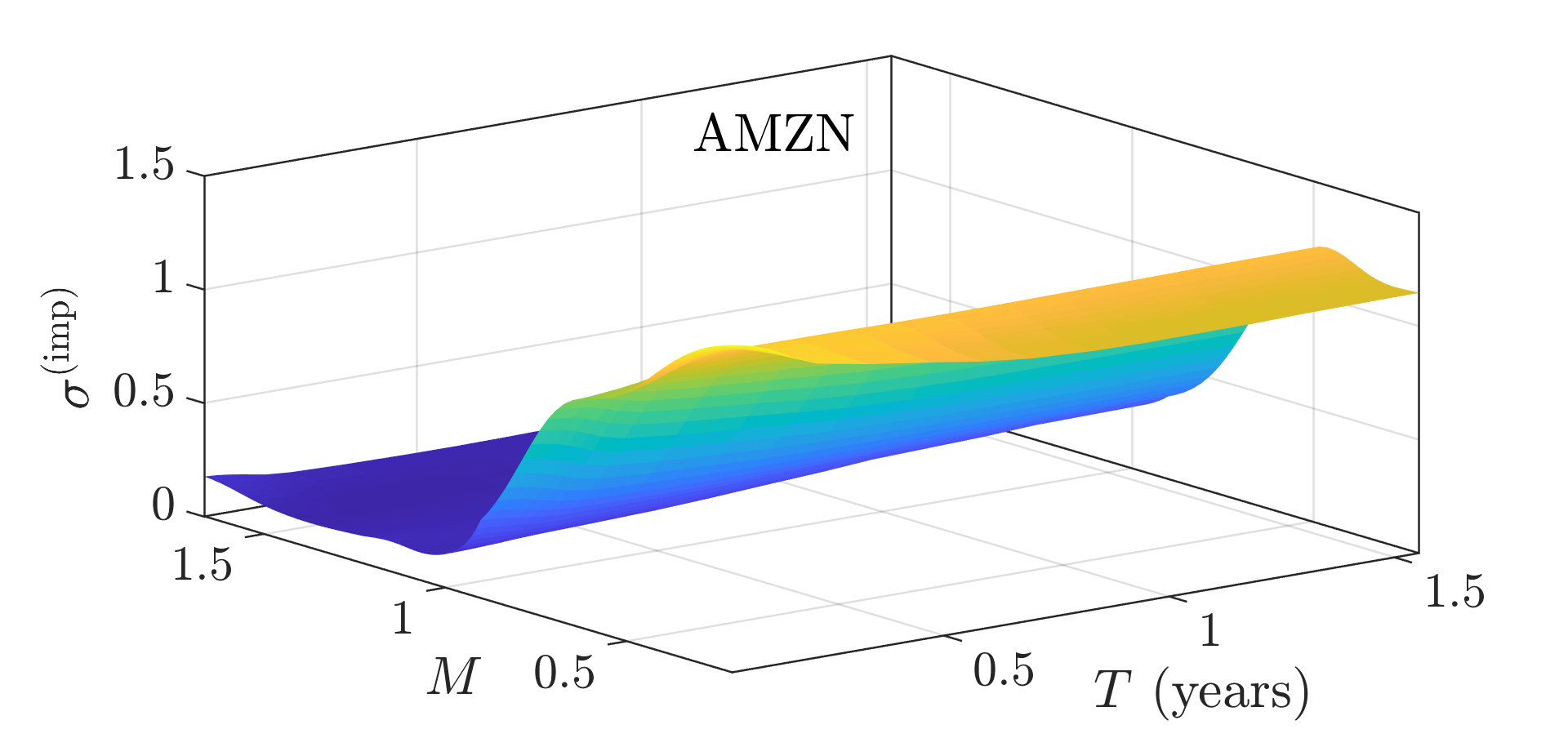}
  \includegraphics[width=0.32\linewidth]{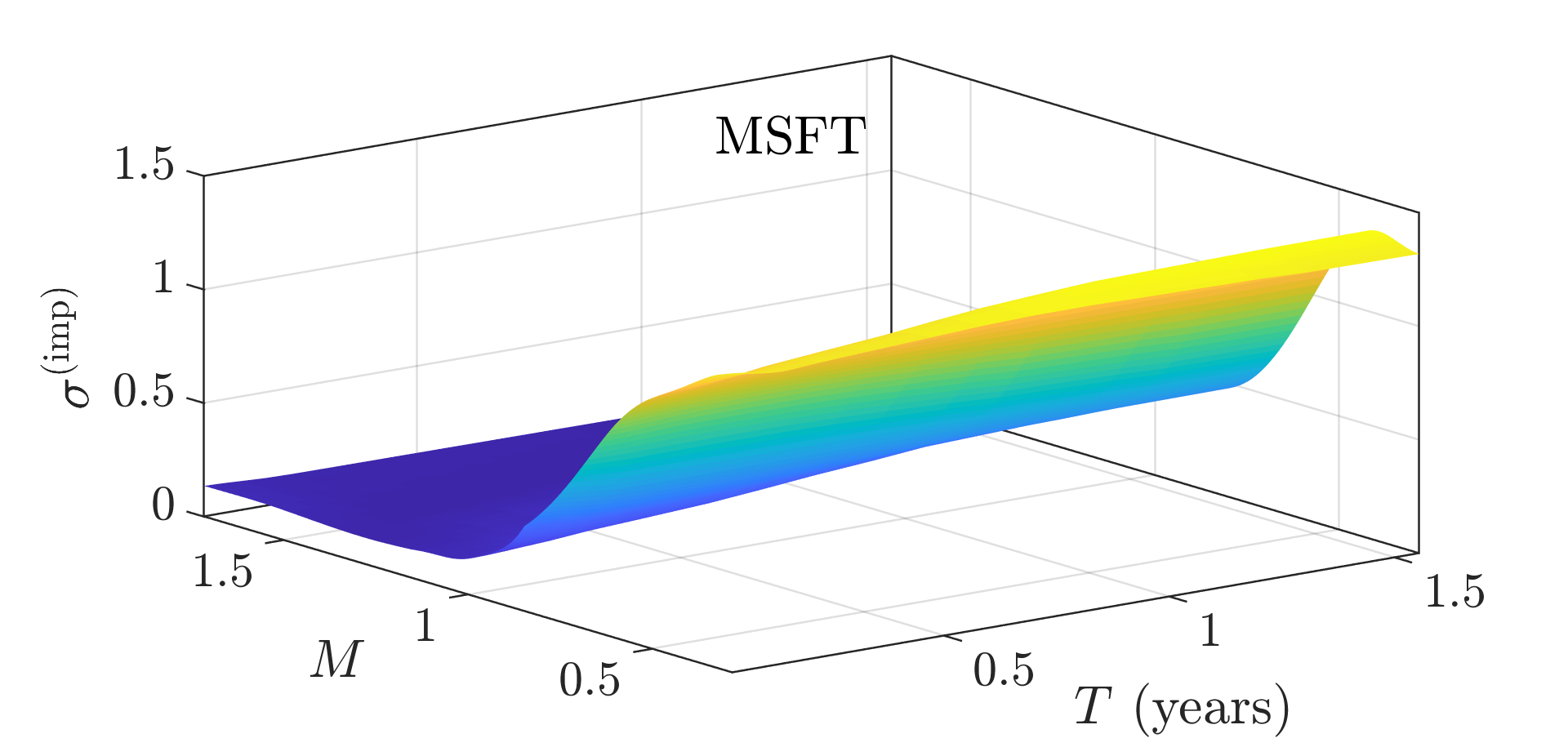}
   \includegraphics[width=0.32\linewidth, height=0.20\textheight]{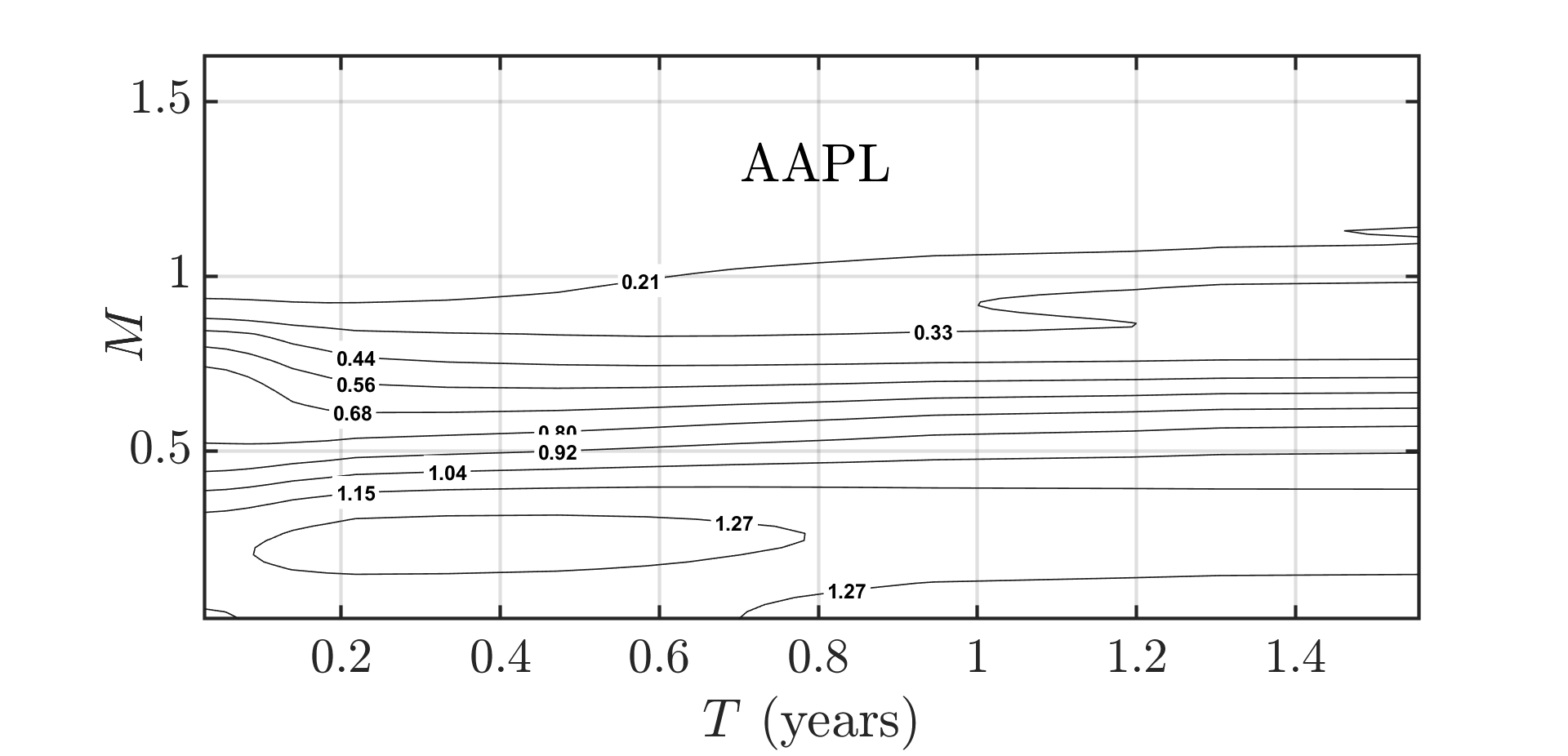}
    \includegraphics[width=0.32\linewidth, height=0.20\textheight]{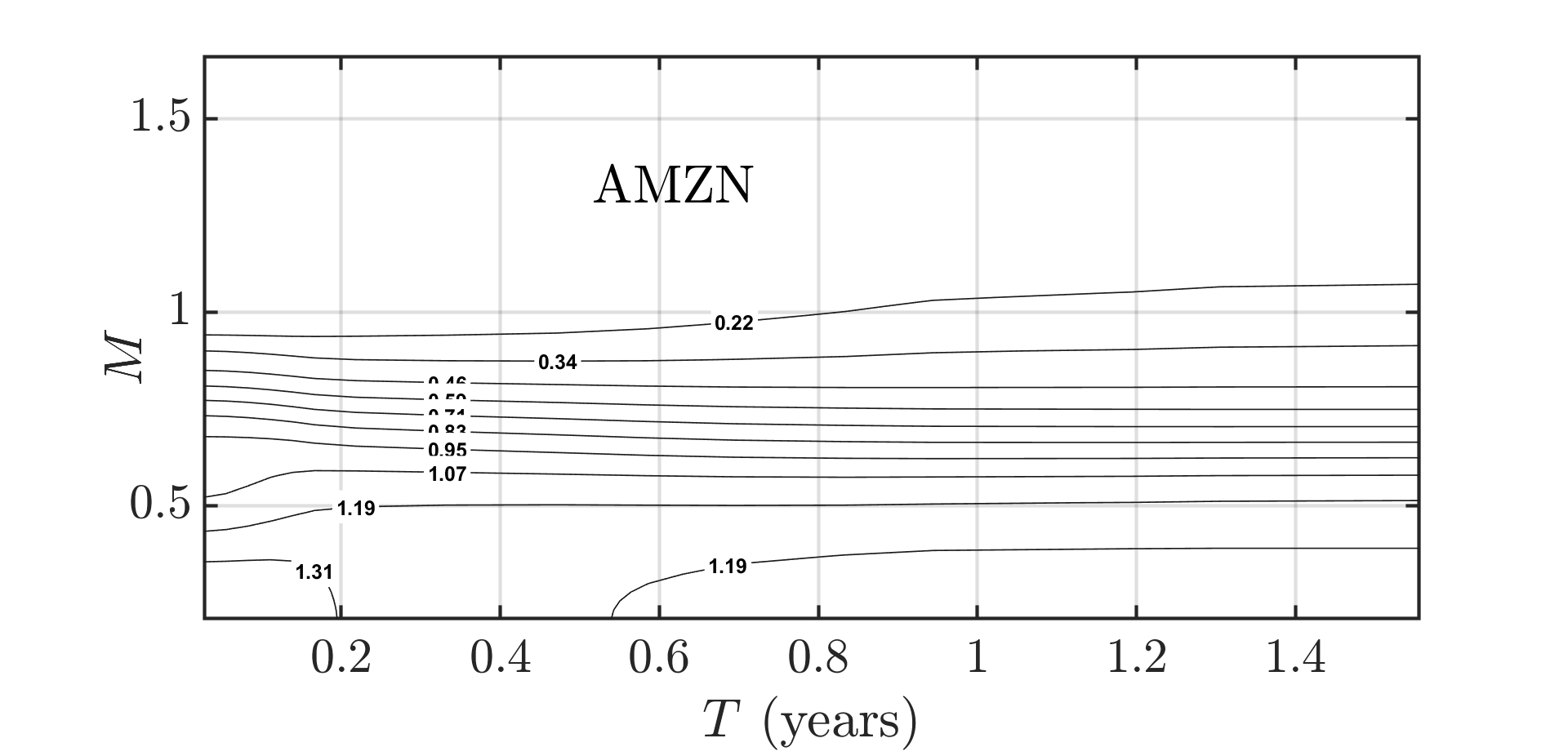}
     \includegraphics[width=0.32\linewidth, height=0.20\textheight]{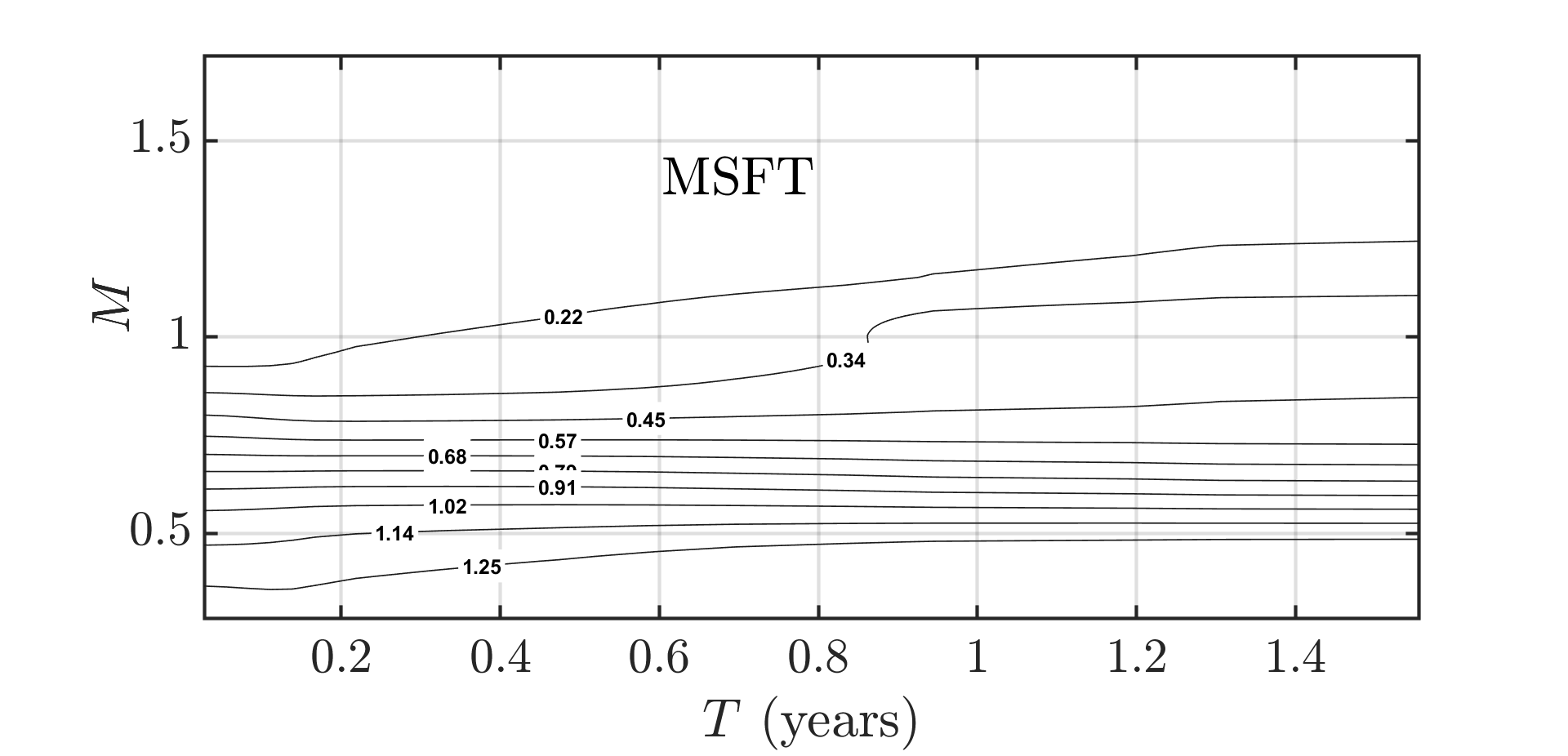}
 \caption{ (top row) State-dependent implied volatility surfaces. (bottom row) Contour plots of the state-dependent implied volatility surface on the \((T,M)\) plane.}
  \label{fig:IV_Surface}
\end{figure}

Figure~\ref{fig:IV_Surface} reports the state-dependent implied volatility surfaces \(\sigma^{\mathrm{(imp)}}(T,M)\) backed out from the calibrated GABM lattice for AAPL, AMZN, and MSFT. Across all three names, implied volatilities are markedly elevated relative to the corresponding historical benchmarks in Table~\ref{tab:params}, especially for low moneyness \(M<1\) where out-of-the-money puts reside, indicating a substantial volatility risk premium embedded in downside protection.

The surface plots reveal the familiar equity smirk: volatility is highest for short-dated, deep out-of-the-money puts and declines as moneyness increases toward and beyond \(M=1\), while exhibiting a mild upward term structure in \(T\) at low \(M\).

The contour panels highlight dense, nearly horizontal level curves in the low-\(M\) region and more widely spaced curves near and above the money, with AAPL and MSFT showing slightly steeper skew than AMZN, thereby confirming that the state-dependent volatility specification captures asset-specific tail risk while remaining consistent with observed market smiles.

\section{Conclusion} \label{sec:conclusion}
This paper develops a rigorous and tractable extension of classical option–pricing by constructing a Geometric Asymmetric Brownian Motion (GABM) and embedding it within the Bachelier--Black--Scholes--Merton paradigm. The key innovation is to introduce asymmetry through the Cherny--Shiryaev--Yor invariance principle (CSYIP), which yields pathwise asymmetric stochastic integrals driven by Brownian motion and their discrete-time analogues. This mechanism generates state-dependent local volatility and skewness while preserving the diffusion backbone and the no-arbitrage discipline of the classical theory.

On the theoretical side, we prove limit theorems that connect asymmetric random walks to their continuous-time counterparts and use these limits to define GABM dynamics. Within this framework we derive closed-form prices for European claims and identify the market price of risk that produces an equivalent martingale measure. We further construct a recombining binomial lattice that converges to the continuous-time limit, providing a numerically stable and easily implementable scheme that extends classical risk-neutral valuation to asymmetric random walks. The analysis is generalized in two directions: first, to a market with two perfectly correlated GABMs and a money-market account, and second, to a setting without a riskless asset, where pricing is obtained from replication in a two-asset market and results in a nonlinear valuation PDE. These extensions clarify how asymmetry propagates through multi-asset hedging, the structure of Greeks, and the conditions for completeness.

On the empirical side, we calibrate the model to equity options on AAPL, AMZN, and MSFT. The inferred parameters reproduce salient features of option data-persistent left skew, heavy tails, and maturity-dependent curvature--within an internally consistent diffusion model. The normal-inverse-Gaussian smoothing of the state-dependent volatility function $h(.)$ materially improves the quality and stability of implied-surface fits over symmetric benchmarks, while the binomial implementation yields fast and accurate valuations suitable for large cross sections.
Together, these results demonstrate that introducing pathwise asymmetry through CSYIP offers a parsimonious, theoretically grounded alternative to jump-diffusion or purely distributional Lévy specifications.

Several avenues for future work follow naturally from our construction. First, the invariance-based methodology can be coupled with time changes and subordinators to capture activity clustering without abandoning the Brownian core; Section~\ref{sec:2} already establishes the needed joint convergence with local time. Second, a multifactor GABM with correlated asymmetric components can address term-structure co-movements across strikes and maturities and provide a platform for stochastic interest rates. Third, path-dependent claims (barriers, Asians, lookbacks) merit a dedicated study, as the local-time channel in $h(.)$ interacts nontrivially with hitting times; our convergent lattice gives a direct numerical route. Fourth, rigorous statistical inference-estimation, identifiability, and finite-sample properties of $h(.)$--can be developed using likelihood expansions and filtering, enabling real-time risk management. Finally, it is important to quantify replication error and discretization bias under state-dependent volatility to guide hedging frequency and model-risk capital.

In summary, the paper shows that asymmetric behavior observed in returns and implied volatilities can be modeled inside a diffusion setting by altering the pathwise structure rather than replacing Brownian motion. The resulting theory preserves transparency and hedging intuition, admits efficient algorithms, and matches the empirical stylized facts. We view GABM as a foundation for a broader rational-finance in which asymmetry, state-dependence, and market microstructure are incorporated via invariance principles while retaining the coherence of continuous-time arbitrage pricing.


\begin{appendices} 
\section [\appendixname~\thesection]{Proof of Lemma~\ref{lemma:3.1}} \label{appendix:Lemma_3.1}

\begin{proof}

 From Eq.\eqref{eq:3.1} $S_t^{(h)} = S_0 \exp\left( \nu t + \sigma B_t + \gamma C_t \right) = f^{(S)}(t, B_t, C_t) $. Then for $t \geq 0, x,y \in \mathbb R, f^{S^{(h)}}(t,x,y) = S_0 \exp\left( \nu t + \sigma x + \gamma y \right)$ yields, 

\begin{equation*}
\begin{array}{rl}
\begin{aligned}
\frac{\partial f^{S^{(h)}}(t,x,y)}{\partial t} &= \nu  f^{S^{(h)}}(t,x,y), \\
\frac{\partial f^{S^{(h)}}(t,x,y)}{\partial y} &= \gamma  f^{S^{(h)}}(t,x,y),
\end{aligned}
& \hspace{1cm}
\begin{aligned}
\frac{\partial f^{S^{(h)}}(t,x,y)}{\partial x} &= \sigma  f^{S^{(h)}}(t,x,y), \\
\frac{\partial^2 f^{S^{(h)}}(t,x,y)}{\partial x^2} &= \sigma^2  f^{S^{(h)}}(t,x,y),
\end{aligned}
\\[1em] 
\begin{aligned}
\frac{\partial^2 f^{S^{(h)}}(t,x,y)}{\partial y^2} &= \gamma^2  f^{S^{(h)}}(t,x,y),
\end{aligned}
& \hspace{1cm}
\begin{aligned}
\frac{\partial^2 f^{S^{(h)}}(t,x,y)}{\partial x \partial y} &= \sigma \gamma  f^{S^{(h)}}(t,x,y).
\end{aligned}
\end{array}
\end{equation*}

Thus,
\begin{equation*}
\begin{split}
    dS_t &= df(t, B_t, C(t)) = \frac{\partial f(t, B_t, C(t))}{\partial t} dt  + \frac{\partial f(t, B_t, C(t))}{\partial x} dB(t) + \frac{\partial f(t, B_t, C(t))}{\partial y} dC(t)\\
   & + \frac{1}{2} \frac{\partial^2 f(t, B_t, C(t))}{\partial x^2} (dB(t))^2 + \frac{\partial^2 f(t, B_t, C(t))}{\partial x \partial y} dB(t) dC(t) + 
    \frac{1}{2} \frac{\partial^2 f(t, B_t, C(t))}{\partial y^2} (dC(t))^2,
    \end{split}
\end{equation*}
\begin{equation*}
\begin{split}
    \text{where} \quad dC(t) = d\left(\int_0^t h(B_s) dB_s \right) = h(B_t) dB(t), \quad (dB(t))^2 = dt, \\
    \quad dB(t) dC(t) = dB(t) h(B_t) dB(t) = h(B_t) dt, \quad (dC(t))^2 = (h(B_t) dB(t))^2 = h(B_t)^2 dt.
  \end{split}
\end{equation*}
\begin{equation*}
\begin{split}
    \text{This yields,} \quad dS_t^{(h)} &= \nu S_t^{(h)} dt + \sigma S_t^{(h)} dB(t) + \gamma S_t^{(h)} dC(t) + \frac{1}{2} \sigma^2 S_t^{(h)} (dB(t))^2 \\
    &+ \sigma \gamma S_t^{(h)} dB(t) dC(t) + \frac{1}{2} \gamma^2 S_t^{(h)} (dC(t))^2\\
    &= \nu S_t^{(h)} dt + \sigma S_t^{(h)} dB(t) + \gamma S_t^{(h)} h(B_t) dB(t) + \frac{1}{2} \sigma^2 S_t^{(h)} dt\\ 
&+\sigma \gamma S_t^{(h)} h(B_t) dt + \frac{1}{2} \gamma^2 S_t^{(h)} h(B_t)^2 dt\\
   & = \left( \nu + \frac{1}{2} \sigma^2 + \sigma \gamma h(B_t) + \frac{1}{2} \gamma^2 h(B_t)^2 \right) S_t^{(h)} dt + \left( \sigma + \gamma h(B_t) \right) S_t^{(h)} dB(t)\\
   & = S_t^{(h)}\left[\left( \nu + \frac{1}{2} \sigma^2 + \sigma \gamma h(B_t) + \frac{1}{2} \gamma^2 h(B_t)^2 \right) dt + \left( \sigma + \gamma h(B_t) \right) dB(t) \right ],\\  & \text{where}  \quad   t \geq 0, S_0 >0.
    \end{split}
\end{equation*}
\end{proof}

\section [\appendixname~\thesection]{Equivalent Martingale measure and dynamics under \texorpdfstring{$\mathbb{Q}$}{Q}} \label{appendix:EMM}

To find the unique equivalent martingale measure $\mathbb{Q}$, define the discounted stock price:
\[
Z_t = \frac{S_t^{(h)}}{\beta_t}, \quad \text{where } \beta_t = \beta_0 e^{rt}.
\]

From Eq.~\eqref{eq:3.3}, and using $d\beta_t = r \beta_t dt$, we get:
\[
dZ_t = Z_t \left[ (\nu - r + \tfrac{1}{2} \sigma^2 + \sigma \gamma h(B_t) + \tfrac{1}{2} \gamma^2 h(B_t)^2) dt + (\sigma + \gamma h(B_t)) dB_t \right].
\]

We now perform a Girsanov change of measure with Radon-Nikodym derivative:
\[
\frac{d\mathbb{Q}}{d\mathbb{P}} \bigg|_{\mathcal{F}_t} = \exp\left( -\int_0^t \theta_s \, dB_s - \frac{1}{2} \int_0^t \theta_s^2 \, ds \right),
\]
where
\[
\theta_t = \frac{ \nu - r + \tfrac{1}{2} \sigma^2 + \sigma \gamma h(B_t) + \tfrac{1}{2} \gamma^2 h(B_t)^2 }{ \sigma + \gamma h(B_t) }.
\]

Then the process
$B_t^{(\mathbb{Q})} := B_t + \int_0^t \theta_s ds$
is a $\mathbb{Q}$-Brownian motion, and under $\mathbb{Q}$,
$dS_t^{(h)} = r S_t^{(h)} dt + (\sigma + \gamma h(B_t)) S_t^{(h)} dB_t^{(\mathbb{Q})}.$

\section [\appendixname~\thesection] {Proof of  Lemma~\ref{lemma:3.2}} \label{appendix:Lemma_3.2}

\begin{proof}

 By Eq.\eqref{eq:3.5}, we have
\begin{equation*}    
      dS_t^{(h)} = \mu_t^{(h)} dt + \sigma_t^{(h)} dB_t,
\end{equation*}

where \( \mu_t^{(h)} = \left( \nu + \frac{1}{2} \sigma^2 + \sigma \gamma h(B_t) + \frac{1}{2} \gamma^2 (h(B_t))^2 \right) S_t^{(h)} \) and \( \sigma_t^{(h)} = \left( \sigma + \gamma h(B_t) \right) S_t^{(h)} \). Then, by the It\^o formula \footnote{See Section 5D of \cite{duffie_2001}.},
\begin{equation*}  \label{eq: 3.11}
    \begin{split}
       df_t & = df(S_t^{(h)}, t)\\
       & = \left[ \frac{\partial f(S_t^{(h)}, t)}{\partial t} + \frac{\partial f(S_t^{(h)}, t)}{\partial x} \mu_t^{(h)} + \frac{1}{2} \frac{\partial^2 f(S_t^{(h)}, t)}{\partial x^2} ({\sigma_t^{(h)}})^2 \right] dt + \frac{\partial f(S_t^{(h)}, t)}{\partial x} \sigma_t^{(h)} dB_t\\
        &= \left[ \frac{\partial f(S_t^{(h)}, t)}{\partial t} + \frac{\partial f(S_t^{(h)}, t)}{\partial x} \left( \nu + \frac{1}{2} \sigma^2 + \sigma \gamma h(B_t) + \frac{1}{2} \gamma^2 (h(B_t))^2 \right) S_t^{(h)} \right. \\
       & \quad  +\left.  \frac{1}{2} \frac{\partial^2 f(S_t^{(h)}, t)}{\partial x^2} \left( \sigma + \gamma h(B_t) \right)^2 ({S_t^{(h)}})^2   \right]dt
        + \frac{\partial f(S_t^{(h)}, t)}{\partial x} \left( \sigma + \gamma h(B_t) \right) S_t^{(h)} dB_t.
    \end{split}
\end{equation*}
\end{proof}

\section [\appendixname~\thesection]{Parameters on Lemma ~\ref{Lemma:4.1}} \label{Appendix:para_4.1}

In Lemma~\ref{Lemma:4.1}, we stated the dynamics of the option price \( f_t = f(\mS_{1,t}^{(h)}, \mS_{2,t}^{(h)}, \beta_t, t) \) as:

\[
df_t = \mu^{(f)}(t) \, dt + \sigma^{(f)}(t) \, dB_t + \gamma^{(f)}(t) \, dC_t.
\]

The expressions for the coefficients are provided below.

\vspace{1em}
\noindent\textbf{Drift term:}
\begin{equation*} \label{eq:appendix_mu}
\begin{split}
\mu^{(f)}(t) &= \frac{\partial f}{\partial t}
+ \frac{\partial f}{\partial x_1} \mu_1(B_t) \mS_{1,t}^{(h)} 
+ \frac{\partial f}{\partial x_2} \mu_2(B_t) \mS_{2,t}^{(h)} 
+ \frac{\partial f}{\partial y} r_t \beta_t \\
&\quad + \frac{1}{2} \frac{\partial^2 f}{\partial x_1^2} \left( \sigma_1^2 + \gamma_1^2 h^2(B_t) + 2\sigma_1 \gamma_1 h(B_t) \right) (\mS_{1,t}^{(h)})^2 \\
&\quad + \frac{1}{2} \frac{\partial^2 f}{\partial x_2^2} \left( \sigma_2^2 + \gamma_2^2 h^2(B_t) + 2\sigma_2 \gamma_2 h(B_t) \right) (\mS_{2,t}^{(h)})^2 \\
&\quad + \frac{\partial^2 f}{\partial x_1 \partial x_2} \left( \sigma_1 \sigma_2 + (\sigma_1 \gamma_2 + \gamma_1 \sigma_2) h(B_t) + \gamma_1 \gamma_2 h^2(B_t) \right) \mS_{1,t}^{(h)} \mS_{2,t}^{(h)}.
\end{split}
\end{equation*}

\vspace{1em}
\noindent\textbf{Brownian diffusion term:}
\begin{equation*} \label{eq:appendix_sigma}
\sigma^{(f)}(t) = \frac{\partial f}{\partial x_1} \sigma_1 \mS_{1,t}^{(h)} + \frac{\partial f}{\partial x_2} \sigma_2 \mS_{2,t}^{(h)}.
\end{equation*}

\vspace{1em}
\noindent\textbf{Functional volatility term:}
\begin{equation*} \label{eq:appendix_gamma}
\gamma^{(f)}(t) = \frac{\partial f}{\partial x_1} \gamma_1 \mS_{1,t}^{(h)} + \frac{\partial f}{\partial x_2} \gamma_2 \mS_{2,t}^{(h)}.
\end{equation*}

\textbf{Time-dependent riskless rate}\label{appendix:time_dependent_rt}

In the derivation of the alternative PDE \eqref{eq:4.10}, the time-dependent interest rate \( r_t \) appears as a result of eliminating the explicitly traded riskless asset \( \mB \) from the market model. The expression for \( r_t \) is given by:

\begin{equation} \label{eq: 4.26_appendix}
    r_t = r_t(x_1, x_2, t) = \frac{\mu_{1,h}(x_i, t) \sigma_{2,h}(x_i, t) - \mu_{2,h}(x_i, t) \sigma_{1,h}(x_i, t)}{\sigma_{2,h}(x_i, t) - \sigma_{1,h}(x_i, t)},
\end{equation}

where the drift and volatility functions \( \mu_{i,h} \) and \( \sigma_{i,h} \) (for \( i = 1,2 \)) are defined as follows:

\begin{equation} \label{eq: 4.27_appendix}
   \mu_{i,h}(x_i, t) = \nu_i + \frac{1}{2} \sigma_i^2 + \sigma_i \gamma_i h(b(t)) + \frac{1}{2} \gamma_i^2 \left(h(b(t))\right)^2,
\end{equation}

\begin{equation} \label{eq: 4.28_appendix}
   \sigma_{i,h}(x_i, t) = \left( \sigma_i + \gamma_i h(b(t)) \right) x_i.
\end{equation}

Here, the function \( h(\cdot) \) denotes the generalized amplitude modulation from the GABM framework, and \( b(t) \) denotes the path of the Brownian motion \( B_t \), i.e., \( B_t = b(t) \) in this deterministic representation for PDE derivation.

Note that \( r_t \) depends on the asset-specific model parameters \( \nu_i, \sigma_i, \gamma_i \), the market path via \( b(t) \), and the current asset prices \( x_1, x_2 \). It captures the effective interest rate implicitly defined by the correlation between the two risky assets in the absence of an independently traded risk-free asset.

\section [\appendixname~\thesection]{Derivation details for section \ref{sec:5}} \label{appendix:5}

We provide the explicit derivation of Lemma~\ref{Lemma:5.1} and the matching of coefficients leading to the PDE~\eqref{eq:5.5}.
\vspace{0.2cm}

\textbf{Application of It\^o's formula}

For $f_t = f(S_{1,t}^{(h)},S_{2,t}^{(h)},t)$, It\^o’s formula yields
\begin{equation} \label{eq:ito_appendix}
\begin{aligned}
df_t &= \frac{\partial f}{\partial t}\, dt 
+ \frac{\partial f}{\partial x_1}\, dS_{1,t}^{(h)} 
+ \frac{\partial f}{\partial x_2}\, dS_{2,t}^{(h)} \\
&\quad + \tfrac{1}{2} \frac{\partial^2 f}{\partial x_1^2}\, d[S_1^{(h)}]_t
+ \tfrac{1}{2} \frac{\partial^2 f}{\partial x_2^2}\, d[S_2^{(h)}]_t
+ \frac{\partial^2 f}{\partial x_1 \partial x_2}\, d[S_1^{(h)}, S_2^{(h)}]_t.
\end{aligned}
\end{equation}

Since
\[
dS_{i,t}^{(h)} = \mu_i(B_t) S_{i,t}^{(h)}\,dt + \sigma_i S_{i,t}^{(h)}\,dB_t + \gamma_i S_{i,t}^{(h)}\,dC_t,
\]
the quadratic variations are
\[
d[S_i^{(h)}]_t = \big(\sigma_i^2 + \gamma_i^2 (h(B_t))^2 + 2\sigma_i\gamma_i h(B_t)\big)(S_{i,t}^{(h)})^2 \, dt,
\]
\[
d[S_1^{(h)},S_2^{(h)}]_t = \big(\sigma_1\sigma_2 + \sigma_1\gamma_2 h(B_t) + \sigma_2\gamma_1 h(B_t) + \gamma_1\gamma_2 (h(B_t))^2\big) S_{1,t}^{(h)} S_{2,t}^{(h)}\, dt.
\]

Substituting these into \eqref{eq:ito_appendix} yields the coefficients described as in Lemma ~\ref{Lemma:4.1}.  

\vspace{0.2cm}
\textbf{Matching of the self-financing portfolio}

The self-financing condition \eqref{eq:5.3} requires

$df_t = a_{1,t}\, dS_{1,t}^{(h)} + a_{2,t}\, dS_{2,t}^{(h)}.$
Comparing coefficients of $dB_t$ and $dC_t$ with \eqref{eq:5.2} gives
$\sigma^{(f)}(t) = a_{1,t}\sigma_1 S_{1,t}^{(h)} + a_{2,t}\sigma_2 S_{2,t}^{(h)},\quad
\gamma^{(f)}(t) = a_{1,t}\gamma_1 S_{1,t}^{(h)} + a_{2,t}\gamma_2 S_{2,t}^{(h)}.$
Therefore,
\[
a_{1,t} = \frac{\partial f}{\partial x_1}, \qquad 
a_{2,t} = \frac{\partial f}{\partial x_2},
\]
which coincides with Eq.~\eqref{eq:5.4}.

Finally, substituting the drift components into \eqref{eq:ito_appendix}, under replication, produces the PDE~\eqref{eq:5.5}.

\section [\appendixname~\thesection]{Technical details on binomial approximation under ARW} \label{appendix:6}

 We provide a detailed formulation and justification of the binomial option pricing model introduced in Section~\ref{sec:6}.

\textbf {Construction of the discrete processes}

For a fixed \( T > 0 \) and \( n \in \mathbb{N} \), define the time-step size \(\Delta t = T / n\). For a piecewise continuous function (PCF) \( h \), define the scaled function
\[
h_T(x) := h\left( \frac{x}{\sqrt{T}} \right).
\]

Set the discrete-time processes on the Skorokhod space \(\mD[0,T]^{\otimes 2}\):
\[
B^{(n)}(t) := \sqrt{\Delta t} \sum_{k=1}^{\lfloor t/\Delta t \rfloor} \xi_k,
\]
\[
C_T^{(n)}(t) := \sqrt{\Delta t} \sum_{k=1}^{\lfloor t/\Delta t \rfloor} h_T\left( \sqrt{\Delta t} \sum_{i=1}^{k-1} \xi_i \right) \xi_k,
\]
with initial conditions \( B^{(n)}(0) = C_T^{(n)}(0) = 0 \), where the \(\xi_i\) are i.i.d. symmetric Bernoulli random variables with
\[
P(\xi_i=1) = P(\xi_i=-1) = \frac{1}{2}.
\]

The bivariate process \((B^{(n)}(t), C_T^{(n)}(t))\) converges weakly to \((B_t, C_{t,T})\) in \(\mD[0,T]^{\otimes 2}\), where \(B_t\) is a standard Brownian motion and
\[
C_{t,T} = \int_0^t h_T(B_v) dB_v,
\]
giving the continuous-time limit of the discrete random walks coupled through \( h_T(\cdot) \).

\vspace{0.2cm}
\textbf {Discrete price process and path-dependent volatility}

Define the discrete price process:
\[
S_t^{(n,h)} = S_0 \exp\left( \nu k \Delta t + \sigma B^{(n)}(t) + \gamma C_T^{(n)}(t) \right),
\quad t \in [k \Delta t, (k+1) \Delta t),
\]
which can be expressed as
\[
S_t^{(n,h)} = S_0 \exp \Bigg( \nu k \Delta t + \sigma \sqrt{\Delta t} \sum_{j=1}^k \xi_j + \gamma \sqrt{\Delta t} \sum_{j=1}^k h_T\left( \sqrt{\Delta t} \sum_{i=1}^{j-1} \xi_i \right) \xi_j \Bigg).
\]

The path-dependent volatility is defined by
\begin{equation} \label{eq:path_vol_appendix}
\eta_{k,\Delta t} := \sigma + \gamma h_T\left( \sqrt{\Delta t} \sum_{i=1}^k \xi_i \right).
\end{equation}

\section [\appendixname~\thesection]{European put price surfaces and contour maps} \label{PutSurfaces:ALL}

\begin{figure}[htbp]
  \centering
  \includegraphics[width=0.32\linewidth]{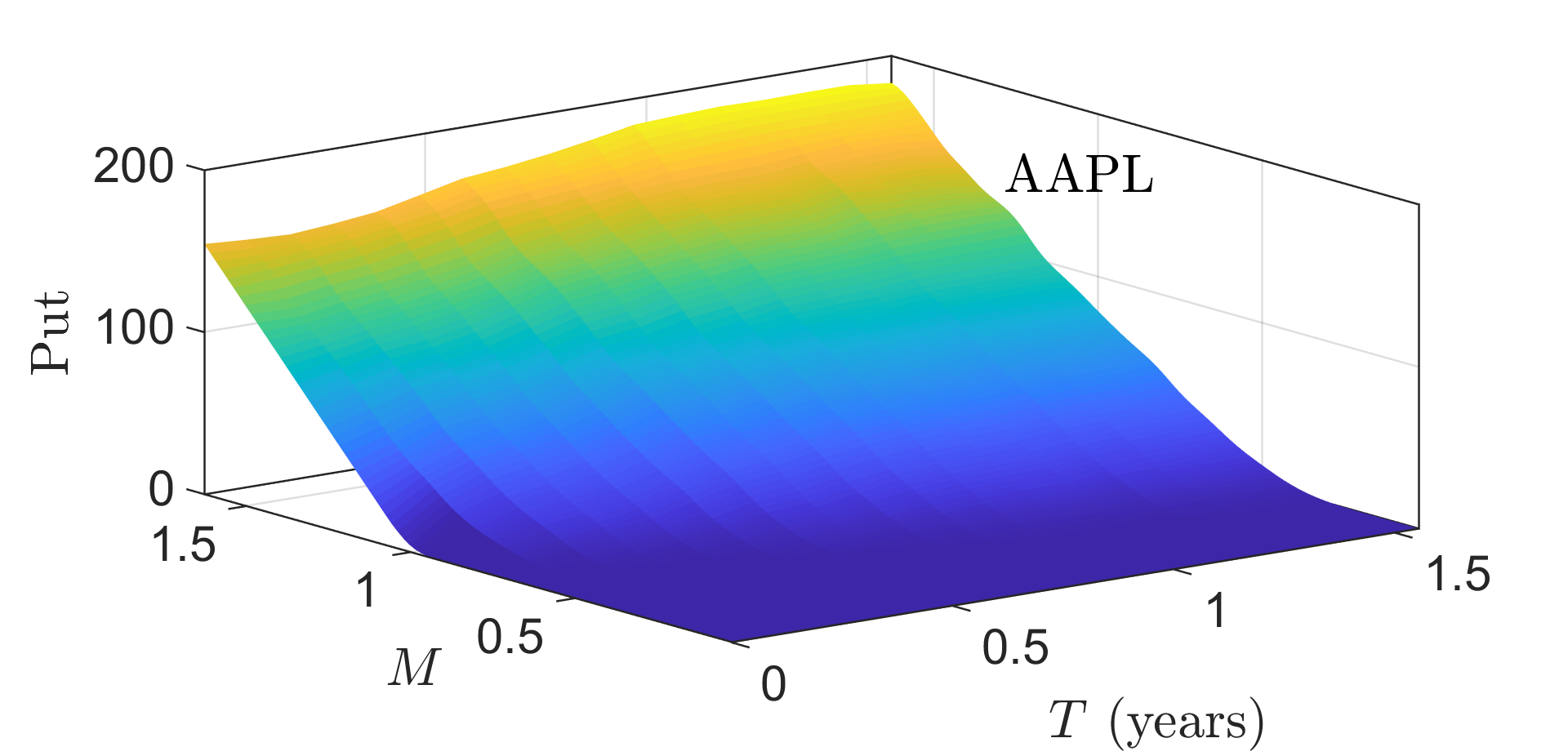}
  \includegraphics[width=0.32\linewidth]{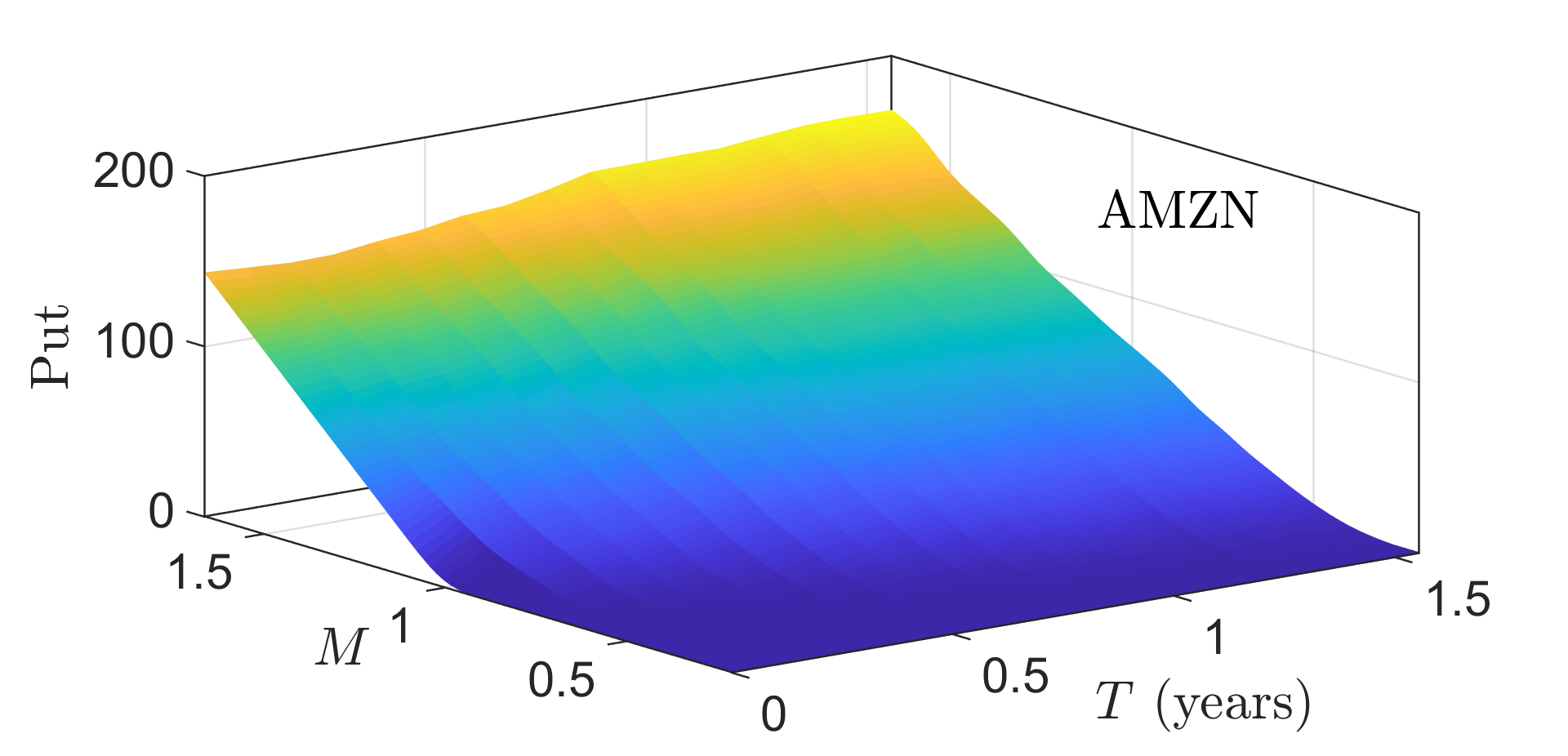}
   \includegraphics[width=0.32\linewidth]
  {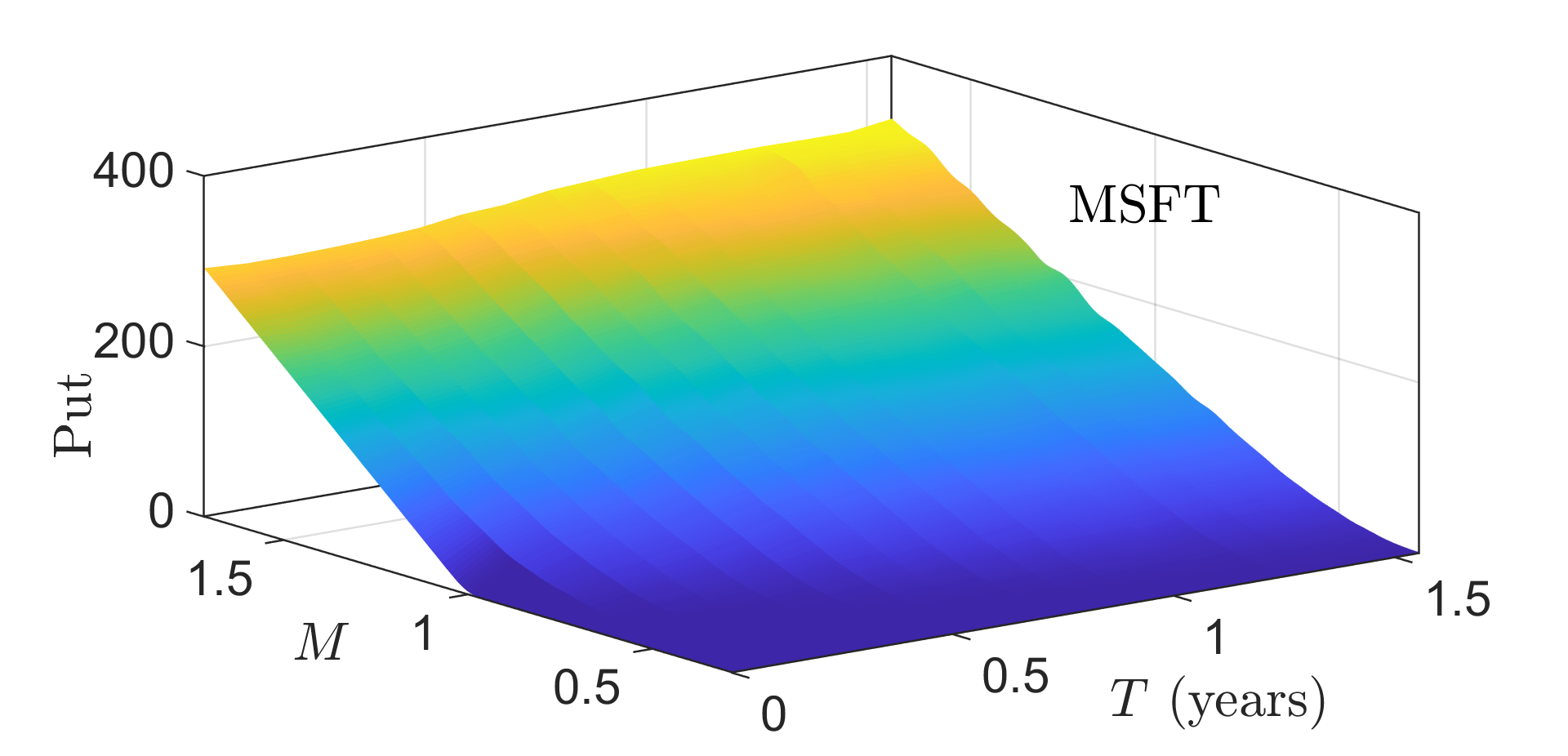}
   \includegraphics[width=0.32\linewidth, height=0.20\textheight]
  {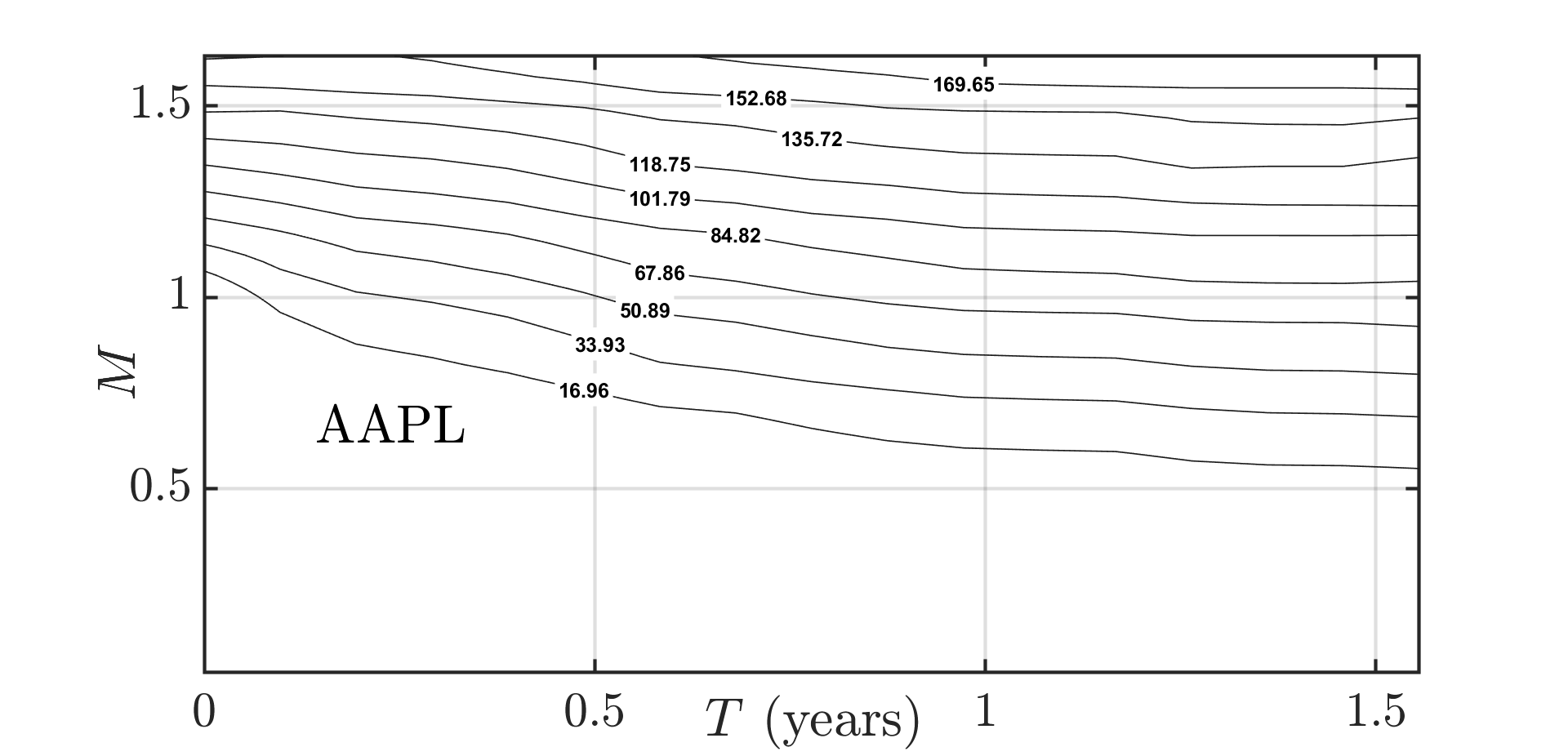}
   \includegraphics[width=0.32\linewidth,height=0.20\textheight]
  {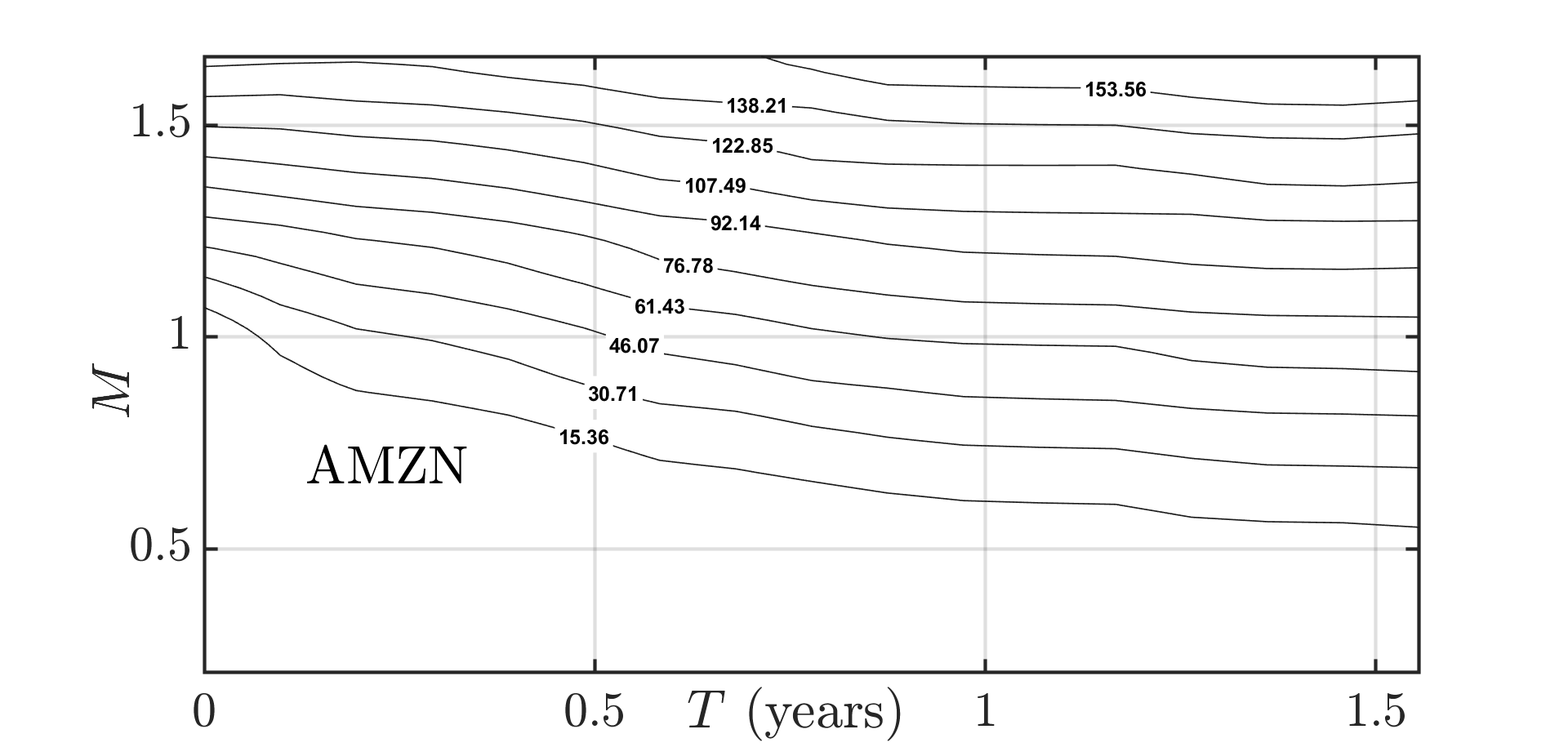}
   \includegraphics[width=0.32\linewidth,height=0.20\textheight]
  {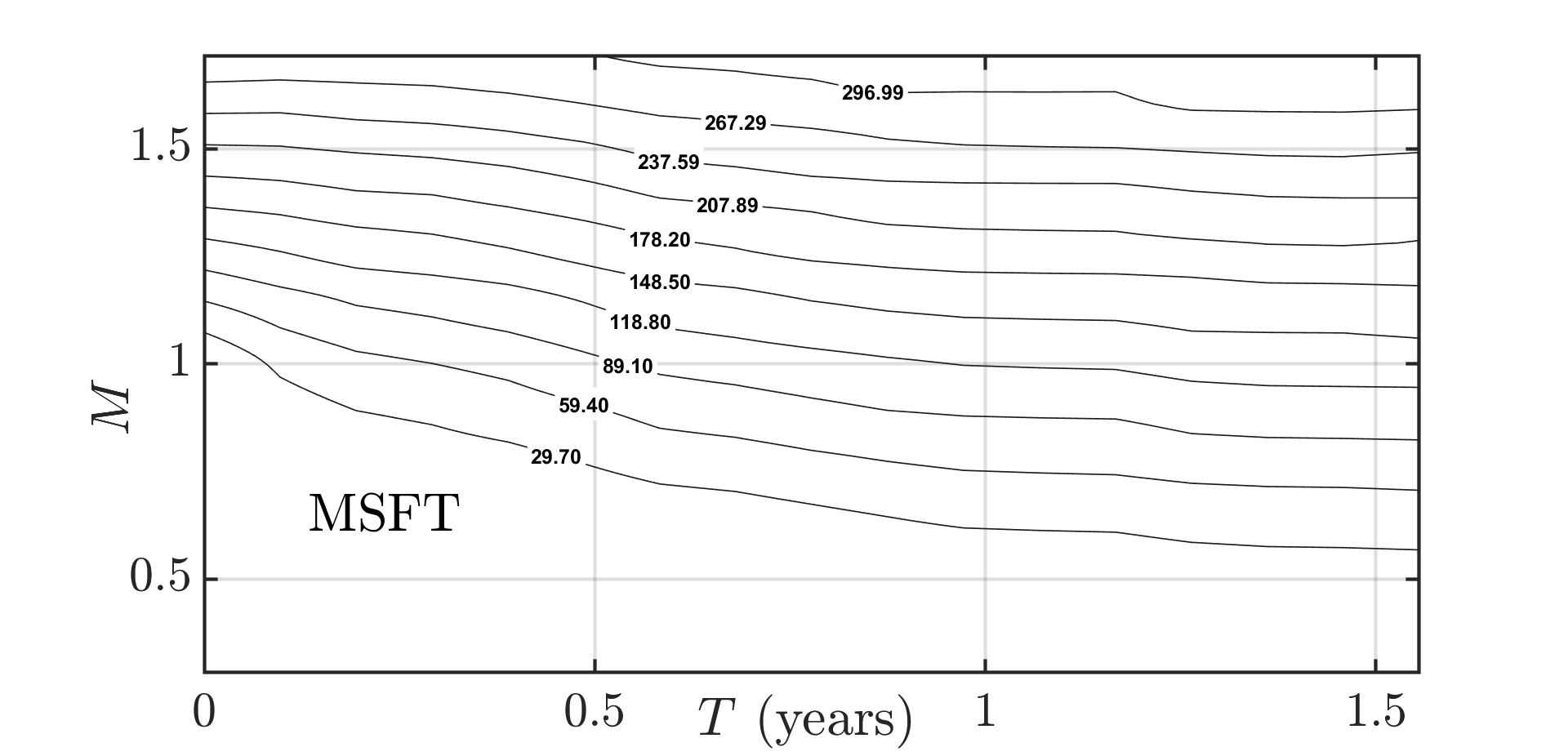}
 \caption{ (top row) Computed put option prices. (bottom row) Surface contours of the computed put price surface projected on the $(T, M)$ plane. }
\label{fig:Put-surface}
\end{figure}

\end{appendices}
\clearpage
\bibliographystyle{agsm}
\bibliography{Asymmetric_R_Walk}

@inproceedings{Bachelier_1900,
  title={Th{\'e}orie de la sp{\'e}culation},
  author={Bachelier, Louis},
  booktitle={Annales scientifiques de l'{\'E}cole normale sup{\'e}rieure},
  volume={17},
  pages={21--86},
  year={1900}
}

@article{Bakshi_2003,
  title={Stock return characteristics, skew laws, and the differential pricing of individual equity options},
  author={Bakshi, Gurdip and Kapadia, Nikunj and Madan, Dilip},
  journal={The Review of Financial Studies},
  volume={16},
  number={1},
  pages={101--143},
  year={2003},
  publisher={Oxford University Press}
}

@article{Barberis_1998,
  title={A model of investor sentiment},
  author={Barberis, Nicholas and Shleifer, Andrei and Vishny, Robert},
  journal={Journal of financial economics},
  volume={49},
  number={3},
  pages={307--343},
  year={1998},
  publisher={Elsevier}
}

@article{Black_1973,
  title={The pricing of options and corporate liabilities},
  author={Black, Fischer and Scholes, Myron},
  journal={Journal of political economy},
  volume={81},
  number={3},
  pages={637--654},
  year={1973},
  publisher={The University of Chicago Press}
}

@article{Black_1986,
  title={Noise},
  author={Black, Fisher},
  journal={The Journal of Finance},
  volume={XLI},
  number={3},
  pages={529--543},
  year={1986},
  publisher={Elsevier}
}

@article{Carr_1999,
  title={Option valuation using the fast Fourier transform},
  author={Carr, Peter and Madan, Dilip},
  journal={Journal of computational finance},
  volume={2},
  number={4},
  pages={61--73},
  year={1999}
}

@article{Cherny_2003,
  title={Limit behavior of the" horizontal-vertical" random walk and some extensions of the Donsker-Prokhorov invariance principle},
  author={Cherny, Aleksander S and Shiryaev, Albert N and Yor, Marc},
  journal={Theory of Probability \& Its Applications},
  volume={47},
  number={3},
  pages={377--394},
  year={2003},
  publisher={SIAM}
}

@article{Christoffersen_2013,
  title={Capturing option anomalies with a variance-dependent pricing kernel},
  author={Christoffersen, Peter and Heston, Steven and Jacobs, Kris},
  journal={The Review of Financial Studies},
  volume={26},
  number={8},
  pages={1963--2006},
  year={2013},
  publisher={Oxford University Press}
}

@book{Chung_1990,
  title={Introduction to stochastic integration},
  author={Chung, Kai Lai and Williams, Ruth J},
  volume={2},
  year={1990},
  publisher={Springer},
  address={New York}
}

@article{Corns_2007,
  title={Skew {B}rownian motion and pricing European options},
  author={Corns, T Richard A and Satchell, Stephen E},
  journal={The European Journal of Finance},
  volume={13},
  number={6},
  pages={523--544},
  year={2007},
  publisher={Taylor \& Francis}
}

@article{Cont_2001,
  title={Empirical properties of asset returns: stylized facts and statistical issues},
  author={Cont, Rama},
  journal={Quantitative finance},
  volume={1},
  number={2},
  pages={223},
  year={2001},
  publisher={IOP Publishing}
}

@book{Duffie_2001,
  title={Dynamic asset pricing theory},
  author={Duffie, Darrell},
  year={2001},
  publisher={Princeton University Press},
  address={Princeton, New Jersey}
}

@article{Eberlein_1995,
  title={Hyperbolic distributions in finance},
  author={Eberlein, Ernst and Keller, Ulrich},
  journal={Bernoulli},
  pages={281--299},
  year={1995},
  publisher={JSTOR}
}

@article{Gnawali_2025,
  title={Hedging via Perpetual Derivatives: Trinomial Option Pricing and Implied Parameter Surface Analysis},
  author={Gnawali, Jagdish and Lindquist, W Brent and Rachev, Svetlozar T},
  journal={Journal of Risk and Financial Management},
  volume={18},
  number={4},
  pages={192},
  year={2025},
  publisher={MDPI}
}

@article{Heston_1993,
  title={A closed-form solution for options with stochastic volatility with applications to bond and currency options},
  author={Heston, Steven L},
  journal={The review of financial studies},
  volume={6},
  number={2},
  pages={327--343},
  year={1993},
  publisher={Oxford University Press}
}

@article{Hu_2020a,
  title={Option pricing in markets with informed traders},
  author={Hu, Yuan and Shirvani, Abootaleb and Stoyanov, Stoyan and Kim, Young Shin and Fabozzi, Frank J and Rachev, Svetlozar T},
  journal={International Journal of Theoretical and Applied Finance},
  volume={23},
  number={06},
  pages={2050037},
  year={2020},
  publisher={World Scientific}
}

@article{Hu_2020b,
  title={Option pricing incorporating factor dynamics in complete markets},
  author={Hu, Yuan and Shirvani, Abootaleb and Lindquist, W Brent and Fabozzi, Frank J and Rachev, Svetlozar T},
  journal={Journal of Risk and Financial Management},
  volume={13},
  number={12},
  pages={321},
  year={2020},
  publisher={MDPI}
}

@article{Hu_2022,
  title={Market complete option valuation using a {J}arrow-{R}udd pricing tree with skewness and kurtosis},
  author={Hu, Yuan and Lindquist, W Brent and Rachev, Svetlozar T and Shirvani, Abootaleb and Fabozzi, Frank J},
  journal={Journal of Economic Dynamics and Control},
  volume={137},
  pages={104345},
  year={2022},
  publisher={Elsevier}
}

@article{Hu_2024,
  title={Option Pricing Using a Skew Random Walk Binary Tree},
  author={Hu, Yuan and Lindquist, W Brent and Rachev, Svetlozar T and Fabozzi, Frank J},
  journal={Journal of Risk and Financial Management},
  volume={17},
  number={4},
  pages={138},
  year={2024},
  publisher={MDPI}
}

@article{Jackwerth_2000,
  title={Recovering risk aversion from option prices and realized returns},
  author={Jackwerth, Jens Carsten},
  journal={The Review of Financial Studies},
  volume={13},
  number={2},
  pages={433--451},
  year={2000},
  publisher={Oxford University Press}
}

@article{Merton_1973,
  title={An intertemporal capital asset pricing model},
  author={Merton, Robert C},
  journal={Econometrica: Journal of the Econometric Society},
  pages={867--887},
  year={1973},
  publisher={JSTOR}
}

@phdthesis{Prause_1999,
  title     = {The Generalized Hyperbolic Model: Estimation, Financial Derivatives, and Risk Measures},
  author    = {Prause, Karsten},
  year      = {1999},
  school    = {Albert-Ludwigs-Universität Freiburg}
}

@article{Rachev_2017,
  title={Financial markets with no riskless (safe) asset},
  author={Rachev, Svetlozar T and Stoyanov, Stoyan V and Fabozzi, Frank J},
  journal={International Journal of Theoretical and Applied Finance},
  volume={20},
  number={08},
  pages={1-24},
  year={2017},
  publisher={World Scientific}
}

@article{Rosinski_2007,
  title={Tempering stable processes},
  author={Rosi{\'n}ski, Jan},
  journal={Stochastic processes and their applications},
  volume={117},
  number={6},
  pages={677--707},
  year={2007},
  publisher={Elsevier}
}

@book{Shiller_1981,
  title={Do stock prices move too much to be justified by subsequent changes in dividends?},
  author={Shiller, Robert J and others},
  year={1981},
  publisher={National Bureau of Economic Research Cambridge, MA}
}
\end{document}